\documentclass[english,11pt]{article}
\usepackage{verbatim}

\usepackage{subcaption}
\usepackage{tikz}
\usetikzlibrary{decorations.pathmorphing}
\usetikzlibrary{decorations.markings}
\usetikzlibrary{shapes.gates.logic.US,trees,positioning,arrows}

\tikzset{
	arn/.style = {circle, white, draw=black, fill=black, inner sep = 1.5},
	arn_l/.style = {circle, white, draw=black, fill=black, inner sep = 2.2},
	photon/.style={draw=black, very thick, dashed},
	electron/.style={draw=black, very thick, line width=0.08cm},
	tr/.style={buffer gate US,thick,draw,fill=gray!60,rotate=90,	anchor=east,minimum width=2.25cm},
	br/.style={buffer gate US,thick,draw,fill=gray!60,rotate=90,	anchor=east,minimum width=4.5cm},
	brr/.style={buffer gate US,draw,fill=gray!60,rotate=90,	anchor=east,minimum width=4.5cm, opacity = 0.6},
	trr/.style={buffer gate US,thick,draw,fill=gray!60,rotate=90,	anchor=east,minimum width=2.25cm, opacity = 0.6},
	trrr/.style={buffer gate US,draw,fill=white!60,rotate=90,	anchor=east,minimum width=2.25cm, opacity = 0.5}
}
\usepackage[utf8]{inputenc}
\usepackage{amssymb,amsmath,amsthm}
\usepackage{verbatim,hyperref}
\usepackage{dsfont}
\usepackage{algorithm}
\usepackage[noend]{algpseudocode}

\usepackage{amsfonts}

\makeatletter
\def\BState{\State\hskip-\ALG@thistlm}
\makeatother
\usepackage{graphicx}
\usepackage[utf8]{inputenc}
\usepackage[english]{babel}
\usepackage{color}
\usepackage{verbatim}

\usepackage{float}
\usepackage{array}
\usepackage{hyperref}
\usepackage{setspace}
\usepackage[margin=1.1in,footskip=0.25in]{geometry}
\usepackage{bm}

\newtheorem{theorem}{Theorem}
\newtheorem{lemma}{Lemma}

\newtheorem{corollary}{Corollary}

\newtheorem{claim}{Claim}

\newenvironment{proofof}[1]{\noindent{\bf Proof of #1:}}{$\qed$\par}

\theoremstyle{definition}
\newtheorem{definition}{Definition}

\DeclareMathOperator{\poly}{poly}

\newcommand{\supp}{\mathrm{supp}~}
\newcommand{\tfull}{T^{\mathrm{full}}}
\newcommand{\tree}{\mathrm{Tree}}

\bmdefine{\aaa}{a}
\bmdefine{\jj}{j}
\bmdefine{\rr}{r}
\bmdefine{\lv}{l}
\bmdefine{\sv}{s}
\bmdefine{\tv}{t}
\bmdefine{\ff}{f}
\bmdefine{\gg}{g}
\bmdefine{\hh}{h}
\bmdefine{\tt}{t}
\bmdefine{\qq}{q}
\bmdefine{\vv}{v}
\bmdefine{\ww}{w}
\bmdefine{\phib}{\phi}
\newcommand{\matA}{\mathbf{A}}
\newcommand{\shiftset}{\mathcal{A}}
\newcommand{\dom}{\mathrm{Dom}}

\bmdefine{\alphav}{\alpha}
\bmdefine{\betav}{\beta}
\bmdefine{\bv}{B}
\bmdefine{\bb}{b}
\newcommand{\bnextv}{\bm{B}^\mathrm{next}}
\newcommand{\bnext}{B^\mathrm{next}}
\newcommand{\bnextvt}[1]{\bm{B}^{\mathrm{next}, (#1)}}
\newcommand{\bprevv}{\bm{B}^\mathrm{prev}}
\newcommand{\bprev}{B^\mathrm{prev}}
\newcommand{\bprevvt}[1]{\bm{B}^{\mathrm{prev}, (#1)}}

\newcommand{\bbasev}{\bm{B}^\mathrm{base}}
\newcommand{\bbase}{B^\mathrm{base}}
\newcommand{\bbasevt}[1]{\bm{B}^{\mathrm{base}, (#1)}}

\newcommand{\bvec}{\bm{B}}

\newcommand{\wh}{\widehat}
\newcommand{\wt}{\widetilde}
\newcommand{\EE}{\mathcal{E}}
\newcommand{\bxi}{\boldsymbol{\xi}}
\newcommand{\E}{\mathbb{E}}
\newcommand{\unif}{\mathrm{Unif}}

\newcommand{\argmin}{\text{argmin}}

\DeclareMathOperator{\subtree}{\mathrm{FrequencyCone}}

\newcommand{\C}{{\mathbb C}}

{\makeatletter
	\gdef\xxxmark{%
		\expandafter\ifx\csname @mpargs\endcsname\relax 
		\expandafter\ifx\csname @captype\endcsname\relax 
		\marginpar{xxx}
		\else
		xxx 
		\fi
		\else
		xxx 
		\fi}
	\gdef\xxx{\@ifnextchar[\xxx@lab\xxx@nolab}
	\long\gdef\xxx@lab[#1]#2{{\bf [\xxxmark #2 ---{\sc #1}]}}
	\long\gdef\xxx@nolab#1{{\bf [\xxxmark #1]}}
}

\usepackage{enumitem}

\newcommand{\e}{\epsilon}
\renewcommand{\Pr}{\mathrm{Pr}}

\usepackage{thmtools}

\begin{document}

\title{Dimension-independent Sparse Fourier Transform}

\author{Michael Kapralov\\EPFL\\ {michael.kapralov@epfl.ch} \and Ameya Velingker\thanks{This work was completed while the author was a research scientist in the School of Computer and Communication Sciences, EPFL.}\\Google Research\\ {ameyav@google.com} \and Amir Zandieh\\EPFL\\ {amir.zandieh@epfl.ch}}

\maketitle

\begin{abstract}
The Discrete Fourier Transform (DFT) is a fundamental computational primitive, and the fastest known algorithm for computing the DFT is the FFT (Fast Fourier Transform) algorithm. One remarkable feature of FFT is the fact that its runtime depends only on the size $N$ of the input vector, but {\em not on the dimensionality of the input domain}: FFT runs in time $O(N\log N)$ irrespective of whether the DFT in question is on $\mathbb{Z}_N$ or $\mathbb{Z}_n^d$ for some $d>1$, where $N=n^d$.

The state of the art for Sparse FFT, i.e. the problem of computing the DFT of a signal that has at most $k$ nonzeros in Fourier domain, is very different: all current techniques for sublinear time computation of Sparse FFT incur an exponential dependence on the dimension $d$ in the runtime. In this paper we give the first algorithm that computes the DFT of a $k$-sparse signal in time $\poly(k, \log N)$ {\em in any dimension $d$}, avoiding the curse of dimensionality inherent in all previously known techniques.  Our main tool  is a new class of filters that we refer to as {\em adaptive aliasing filters}: these filters allow isolating frequencies of a $k$-Fourier sparse signal using $O(k)$ samples in time domain and $O(k\log N)$ runtime per frequency, in any dimension $d$.

We also investigate natural average case models of the input signal:  {\bf (1)} worst case support in Fourier domain with randomized coefficients and {\bf (2)} random locations in Fourier domain with worst case coefficients. Our techniques lead to an $\widetilde O(k^2)$ time algorithm for the former and an $\widetilde O(k)$ time algorithm for the latter.

\end{abstract}
\setcounter{page}{0}
\newpage
\setcounter{page}{1}

\section{Introduction}

The Discrete Fourier Transform (DFT) is one of the most widely used computational primitives in modern computing, with numerous applications in data analysis, signal processing, and machine learning. The fastest algorithm for computing the DFT is the Fast Fourier Transform (FFT) algorithm of Cooley and Tukey, which has been recognized as one of the 10 most important algorithms of the 20th century~\cite{citeulike:6838680}. The FFT algorithm is very efficient: it computes the Discrete Fourier Transform of a length $N$ complex-valued signal in time $O(N\log N)$. This applies to vectors in any dimension: FFT works in $O(N\log N)$ time irrespective of whether the DFT is on the line, on a $\sqrt{N}\times \sqrt{N}$ grid, or is in fact the Hadamard transform on $\{0, 1\}^d$, with $d=\log_2 N$.

In any applications of the Discrete Fourier Transform, the input signal $x\in \C^N$ often satisfies {\em sparsity} or {\em approximate sparsity} constraints: the Fourier transform $\wh{x}$ of $x$ has a small number of coefficients $k$ or is close to a signal with a small number of coefficients (e.g., this phenomenon is the motivation for compression schemes such as JPEG and MPEG). This has motivated a rich line of work on the {\em Sparse FFT} problem: given access to a signal $x\in \C^N$ in time domain that is sparse in Fourier domain, compute the $k$ nonzero coefficients in {\em sublinear} (i.e., $o(N)$) time. 

Very efficient algorithms for the Sparse FFT problem have been developed in the literature~\cite{GL,KM,Man,GGIMS,AGS,GMS,Iw,Ak,HIKP,HIKP2,LWC,BCGLS,HAKI,pawar2013computing,heidersparse, IKP, IK14a,K16,PZ15,ChenKPS16,Kapralov17}. The state-of-the-art approach, due to \cite{HIKP2}, yields  an $O(k\log N)$ runtime algorithm for the following exact $k$-sparse Fourier transform problem: given access to an input signal of length $N$ whose Fourier transform has at most $k$ nonzeros, output the nonzero coefficients and their values. This highly efficient algorithm comes with a caveat, however: the runtime of $O(k\log N)$ only holds for the Fourier transform on the line, namely, $\mathbb{Z}_N$. The algorithm naturally extends to higher dimensions, namely, $\mathbb{Z}_n^d$, where $N=n^d$, but with an exponential loss in runtime; the runtime becomes $O(k\log^d N)$ as opposed to $O(k\log N)$. Interestingly, the other extreme of $d=\log_2 N$, i.e., the Hadamard transform, has been known to admit an $O(k\log N)$ algorithm since the seminal work of Goldreich and Levin~\cite{GL}. However, all intermediate values of $d$ exhibit a {\em curse of dimensionality}. This  is in sharp contrast with FFT itself, which runs in time $O(N\log N)$, where $N=n^d$ is the length of the input signal, in {\em any dimension $d$}. The focus of our work is to design sublinear time algorithms for Sparse FFT that avoid this curse of dimensionality. Our main point of attention is the Sparse FFT problem: 
\begin{equation}\label{eq:exact-sfft}
\begin{split}
\text{\bf Input:}& \text{~~~access to $x:[n]^d\to \C$,}\\
&\text{~~~integer $k\geq 1$ such that $|\supp \wh{x}|\leq k$}\\
\text{\bf Output:}& \text{~~~nonzero elements of $\wh{x}$ and their coefficients}
\end{split}
\end{equation}
Our main result is the first sublinear algorithm for exact Sparse FFT~\eqref{eq:exact-sfft}, as stated in the following theorem.
\begin{theorem}[Main result, informal version of Theorem~\ref{thm:sfft-worstcase} in Section~\ref{subsec:recovadapt}] \label{thm:main}
	For any integer $n$ that is a power of two and any positive integer $d$, there exists a deterministic algorithm that, given access to a signal $x\in \C^{n^d}$ with $\|{\wh{x}}\|_0 \leq k$, recovers $\wh{x}$ in time $\poly(k, \log N)$.
\end{theorem}

We note that this is the first sublinear time Sparse FFT algorithm that avoids an exponential dependence on the dimension $d$. One should note that the runtime still depends on $d$, since $\log_2 N=d\log_2 n$ is lower bounded by $d$, but this dependence is polynomial as opposed to exponential. 

\subsection{Significance of our results and related work}

\paragraph{Significance of our results.} The state of the art in high dimensional Sparse Fourier Transforms presents an interesting conundrum: algorithms with runtime $O(k\log N)$ are known for $d=1$ (Discrete Fourier Transform on the line, see~\cite{HIKP2}) and $d=\log_2 N$ (the Hadamard transform, see~\cite{GL}), but for all intermediate values of $d$ the runtime scales exponentially in $d$. Given that FFT itself is dimension-insensitive, this strongly suggests that exciting new algorithmic techniques can be developed for the high-dimensional version of the problem. Our paper designs the first approach to high dimensional Sparse FFT that does not suffer from the curse of dimensionality, and naturally leads to several exciting open problems that we hope will spur further progress in this area. 

In addition, we note that rather high-dimensional versions of the Fourier transform arise in applications (e.g., 2D, 3D and 4D-NMR in medical imaging), and designing practical Sparse FFT algorithms for  this regime is an important problem. We hope that new techniques for dimension-independent Sparse FFT will lead to progress in this direction as well.

\paragraph{Sample complexity of high-dimensional Sparse FFT.} We note that, besides runtime, another very important parameter of a Sparse FFT algorithm is {\em sample complexity}, i.e., the number of samples that an algorithm needs to access in time domain in order to compute the top few coefficients of the Fourier transform. The sample complexity of Sparse FFT, unlike runtime, does not suffer from a curse of dimensionality. Indeed, there exist several algorithms with $\widetilde O(N)$ runtime that can recover the top $k$ coefficients of $\wh{x}$ using only $k\poly(\log N)$  accesses in time domain, irrespective of the dimensionality of the problem. This can be achieved, for example, using either results on the restricted isometry property (RIP)~\cite{CTao, RV, Bourgain2014, CGV, haviv2017restricted}, or using the filtering approach developed in the Sparse FFT literature, with $\widetilde O(N)$ decoding time. Thus, the challenge is to achieve sublinear {\em runtime} without an exponential dependence on the dimension.

\bigskip

We now outline existing approaches to Sparse FFT and explain why they fail to scale well in high dimensions:

\paragraph{State-of-the-art approaches to Sparse FFT and their lack of scalability in high dimensions.} The main idea behind many recently developed algorithms for the Sparse FFT problem is the ``hashing'' approach inherited from sparse recovery with arbitrary linear measurements. Given access to a signal $x:[n]^d\to \C$, one designs linear measurements of $x$ that allow one to ``hash'' the nonzero positions of $\wh{x}$ into a number of ``buckets.'' The number of buckets $B=b^d$ is chosen to be a constant factor larger than the sparsity $k$ to ensure that a large constant fraction of the nonzero positions of $\wh{x}$ are isolated in their buckets. Every isolated element can be recovered and subtracted from $x$ for future iterations of the same hashing scheme, thereby ensuring convergence. The idea of hashing is implemented via filtering: one designs a filter $G:[n]^d\to \C$ such that $\wh{G}$ approximates a ``bucket,'' i.e., $\wh{G}$ is close to $1$ on an $\ell_\infty$ ball of side length $\approx (N/B)^{1/d}=n/b$ in dimension $d$. The content of the $\jj$-th `bucket', for $\jj\in [\bv]$, is then
\begin{equation}\label{eq:hashing}
\widehat{(x_{\cdot -a}\cdot G)}_{\jj\cdot n/b}=\sum_{\ff\in [n]^d} \wh{x}_\ff e^{2\pi \ff^Ta/n}\cdot \wh{G}_{\jj\cdot n/b-\ff}.
\end{equation}
Since $\wh{G}$ is essentially $1$ on the $\ell_\infty$ ball around the center $\jj\cdot n/b$ of the `bucket' and essentially zero outside,~\eqref{eq:hashing} gives the algorithm time domain access to the restriction of $\wh{x}$ to the ``bucket,'' i.e., the essential support of $\wh{G}$, where $a\in [n]^d$ is the location in time domain at which the signal is being accessed. A pseudorandom permutation of the frequency space ensures that such a bucket is likely to contain just a single element of the support, which enables the algorithm to recover at least a constant fraction of elements in a single round and perform iterative recovery. Furthermore, if the (essential) support of $G$ in time domain is small, one obtains an efficient algorithm.

The difficulty that arises in using~\eqref{eq:hashing} in high dimensions is the fact that it is not known how to ensure that $\wh{G}$ is close to $1$ in an appropriately defined ``bucket'' while simultaneously ensuring that $|\supp G|$ is small. For example, the filters constructed in~\cite{HIKP2} ensure that $\wh{G}$ is polynomially close to $1$ in Fourier domain, but this comes at the expense of $|\supp G|$ being larger than $k$ (the ideal support size) by a factor of $\Theta(\log n)$, and this effect is even more pronounced in higher dimensions, resulting in a $\log^d n$ loss in runtime. The other extreme would be to choose $G$ to be equal to $1$ on an $\ell_\infty$ ball with $k$ points around the origin, but in that case, its Fourier transform $\wh{G}$ is the sinc function, which is only a constant factor approximation to the indicator of the corresponding $\ell_\infty$ box in Fourier domain (i.e., the ideal ``bucket''). In dimension $d$, the approximation degrades to $c^d$ for some constant $c\in (0, 1)$, leading to exponential loss in runtime. Indeed, suppose that all elements of $\wh{x}$ have roughly the same value. Then for a given element $\ff\in \supp \wh{x}$, the expected contribution of other elements to the noise in the ``bucket'' that $\ff$ is hashed to is $||\wh{x}||_2^2/B$, but the contribution of $\wh{x}_\ff$ to its own bucket is (most of the time) only $c^d$ of its value, and, hence, only an exponentially small fraction of coefficients can be recovered in a given round of hashing. \footnote{In addition, the discussion above assumes the presence of an approximate pairwise hashing lemma for high dimensions that does not lose an exponential factor in the dimension (it is known that such a lemma holds with at most about a factor of $2^d$ loss~\cite{IK14a}, but no dimension-independent version is available in the literature).}

\paragraph{Related work.} In~\cite{CheraghchiI17}, the authors presented a deterministic Sparse Fourier transform algorithm for the Hadamard transform, i.e., $d=\log_2 N$, that runs in nearly linear time in the sparsity parameter $k$, but it is not known how this extends to lower dimensions. In~\cite{Iwen10, Iw-arxiv} the author gives a $\widetilde O(k^2)$ time deterministic algorithm for the Sparse Fourier Transform, but the algorithm only applies to a related but distinctly easier problem. Specifically, the problem considers a continuous function on $[0, 2\pi)$ whose Fourier transform is bandlimited and sparse. The presented algorithm requires sampling the signal at arbitrary locations in $[0, 2\pi)$. A natural approach is to emulate sampling off-grid (i.e., at arbitrary points in $[0, 2\pi)$) given discrete samples that we have access to, which is achieved in~\cite{merhi2017new} giving an $\widetilde O(k^2)$ time deterministic algorithm for one dimensional sparse FFT. But this is a challenging task in multi-dimensional setting for several reasons. First, we are operating under the sparsity assumption alone, and no powerful general interpolation techniques that work under the sparsity assumption alone are available, to the best of our knowledge. Furthermore, even if the function were bandlimited, a natural approach to interpolation would involve some form of Taylor expansion or semi-equispaced Fourier Transform, however, both approaches incur a $\log^d N$ loss in dimension $d$. Indeed, similar exponential dependence on the dimensionality of the problem manifests itself in Fast Multipole Methods~\cite{greengard1987fast, beatson1997short} and the Sparse FFT algorithms mentioned above. Finally, one should also note that whereas the problem of computing the Fourier transform on a $p\times q$ grid with $p$ mutually prime with $q$ is equivalent to a one-dimensional Fourier transform on $\mathbb{Z}_{pq}$, the standard case of side lengths that are powers of two (for which we have the most efficient FFT algorithms) does not admit such a reduction. Furthermore, such a reduction appears to be quite challenging in high dimensions for reasons outlined above, and even more so for highly oscillatory functions that Sparse FFT algorithms need to handle.

\section{Overview of our results and techniques}\label{sec:overview}

Prior works on Sparse FFT have primarily focused on efficiently implementing hashing-based ideas developed in the extensive literature on sparse recovery using general linear measurements (e.g.,~\cite{GHIKPS}), which meets with several difficulties. In particular, the presence of multiplicative subgroups in $\mathbb{Z}_n^d$ has been a hurdle in analyzing Sparse FFT algorithms: while aliasing filters have optimal performance from the point of view of the uncertainty principle, their applications have been limited due to the fact that frequencies that belong to the same subgroup get hashed together if such filters are used, making it impossible to reason about isolation of individual frequencies. At the same time, FFT itself owes much of its efficiency to the very same multiplicative subgroups of $\mathbb{Z}_n^d$, and a natural question is whether one can design a Sparse FFT algorithm that operates on similar principles. This is precisely the approach that we take.

\paragraph{Adaptive aliasing filters.} The main technical innovation that allows us to avoid exponential dependence on the dimension and obtain Theorem~\ref{thm:main} is a new family of filters for isolating a subset of frequencies in Fourier domain in a sparse signal $\wh{x}$ using few samples in time domain. We refer to the family of filters as {\em adaptive aliasing filters.}  

\begin{definition}[$(\ff, S)$-isolating filter, informal version of Definition~\ref{IsolatingFilter-highdim}, see Section~\ref{sec:filters}] \label{def:introisol}
	Suppose $n$ is a power of two integer and $S \subseteq [n]^d$ for a positive integer $d$.  Then, for any frequency $\ff\in S$, a filter $G: [n]^d \to \C$ is called \emph{$(\ff,S)$-isolating} if $\wh{G}_\ff=1$ and $\wh{G}_{\ff'}=0$ for every $\ff'\in S\setminus \{\ff\}$. 
\end{definition}

We explain the intuition behind the construction of the filter in Section~\ref{subsec:recovadapt} below and provide the details later in Section~\ref{sec:filters}.

The reason why an $(\ff, S)$-isolating filter $G$ is useful lies in the fact that for every signal $x \in \C^{n^d}$ with $\supp\wh{x} \subseteq S$ we have, for all $\tt\in [n]^d$
\begin{align*}
\sum_{\jj\in[n]^d} x_\jj G_{\tt-\jj} &= (x*G)_{\tt}= \frac{1}{N}\sum_{\jj \in [n]^d} \wh x_\jj \cdot \wh G_\jj \cdot e^{2\pi i \frac{\jj^T\tt}{n}}= \frac{1}{N} \wh{x}_\ff e^{2\pi i \frac{\ff^T \tt}{n}} 
\end{align*}
Thus, the filter $G$ enables access to the time domain representation of the restriction of $\wh{x}$ to $\ff$ in time proportional to $|\supp G|$, at any point $\tt$. Of course, this is only useful if the support of $G$ is small. The main technical lemma of our paper shows that for every support set $S\subseteq \wh{x}$, there exists an $\ff\in S$ that can be isolated efficiently:
\begin{lemma}[Informal version of Corollary~\ref{cor:isofilter} in Section~\ref{sec:filters}] \label{lm:main-tech}
	For every power of two $n\geq 1$, positive integer $d$, and set $S \subseteq [n]^d$, there exists an $\ff \in S$ and an $(\ff,S)$-isolating filter $G$ such that $|\supp G| \leq |S|$. 
\end{lemma}
The proof of the lemma is given in Section~\ref{sec:filters}.

\paragraph{Accessing the residual signal.} Lemma~\ref{lm:main-tech} suggests a natural approach to the estimation problem with Fourier measurements in high dimensions: iteratively construct an $(\ff, S)$-isolating filter $G$, estimate $\ff$, remove $\ff$ from $S$, and proceed. The hope is that we can essentially assume that we are given access to ${\mathcal F}^{-1}(\wh{x}_{S\setminus \{\ff\}})$ once we have estimated $\ff$. In general, if we have been able to estimate the values of $\wh{x}_\ff$ for all $\ff\in C$ with some $C\subseteq S$, then we would like to obtain access to 
$$
x_\tt-\sum_{\ff\in C} \wh{x}_{\ff} \cdot e^{2\pi i \ff^T\tt}.
$$

Note that we would need $x_\tt$ for $\tt$ in the support of $G$ at the next iteration, and this support is generally a rather complicated set of size $\Omega(k)$, from which we need to subtract the inverse Fourier transform of the signal estimated so far. This problem is the non-uniform Fourier transform problem, and no subquadratic methods for subtraction are known even in dimension $d=1$ when the set in time domain that we want to compute the inverse Fourier transform on is arbitrary. Even if the target set is an $\ell_\infty$-box, the best known algorithms for this problem run in time $\Omega(k\log^{d} (1/\e))$, where $\e>0$ is the precision parameter of the computation---this reduces to quadratic time even when $d=\Omega(\log k/\log\log k)$ and inverse polynomial in $k$ precision is desired. Thus, subtracting from time domain would result in at least cubic runtime in $k$. Instead, we subtract the influence of the residual in frequency domain, which requires $O(k)$ evaluations of $\wh{G}$ (as we show, $\wh{G}$ can be evaluated at a cost of just $O(\log N)$). Note that it is crucial here that we peel off one coefficient at at time. Any improvements to this process, if they were to achieve $k^{2-\Omega(1)}$ runtime overall, would likely also imply improvements in the computation of approximate \emph{non-uniform} Fourier transform: given a $k$-sparse signal $\wh{x}$ and a set $T\subseteq [n]^d$ with $|T|\leq k$, output $y:[n]^d\to \C$ such that $||(x-y)_T||_2^2\leq \e ||x||_2^2$. However, it seems plausible that quadratic runtime in $k$ is essentially optimal for the non-uniform Fourier transform problem: specifically, that  under natural complexity theoretic assumptions there exists no algorithm for the $\e$-approximate non-uniform Fourier transform problem with runtime $k^{2-\Omega(1)}$ when $d=\Omega(\log k)$ and $\e<1/k^C$ for sufficiently large constant $C$. We note that current techniques do not provide a subquadratic algorithm even for simple sets $T$ such as the $\ell_\infty$ box with $k$ points in dimension $d=\Omega(\log k/\log\log k)$ (due to the $k\log^d (1/\e)$ dependence mentioned above; a similar exponential dependence on the dimension is present in Fast Multipole Methods~\cite{1987JCoPh..73..325G,beatson-greengard}). For an arbitrary set $T$ no subquadratic algorithm is known even when $d=1$.

\paragraph{Putting it together: estimation with Fourier measurements}
Combining the aforementioned ideas, we are able to develop a deterministic algorithm for the \emph{estimation problem with Fourier measurements} in high dimensions:
\begin{equation}\label{eq:exact-est}
\begin{split}
\text{\bf Input:}& \text{~~~access to $x:[n]^d\to \C$,}\\
&\text{~~~subset $S\subseteq [n]^d$ such that $\supp \wh{x}\subseteq S$}\\
\text{\bf Output:}& \text{~~~$\wh{x}_S$}
\end{split}
\end{equation}
For the estimation problem~\eqref{eq:exact-est} we obtain the following result.
\begin{theorem}[Estimation guarantee, informal version of Theorem~\ref{thm:est-main} in Section~\ref{sec:worstcase1d}]\label{thm:intro-est-main}
	Suppose $n$ is a power of two integer, $d$ is a positive integer, and $S\subseteq [n]^d$. Then, for any signal $x \in \C^{n^d}$ with $\supp{\wh x} \subseteq S$, the procedure \textsc{Estimate}$(x,S,n,d)$ (see Algorithm~\ref{alg:high-dim-Est-k2}) recovers $ \wh{x}$. Moreover, the sample complexity of this procedure is $O(|S|^2)$ and its runtime is $O(|S|^2 \cdot \log N)$. Furthermore, the procedure \textsc{Estimate} is deterministic.
\end{theorem}

In the rest of this section, we give an overview of our techniques. Throughout the section, we present our results for the one-dimensional setting, as this makes notation simpler. All our results translate to the high-dimensional setting without any loss---see Section~\ref{sec:filters-d} for details.

\subsection{Recovery via adaptive aliasing filters} \label{subsec:recovadapt}
Our main theorem is the following, which presents an algorithm for problem \eqref{eq:exact-sfft} for worst-case signals.

\begin{restatable}[Sparse FFT for worst-case signals]{theorem}{sfftworstcase} \label{thm:sfft-worstcase}
	For any power of two integer $n$ and any positive integer $d$ and any signal $x\in \C^{n^d}$ with $\|{\wh{x}}\|_0 = k$, the procedure \textsc{SparseFFT}$(x, n, d, k)$ in Algorithm~\ref{alg:fullsparsefft} recovers $ {\wh{x}}$. Moreover, the sample complexity of this procedure is $O(k^3 \log^2k \log^2 N)$ and its runtime is $O(k^3 \log^2k \log^2 N)$.
\end{restatable}

The major difference between estimation and recovery (i.e., problem \eqref{eq:exact-est} vs. \eqref{eq:exact-sfft}) is the fact in the latter problem, the set $S$ of frequencies is unknown to us: the algorithm is only given access to $x$ and an upper bound on the sparsity of $\wh{x}$. Since our $(f, S)$-separating filter is adaptive, i.e., depends on $S$, this appears to present a challenge. However, we circumvent this challenge by constructing a sequence of successive approximations to the set $S$. In dimension $1$, these approximations amount to reducing $S$ modulo $2^j$ for all $j=1,\ldots, \log_2 n$, and adaptively probing to learn which of the residue classes are nonzero. As before, our approach extends seamlessly to high dimensions by simply concatenating the $d$ coordinates into a single vector. Note that this is in sharp contrast to all previously known approaches, which are more efficient in low dimensions, but incur an exponential loss overall. We would like to note that at a high level one can view our filtering approach as a way to prune the FFT computation graph in a way that suffices for recovery of a $k$-Fourier sparse vector.

We outline the main ideas in one-dimensional setting here to simplify the presentation (see Section~\ref{sec:filters-d} for the high-dimensional version of the argument). Let $N=n$ be the length of the signal and $d=1$ be the dimension for $n$ a power of two. We define $\tfull_n$ to be a full binary tree of height $\log_2 n$ and define a labelling scheme on the vertices as follows.
\begin{restatable}{definition}{deftfull}
	\label{def:t-full}
	Suppose $n$ is a power of two. Let $\tfull_n$ be a full binary tree of height $\log_2 n$, where for every $j \in \{0, 1, \dots, \log_2 n\}$, the nodes at level $j$ (i.e., at distance $j$ from the root) are labeled with integers in $[2^j]$. For a node $v\in \tfull_n$, we let $f_v$ be its label. The label of the root is $f_{root} =0$. The labelling of $T_n^{full}$ satisfies the condition that for every $j \in [\log_2 n]$ and every $v$ at level $j$, the right and left children of $v$ have labels $f_v$ and $f_v+2^j$, respectively. Note that the root of $\tfull_n$ is at level 0, while the leaves are at level $\log_2 n$. 
\end{restatable}
The tree captures the computation graph of FFT algorithm, where leaves correspond to frequencies in $[n]$ (given by the label), and for any $j \in \{0, 1, \dots, \log_2 n\}$, the nodes at level $j$ (i.e., at distance $j$ from the root) correspond to congruence classes of frequencies modulo $2^j$, as specified by the labelling (see Figure~\ref{subfig:tfull}).

\begin{figure*}[t!]
	\centering
	\begin{subfigure}[t]{0.45\textwidth}
		\centering
		\scalebox{.6}{
			\begin{tikzpicture}[level/.style={sibling distance=60mm/#1,level distance = 2cm}]
			\node [arn] (z){}
			child {node [arn] (a){}
				child {node [arn] (b){}
					child {node [arn] (c){}}
					child {node [arn] (d){}}
				}
				child {node [arn] (e){}
					child {node [arn] (f) {}}
					child {node [arn] (g){}}
				}
			}
			child { node [arn] (h){}
				child {node [arn] (i){}
					child {node [arn] (j){}
					}
					child {node [arn] (k) {}}
				}
				child {node [arn] (l){}
					child {node [arn] (m){}}
					child {node [arn] (n){} }
				}
			};
			
			\node []	at (a)	[label=left:1]	{};
			\node []	at (b)	[label=left:11]	{};
			\node []	at (c)	[label=below:111] {};
			\node []	at (d)	[label=left:011] {};
			\node []	at (e)	[label=left:01] {};
			\node []	at (f)	[label=below:101] {}; 
			\node []	at (g)	[label=left:001] {}; 
			\node []	at (h)	[label=left:0] {}; 
			\node []	at (i)	[label=left:10] {}; 
			\node []	at (j)	[label=below:110] {}; 
			\node []	at (k)	[label=left:010] {}; 		
			\node []	at (l)	[label=left:00] {}; 
			\node []	at (m)	[label=below:100] {}; 
			\node []	at (n)	[label=left:000] {};

			\node [] at (-6,0) [label=right:\LARGE${T_8^{full}}$]	{};
			\draw[draw=black,very  thick, ->] (-3.9,-0.1) -- (-1.5,-0.6);
			
			\end{tikzpicture}
		}
		\par
		
		\caption{$T_8^{full}$ with binary labelling.}
		\label{subfig:tfull}
	\end{subfigure}
	~ 
	\begin{subfigure}[t]{0.45\textwidth}
		\centering
		\scalebox{.6}{
			\begin{tikzpicture}[level/.style={sibling distance=60mm/#1,level distance = 2cm}]
			\node [arn] (z){}
			child {node [arn] {}edge from parent [electron]
				child {node [arn] {}edge from parent [draw=black,very thin]
					child {node [arn] {}edge from parent [draw=black,very thin]}
					child {node [arn] {}edge from parent [draw=black,very thin]}
				}
				child {node [arn] {}edge from parent [electron]
					child {node [arn_l] (a) {}edge from parent [electron]}
					child {node [arn] {}edge from parent [draw=black,very thin]}
				}
			}
			child { node [arn] {} edge from parent [electron]
				child {node [arn] {}edge from parent [electron]
					child {node [arn] {}edge from parent [draw=black,very thin]
					}
					child {node [arn_l] (b) {}edge from parent [electron]}
				}
				child {node [arn] {}edge from parent [electron]
					child {node [arn_l] (c){}edge from parent [electron]}
					child {node [arn] {} edge from parent [draw=black,very thin]}
				}
			};
			
			\node []	at (a)	[label=below:101]	{};
			\node []	at (b)	[label=below:010]	{};
			\node []	at (c)	[label=below:100]	{};
			
			\node [] at (-6.5,0) [label=right:\LARGE
			{Splitting tree T}]	{};
			\draw[draw=black,very  thick, ->] (-3.5,-0.5) -- (-1.8,-1);
			
			\end{tikzpicture}
		}
		\par
		
		\caption{A splitting tree of depth $3$ with three leaves. Set of leaves $S = \{ 2,4,5 \}$.}
		\label{subfig:splittree}
	\end{subfigure}
	\caption{An example of $T_n^{full}$ and a splitting tree with $n=8$ and binary labelling.} \label{Tfull-3sparse}
\end{figure*}
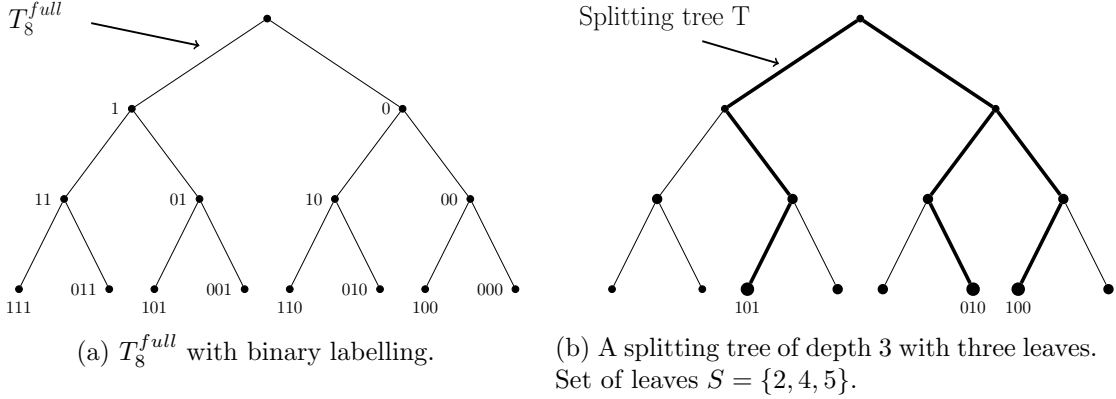

Note that the full FFT algorithm starts from the root of $\tfull_n$ and computes the congruence classes of the Fourier transform of signal $x$ at each level of this tree iteratively. Because it can reuse the computations from each level for computing the next levels, the total runtime of FFT is $O(n\log_2 n)$.

In order to speed up the computation for sparse signals, we introduce the notion of a {\em splitting tree}, which is nothing but the subtree of $\tfull_n$ that contains the nonzero locations of $\wh{x}$ together with paths connecting them to the root. Given a set $S\subseteq [n]$ (the support of $\wh{x}$ in Fourier domain), we define the \emph{splitting tree} of the set $S$ as follows:

\begin{restatable}[Splitting tree]{definition}{defsplit} \label{def:splittree}
	Let $n$ be a power of two. For every $S \subseteq [n]$, the \emph{splitting tree} $T=\tree(S, n)$ of a set $S$ is a binary tree that is the subtree of $\tfull_n$ that contains, for every $j\in [\log_2 n]$, all nodes $v \in \tfull_n$ at level $j$ such that $\{ f \in S : f \equiv f_v \pmod{2^j} \} \neq \emptyset $. 
\end{restatable}

An illustration of such a tree is given in Figure~\ref{subfig:splittree}. In order to recover the identities of the elements in $S$, our algorithm performs an exploration of this tree. At every point, the algorithm constructs a filter $G$ that isolates {\em frequencies in a given subtree} and tests whether that subtree contains a nonzero signal. In order to make this work, we need a construction of filters that isolates the entire subtree as opposed to only one element, as Definition~\ref{def:introisol} does. Fortunately, the actual $(f, S)$-isolating filters $G$ constructed in Lemma~\ref{lm:main-tech} satisfy precisely this property. The stronger isolation properties are captured by the following definition:

\begin{restatable}[Frequency cone of a leaf of $T$]{definition}{deffreqcone}\label{def:iso-t}
	For every power of two $n$, subtree $T$ of $\tfull_n$, and vertex $v\in T$ which is at level $l_T(v)$ from the root, define the {\em frequency cone of $v$ with respect to $T$} as
	$$
	\subtree_T(v):=\left\{ f \in [n]: f \equiv f_v \pmod{2^{l_T(v)}} \right\}.
	$$
\end{restatable}
Intuitively, the frequency cone of a node $v$ in $T$ captures all potential nonzeros of $\wh{x}$ that belong to the subtree of $v$ in $T$ (see Figure~\ref{fig:conjclass}). Our adaptive filter construction lets us obtain time domain access to the corresponding part of the frequency space:

\begin{restatable}[$(v, T)$-isolating filter]{definition}{defvtisol} \label{def:v-t-isolating}
	For every integer $n$, subtree $T$ of $\tfull_n$, and leaf $v$ of $T$, a filter $G\in \C^n$ is called \emph{$(v,T)$-isolating} if the following conditions hold:
	\begin{itemize}
		\item For all $f\in \subtree_T(v)$, we have $\wh{G}_{f}=1$.
		\item For every $f'\in \bigcup_{\substack{u \neq v \\ u: \text{~leaf of~}T}} \subtree_T(u)$, we have $\wh{G}_{f'}=0$.
	\end{itemize}
\end{restatable}

Note that for all signals $x \in \C^n$ with $\supp\wh{x} \subseteq \bigcup_{u: \text{~leaf of~}T} \subtree_T(u)$ and $t\in[n]$,
\[ \sum_{j\in[n]} x_j G_{t-j} = \frac{1}{n} \sum_{f \in \subtree_T(v)}\wh{x}_f e^{2\pi i \frac{f t}{n}}. \]

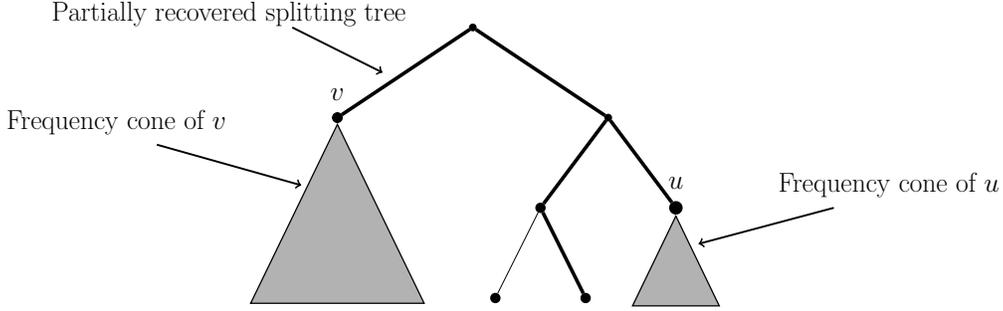
\begin{figure}
	\centering
	\scalebox{.6}{
		\begin{tikzpicture}[level/.style={sibling distance=60mm/#1,level distance = 2cm}]
		\node [arn] (x){}
		child {node (b)[arn_l] {}edge from parent [electron]
		}
		child{node [arn] {}edge from parent [electron]
			child {node [arn] {}
				child {node [arn] {}edge from parent [very thin]}
				child {node [arn] {}}
			}
			child { node (a) [arn_l] {}
			}
		};
		
		\node [tr]	at (a.south)	[label]	{};
		\node [br]	at (b.south)	[label]	{};
		
		\node []	at (a.north)	[label=\LARGE{$u$}]	{};
		\node []	at (b.north)	[label=\LARGE{$v$}]	{};	
		
		\node [] at (-9.6,0.3) [label=right:\LARGE
		{Partially recovered splitting tree}]	{};
		\draw[draw=black,very  thick, ->] (-4,-0) -- (-2,-1);		
		
		\node [] at (-10.6,-2.1) [label=right:\LARGE
		{Frequency cone of $v$}]	{};
		\draw[draw=black,very  thick, ->] (-7,-2.6) -- (-3.8,-3.5);
		
		\node [] at (6.5,-3.5) [label=right:\LARGE
		{Frequency cone of $u$}]	{};
		\draw[draw=black,very  thick, ->] (8.0,-4.0) -- (5,-4.8);
		\end{tikzpicture}
	}
	\par
	\caption{A partially recovered splitting tree (shown in bold). Frequency cones of $u$ and $v$ correspond to the subtrees rooted at nodes $u$ and $v$, respectively, which have not been discovered yet.}
	\label{fig:conjclass}
\end{figure}

\paragraph{Iterative tree exploration process leading to an algorithm with $\widetilde{O} (k^3)$ runtime.} Now that we have defined the framework for our algorithm, we need to specify the order in which the algorithm will be accessing the leaves of the tree in order to minimize runtime. This is governed by the cost of constructing and using a $(v, T)$-isolating filter for various nodes $v$ in $T$. To quantify cost, we introduce the notion of a \emph{weight} of a leaf in the tree.
\begin{restatable}[Weight of a leaf]{definition}{defwt} \label{def:weight}
	Suppose $n$ is a power of two. Let $T$ be a subtree of $\tfull_n$. Then for any leaf $v\in T$, we define its \emph{weight} $w_T(v)$ \emph{with respect to} $T$ to be the number of ancestors of $v$ in tree $T$ with two children.
\end{restatable}
It turns out that the techniques from Lemma~\ref{lm:main-tech} also yield the following.
\begin{lemma}[Informal version of Lemma~\ref{lem:filter-isolate} in Section~\ref{sec:filters}] \label{lem-informal:filter-isolate}
	Suppose $n$ is a power of two. Let $T$ be a subtree of $\tfull_n$. Then for any leaf $v\in T$, there exists a $(v, T)$-isolating filter $G$ with $|\text{supp}~ G|\leq 2^{w_T(v)}$ such that $G$ and $\wh{G}$ can be evaluated at  $O(\log N)$ cost per point.
\end{lemma}

\begin{figure}
	
	\centering
	\scalebox{.6}{
		\begin{tikzpicture}[level/.style={sibling distance=60mm/#1,level distance = 2cm}]

		\node [arn] at (-2,0) [label=above:\LARGE{Splitting tree $T$}]{};
		\node [brr]	at (-2,0)	[label]	{};

		\node [trrr] at (-1.05,-2) [label]	{};
		
		\node [arn] (x) at (-1.05,-2){}
		child [level distance = 1cm,sibling distance=10mm] {node{} edge from parent [draw = white]
		}
		child [level distance = 1cm,sibling distance=10mm] {node [arn] {}
			child [level distance = 1cm,sibling distance=10mm] {node [arn] {}
				child [level distance = 1cm,sibling distance=10mm] {node{} edge from parent [draw = white]
				}
				child [level distance = 1cm,sibling distance=10mm] {node [arn] [label=left:\LARGE{$v_0$}]{}
				}
			}
			child [level distance = 1cm,sibling distance=10mm] {node{} edge from parent [draw = white]
			}
		};
		
		\draw[ dashed,very thick] (-1.05,-2) -- (-2,-4);
		\draw[ dashed,very thick] (-0.1,-4) -- (-2,-4);
		\draw[ dashed,very thick] (-1.05,-2) -- (-0.1,-4);
		
		\node [trr]	at (-0.55,-5)	[label]	{};
		
		\draw[draw=black, very thick, <->] (-4.2,0) -- (-4.2,-4);
		\node []	at (-4.1,-2)	[label=left:\LARGE $h$]	{};
		
		\draw[draw=black,  very thick, <->] (-0.5,0) -- (-0.5,-2);
		\node []	at (-0.6,-1)	[label=right:\LARGE $h_0$]	{};

		\draw[ draw=black,  very thick, <->] (-5.5,0) -- (-5.5,-7);
		\node []	at (-5.5,-3.5)	[label=left:\LARGE $\log_2n$]	{};
		
		\node [arn]	at (-0.2,-7)	[label=below:\LARGE $v$]	{};

		\draw[ draw=black,  very thick, ->] (3,-4) -- (13,-4);
		\foreach \x in { 0, 0.1, 0.2, 0.3 , 0.4, 0.5, 0.6, 0.7 } \foreach \y in {-4, -1, 2 } \draw[draw=black,very thick] (8+\x+\y,-3.75) -- (8+\x +\y,-4.25);

		\draw[ draw=black,  thick, <->] (4,-3.5) -- (4.7,-3.5);
		\node []	at (4.8,-3.5)	[label=above:\LARGE $2^{h-h_0}$]	{};
		
		\draw[ draw=black,  thick, <->] (7.3,-3.5) -- (10.35,-3.5);
		\node []	at (8.9,-3.5)	[label=above:\LARGE $n/2^{h_0}$]	{};
		
		\node [] at (3,-1.5) [label=right:\LARGE {Support of $(v,T)$-isolating filter $G$}]	{};	
		\end{tikzpicture}}
	\par
	\caption{An instance of a $(v,T)$-isolating filter G, where the weight of leaf $v$ is $h$ and hence the filter $G$ satisfies $|\supp G| = 2^{w_T(v)} = 2^h$.} \label{fig:filter-time}
\end{figure}
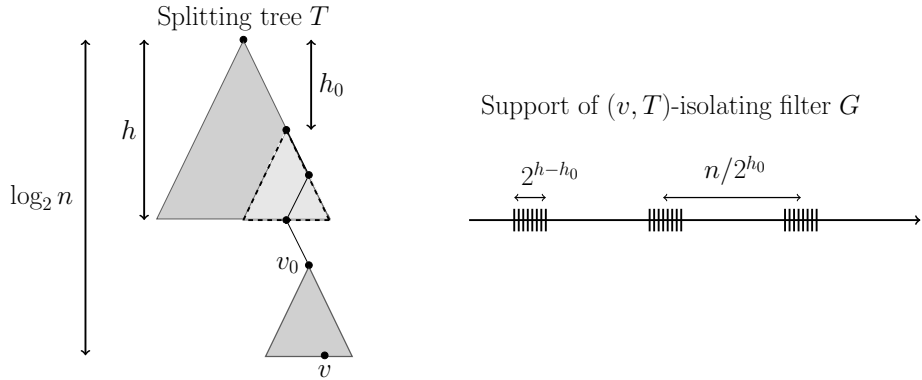

Before describing the algorithm we give an example illustrating  filter support in time domain. Consider a complete binary tree $T$ of height $h \ll \log_2n$. Suppose that $v_0$ is some vertex at level $h_0 < h$ of this tree. Now we take the subtree rooted at $v_0$ and add an appendage of length $\log_2n - h$ to $v_0$. The appendage is a path of $\log_2n - h$ nodes each of which has a single child. This doesn't change the weight of any of the leafs of the original tree because every node on the appendage has exactly one child. One can see an example of such tree in Fig.~\ref{fig:filter-time}. Suppose that the leaf $v$ is a leaf of the subtree rooted at $v_0$, which is moved far from the root by the appendage. In order to isolate $v$ from the elements that are not in the subtree of $v_0$ we need a filter which is $(n/2^{h_0})$-periodic in time domain and in order to isolate from the rest of the elements in subtree of $v_0$ the filter needs to sample the signal at a fine grid of length $2^{h-h_0}$. Note that the support of a $(v,T)$-isolating filter $G$ is $\supp{G} = \left\{ i + (n/2^{h_0}) \cdot j ; \, j \in [2^{h_0}], i \in [2^{h-h_0}] \right\}$. In Fig.~\ref{fig:filter-time} we exhibit a $(v,T)$-isolating filter $G$ which is constructed based on Lemma~\ref{lem:filter-isolate}, where $v$ and $T$ correspond to this instance of splitting tree.

Given Lemma \ref{lem-informal:filter-isolate}, our algorithm is natural. We find the vertex $v^*=\text{argmin}_{v\in T} w_T(v)$, which, by Kraft's inequality, satisfies $w_T(v^*)\leq \log_2 k$. We then define an auxiliary tree $T'$ by appending a left $a$ and a right child $b$ to $v$. Then for each of the children $a, b$, we, in turn, construct a filter $G$ that isolates them from the rest of $T$ (i.e., from the frequency cones of other nodes in $T$) and check whether the corresponding restricted signals are nonzero. The latter is unfortunately a nontrivial task, since the sparsity of these signals can be as high as $k$, and detecting whether a $k$-sparse signal is nonzero requires $\Omega(k)$ samples. However, a fixed set of $k\log^3 N$ locations that satisfies the restricted isometry property (RIP) can be selected, and accessing the signal on those values suffices to test whether it is nonzero. The overall runtime becomes $\widetilde O(k^3)$: the isolating filter has support at most $2k$, while the number of samples needed to test whether the two subtrees of $v$ are nonempty is $\widetilde O(k)$, so peeling off $\leq k$ elements takes $\widetilde O(k^3)$ time overall. This results in Theorem~\ref{thm:sfft-worstcase} (the algorithm is presented as Algorithm~\ref{alg:fullsparsefft}).

\paragraph{$\widetilde O(k^2)$ runtime under random phase assumption.} We note that the runtime can be easily reduced to $\widetilde O(k^2)$ if assumptions are made on the signal that ensure that its energy is evenly spread across time domain, making $\widetilde O(1)$ samples sufficient to detect whether the signal is zero or not. This occurs, for instance, if a signal's frequencies satisfy distributional assumptions (e.g., the values have random phases). We present such a result in Section~\ref{sec:randomsign}. It seems that even under this assumption on the values of the signal, since the support of the signal in Fourier domain is worst case, reducing the runtime below $k^2$ likely requires a major advance in techniques for non-uniform Fourier transform computation.

More formally, we introduce the notion of a \emph{worst-case signal with random phase} as follows:
\begin{restatable}[Worst-case signal with random phase]{definition}{randsign} \label{random-sign}
	For any positive integer $d$ and power of two $n$, we define $x$ to be a \emph{worst-case signal with random phase} having values $\{\beta_\ff\}_{\ff\in [n]^d}$ if
	\[
	\wh{x}_{\ff} = \beta_\ff e^{2\pi i \theta} \quad \text{for uniformly random $\theta\in [0,2\pi)$},
	\]
	independently for every $\ff \in [n]^d$. Furthermore, if $k$ of the values $\{\beta_\ff\}_{\ff\in [n]^d}$ are nonzero, then $x$ is said to be a \emph{worst-case $k$-sparse signal with random phase} and is guaranteed to have sparsity $k$.
\end{restatable}
Note that ``worst-case'' in the above definition signifies the fact that the \emph{support} of the signal is arbitrary (having no distributional assumptions), subject to a potential sparsity constraint. We then present the following theorem:
\begin{restatable}[Sparse FFT for worst-case signals with random support]{theorem}{sfftrandphase} \label{thm:sfft-rand-signs}
	For any power of two integer $n$, positive integer $d$, and worst-case $k$-sparse signal with random phase $x\in \C^{n^d}$, the procedure \textsc{SparseFFT-RandomPhase}$(x, n, d, k)$ in Algorithm~\ref{alg:sparsefftphase} recovers $ {\wh{x}}$ with probability $1-\frac{1}{N^2}$. Moreover, the sample complexity of this procedure is $O(k^2 \log^4 N)$ and its runtime is $O(k^2 \log^4 N)$.
\end{restatable}

\paragraph{Impossibility of reducing the number of iterations (rounds of adaptivity): signals with low Hamming weight support.} We note that our algorithm differs from all prior works in that it uses many rounds of adaptivity. Indeed, the samples that our algorithm takes are guided by values of the signal that have been read in previously queried locations, which is in contrast to most prior Sparse Fourier Transform algorithms. Two notable exceptions in recent literature include the adaptive block Sparse FFT algorithms of \cite{CKSZ17} and \cite{ChenKPS16}.

Our algorithm uses $k$ rounds of adaptivity, peeling off one element at a time. It would be desirable to reduce the number of rounds of adaptivity by perhaps peeling off many elements in one batch as opposed to one at a time. For example, if the locations of the nonzeros of $\wh{x}$ are uniformly random in $[n]^d$, then the splitting tree of $x$ is likely to be rather balanced (see, e.g.~Fig.~\ref{fig:split-rand} for an illustration), so perhaps one can find a filter $G$ that has small support and can be efficiently used to isolate many coefficients at once?  Indeed, this intuition turns out to be correct for signals with uniformly random supports---we show in Section~\ref{subsec:batchpeel} below (with details presented in Section~\ref{sec:avgcase}) that this idea yields a $\widetilde O(k)$ time algorithm. However, rather surprisingly, adversarial instances exist that force the peeling process to use $k^{1-o(1)}$ rounds of adaptivity in the worst case, making our analysis essentially tight. We now present this adversarial instance.

\begin{restatable}[Hamming ball]{definition}{lowhamming} \label{low-Hamming}
	For any power of two integer $n$ any integer $0\leq c\leq \log_2 n$, we define $H_c^n$ to be the \emph{closed Hamming ball} of radius $c$ centered at 0:
	\[ H_c^n= \{ f\in [n] : w(f) \le c \}, \]
	where $w(f)$ is the Hamming weight of the binary representation of $f$, i.e., $w(f)$ is the number of ones in the binary representation of $f$.
\end{restatable}
Note that $|H_c^n| = \sum_{j=0}^c \binom{\log_2 n}{j}$.

\begin{restatable}[Class of signals with low Hamming support]{definition}{lowhammingsupport} \label{low-Hamming-support}
	For any power of two integer $n$ and any integer $c$, Let $\mathcal{X}^n_c$ denote the class of signals in $\C^n$ with support $H^n_c$ as in Definition~\ref{low-Hamming},
	$$\mathcal{X}^n_c = \{ x \in \C^n : \supp{ x } \subseteq H^n_c \}$$
\end{restatable}
Note that for any $x \in \mathcal{X}^n_c$ we have that $\| x \|_0 = \sum_{i=0}^c \binom{\log_2n}{i}$, so for any $c \le (\frac{1}{2} - \epsilon)\log_2n$, the signals that are contained in class $\mathcal{X}^n_c$ are $\Theta\Big( \binom{\log_2n}{c} \Big)$-sparse.

\begin{restatable}[Low Hamming weight binary trees]{definition}{lowhammingweighttrees}
	Suppose $n$ is a power of two integer. Then, we define a \emph{low Hamming weight binary tree $T_c^n$} inductively for $c=0,1,\dots, \log_2 n$:
	\begin{enumerate}
		\item $T_0^n$ is defined to be the unique tree of depth $\log_2 n$ that has a single leaf node and satisfies the property that each non-leaf node has a single right child only. Thus, $T_0^n$ has $\log_2 n + 1$ nodes.
		
		\item For any $c > 0$, $T_c^n$ is constructed as follows: Take $T_0^n$ and label the nodes in order from the root to the leaf as $0,1,\dots, \log_2 n$. Then, for each node $0\leq j < \log_2 n$, take a copy of $T_{c-1}^{n/2^{j+1}}$ and let its root be the left child of node $j$. The resulting tree defines $T_c^n$.
		
	\end{enumerate}
	Note that all the leaves of $T^n_c$ are at level $\log_2n$.
\end{restatable}
It is not hard to see that $T^n_c$ is in fact the splitting tree for the set $H^n_c$ and, hence, the number of its leaves is $\sum_{i=0}^c \binom{\log_2n}{i}$.  An illustration of the tree $T^n_c$ for $c=2$ and $n=32$ is shown in Figure~\ref{fig:hamming-weight-2}.

\begin{figure}
	\centering
	\scalebox{.7}{
		\begin{tikzpicture}[level/.style={sibling distance=50mm/#1, level distance = 2cm}]
		
		label distance=3mm,
		every label/.style={blue},

		\node [arn] (z){}
		child [sibling distance=130mm] {node [arn] (a) {}
			child {node [arn_l] (b) {}}
			child { node [arn] (c){} 
				child {node [arn_l] (d){}}
				child { node [arn] (e){} 
					child {node [arn_l] (f){}}
					child { node [arn] (g){} 
						child {node [arn_l] (h){}}
						child { node [arn_l] (i){}
						}
					}
				}
			}
		}
		child { node [arn] (j){} 
			child [sibling distance=97mm] {node [arn] (k){}
				child {node [arn_l] (l){}}
				child { node [arn] (m){} 
					child {node [arn_l] (n){}}
					child { node [arn] (o){} 
						child {node [arn_l] (p){}}
						child { node [arn_l] (q){} 
						}
					}
				}
			}
			child { node [arn] (r){} 
				child [sibling distance=65mm]{node [arn] (s){}
					child {node [arn_l] (t){}}
					child { node [arn] (u){} 
						child {node [arn_l] (v){}}
						child { node [arn_l] (w){} 
						}
					}
				}
				child { node [arn] (x){} 
					child [sibling distance=30mm]{node [arn] (xx){}
						child {node [arn_l] (y){}}
						child { node [arn_l] (zz){} 
						}
					}
					child { node [arn] (aa){} 
						child {node [arn_l] (bb){}}
						child { node [arn_l] (cc){} 
						}
					}
				}
			}
		};
		
		\node []	at (z.north)	[label=left:]	{};
		\node []	at (a.north)	[label=left:]	{};
		\node []	at (b.north)	[label=left:11]	{};
		\node []	at (c.north)	[label=left:]	{};
		\node []	at (d.north)	[label=left:101]	{};
		\node []	at (e.north)	[label=left:]	{};
		\node []	at (f.north)	[label=left:1001]	{};
		\node []	at (g.north)	[label=left:]	{};
		\node []	at (h.north)	[label=left:10001]	{};
		\node []	at (i.north)	[label=left:1]	{};
		\node []	at (j.north)	[label=right:]	{};
		\node []	at (k.north)	[label=left:]	{};
		\node []	at (l.north)	[label=left:110]	{};
		\node []	at (m.north)	[label=left:]	{};
		\node []	at (n.north)	[label=left:1010]	{};
		\node []	at (o.north)	[label=left:]	{};
		\node []	at (p.north)	[label=left:10010]	{};
		\node []	at (q.north)	[label=left:10]	{};
		\node []	at (r.north)	[label=left:]	{};
		\node []	at (s.north)	[label=left:]	{};
		\node []	at (t.north)	[label=below:1100]	{};
		\node []	at (u.north)	[label=right:]	{};
		\node []	at (v.north)	[label=below:10100]	{};
		\node []	at (w.north)	[label=below:100]	{};
		\node []	at (x.north)	[label=left:]	{};
		\node []	at (xx.north)	[label=right:]	{};
		\node []	at (y.north)	[label=below:11000]	{};
		\node []	at (zz.north)	[label=below:1000]	{};
		\node []	at (aa.north)	[label=left:]	{};
		\node []	at (bb.north)	[label=below:10000]	{};
		\node []	at (cc.north)	[label=below:0]	{};
		
		\node [] at (-9.4,0) [label=right:\Large{Splitting tree with}]	{};
		\node [] at (-9.4,-0.6) [label=right:\Large{Hamming weight 2}]	{};
		\draw[draw=black,very  thick, ->] (-5,0) -- (-1.8,-0.3);
		\end{tikzpicture} 
	}
	\par
	\caption{Illustration of $T_2^n$, the splitting tree corresponding to a family of signals with Hamming weight $2$. For simplicity, we truncate terminal rightward paths from leaves to the bottom level of the tree from the picture.  The corresponding support set of this tree is $S = \{ 0, 1, 2, 3, 4, 5, 6, 8, 9, 10, 12, 16, 17, 18, 20, 24\}$. Note that all the elements of $S$ have binary representations with Hamming weight at most 2.}\label{fig:hamming-weight-2}
\end{figure}
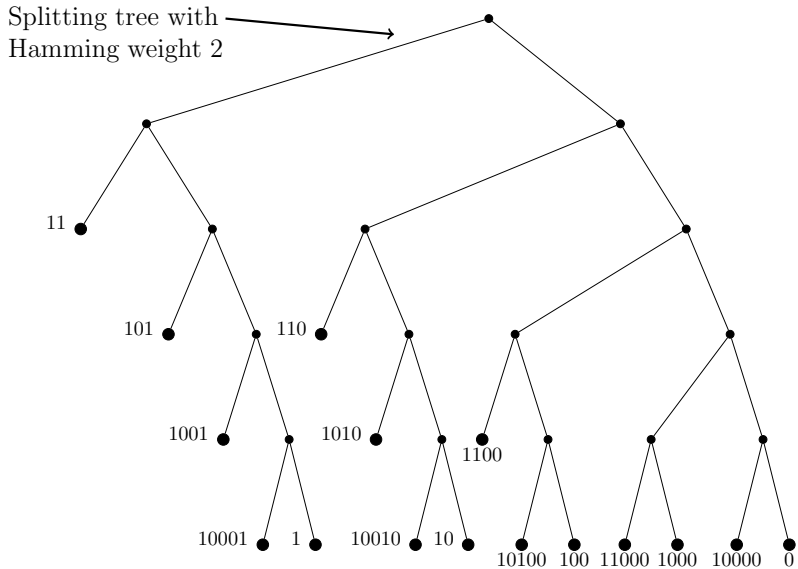

We prove the following theorem in Section~\ref{hard-instance} (see Theorem~\ref{lem:process}):

\begin{theorem}[Informal version of Theorem~\ref{lem:process}]
	A peeling process with threshold $\tau\leq \log_2 k+O(1)$ (i.e. any threshold that allows isolation of an element at cost bounded by $O(k)$) must take $k^{1-o(1)}$ iterations to terminate. 
\end{theorem}

To add to the result above, we note that the lower bound on the number of rounds of adaptivity is not the only cause for quadratic runtime in our algorithm. The other cause is the necessity to update the residual signal as more and more elements are recovered, i.e. perform non-uniform Fourier transform computations. Since no subquadratic approach to this problem are known in high dimensions, it seems plausible that a $k^{2-\Omega(1)}$ runtime algorithm for high-dimensional FFT would also shed light on the complexity of this intriguing problem.

\subsection{Runtime $\widetilde O(k)$ for random supports through a batched peeling process} \label{subsec:batchpeel}

To complement our lower bound of $k^{1-o(1)}$ rounds of adaptive pruning for \emph{worst-case} signals using our adaptive aliasing filters, we show that if the support of the signal is uniformly random, adaptive aliasing filters can be used to achieve an algorithm with $\widetilde O(k)$ runtime. A beautiful $\widetilde O(k)$ runtime and optimal $O(k)$ sample complexity algorithm for this model was given in~\cite{GHIKPS}. The algorithm was stated for $d=2$ but readily extends to high dimensions. Unfortunately, it comes with a major restriction, namely, the sparsity $k$ must be $o(N^{1/d})$. Our approach is different and extends to all $k\leq N$.

We now introduce the notion of a Fourier sparse \emph{random support signal}:
\begin{restatable}[Random support signal]{definition}{defrandsupp}\label{def:randsupport}
	For any positive integer $d$, power of two $n$, and arbitrary $\beta: [n]^d \to \C$, we define $x: [n]^d \to \C$ to be a \emph{random support signal} of Fourier sparsity $k$ (with values given by $\beta$) if $\wh x$ is the signal defined by
	\[
	\wh x_\ff = \begin{cases}
	\beta_\ff &\qquad\text{with probability $k/n^d$}\\
	0 &\qquad\text{with probability $1- k/n^d$}
	\end{cases},
	\]
	where the $x_\ff$ are independently chosen for the various $\ff\in[n]^d$.
	
	In other words, we assume a Bernoulli model for $\supp \wh x$, while the values at the frequencies that are chosen to be in the support are arbitrary.
\end{restatable}

Our algorithmic result for such signals is stated below.
\begin{restatable}[Sparse FFT algorithm for random support signals]{theorem}{sfftalg} \label{thrm:sparse-fft}
	Suppose $d$ is a positive integer and $n$ and $k$ are powers of two. For any signal $x \in \C^{n^d}$ such that $x$ is a random support signal of Fourier sparsity $k$, the procedure \textsc{SparseFFT} $(x, n, d, k)$  (see Algorithm~\ref{alg:FFThighD-k1/r}) returns $ \wh x$ with probability $9/10$. Moreover, the runtime and sample complexity of this procedure are $\widetilde{O}(k)$.
\end{restatable}

\begin{figure}
	\centering
	\scalebox{.7}{
		\begin{tikzpicture}[level/.style={sibling distance=60mm/#1}]
		\node [arn] (z){}
		child { node [arn] (a){}
			child {node [arn] (b) {}
				child {node [arn] (c) {}
					child {node [arn_l] (f) {}}
					child { node [arn_l] (g) {}}
				}
				child {node [arn_l] (h) {}}
			}
			child {node [arn] (i) {}
				child {node [arn_l] (newa) {}}
				child { node [arn_l] (newb) {}}
			}
		}
		child { node [arn] (j) {} 
			child {node [arn] (k) {}
				child {node [arn_l] (l) {}}
				child {node [arn] (m) {}
					child {node [arn] (n) {}
						child {node [arn_l] (o) {}}
						child {node [arn_l] (p) {}}
					}
					child {node [arn_l] (q) {}}
				}
			}
			child {node [arn] (r) {}
				child {node [arn_l] (s) {}}
				child {node [arn] (t) {}
					child {node [arn] (u) {}
						child {node [arn_l] (v) {}}
						child {node [arn_l] (w) {}}
					}
					child {node [arn_l] (x) {}
					}
				}
			}
		};

		\node []	at (z.north)	[label=left:]	{};
		\node []	at (a.north)	[label=left:]	{};
		\node []	at (b.north)	[label=left:]	{};
		\node []	at (c.north)	[label=left:]	{};
		\node []	at (d.north)	[label=left:]	{};
		\node []	at (f.north)	[label=below:1111]	{};
		\node []	at (g.north)	[label=below:0111]	{};
		\node []	at (h.north)	[label=below:011]	{};
		\node []	at (newa.north)	[label=below:101]	{};
		\node []	at (newb.north)	[label=below:001]	{};
		\node []	at (j.north)	[label=right:]	{};
		\node []	at (k.north)	[label=left:]	{};
		\node []	at (l.north)	[label=below:110]	{};
		\node []	at (m.north)	[label=left:]	{};
		\node []	at (n.north)	[label=left:]	{};
		\node []	at (o.north)	[label=left:11010]	{};
		\node []	at (p.north)	[label=below:01010]	{};
		\node []	at (q.north)	[label=left:0010]	{};
		\node []	at (r.north)	[label=left:]	{};
		\node []	at (s.north)	[label=below:100]	{};
		\node []	at (t.north)	[label=right:]	{};
		\node []	at (u.north)	[label=left:]	{};
		\node []	at (v.north)	[label=below:11000]	{};
		\node []	at (w.north)	[label=below:01000]	{};
		\node []	at (x.north)	[label=below:0000]	{};
		
		\node [] at (-8,1) [label=right:\Large{Splitting tree $T$ for a signal}]	{};
		\node [] at (-8,0.35) [label=right:{\Large{with random support}}]	{};		
		\draw[draw=black,very  thick, ->] (-4,0) -- (-1.5,-0.5);
		\end{tikzpicture}
	}
	\par
	\caption{An example of a splitting tree for a signal with uniformly random support (the nodes are labelled in binary). For simplicity, we truncate terminal rightward paths from leaves to the bottom level of the tree from the picture. The corresponding support set of this tree is $S=\{ 0,1,2,3,4,5,6,7,8,10,15,24,26 \}$}.
	\label{fig:split-rand}
\end{figure}
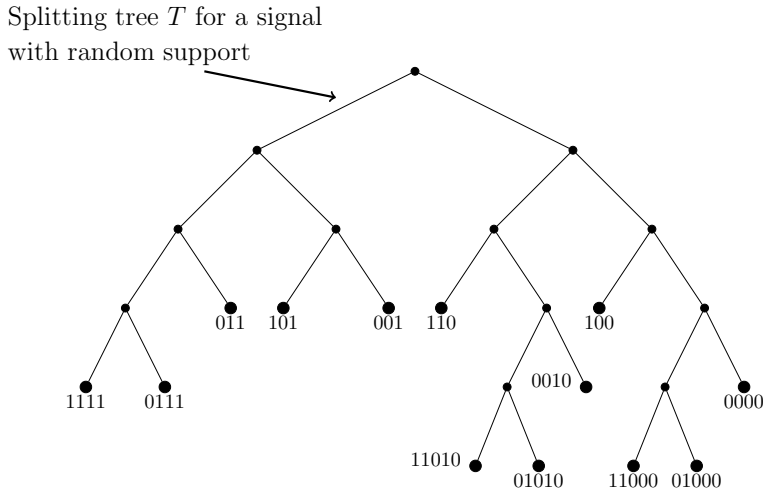

The algorithm is motivated by the idea of speeding up our algorithm for worst-case signals (Algorithm~\ref{alg:fullsparsefft}, also see Theorem~\ref{thm:sfft-worstcase}) by reducing the number of iterations of the process from $\Theta(k)$ down to $O(\log k)$. Such a reduction (which we show to be impossible for worst-case signals in Section~\ref{hard-instance}) requires the ability to peel off many elements of the residual in a single phase of the algorithm, which turns out to be possible if the support of $\wh{x}$ is chosen uniformly at random as in Definition~\ref{def:randsupport}. Indeed, if one considers the splitting tree $T$ of a signal with uniformly random support (see Fig.~\ref{fig:split-rand} for an illustration), one sees that 

\begin{description}
	\item [{\bf (a)}] a large constant fraction of nodes $v\in T$ satisfy $w_T(v)\leq \log_2 k+O(1)$;
	\item [{\bf (b)}] the adaptive aliasing filters $G$ constructed for such nodes will have significantly overlapping support in time domain.
\end{description}

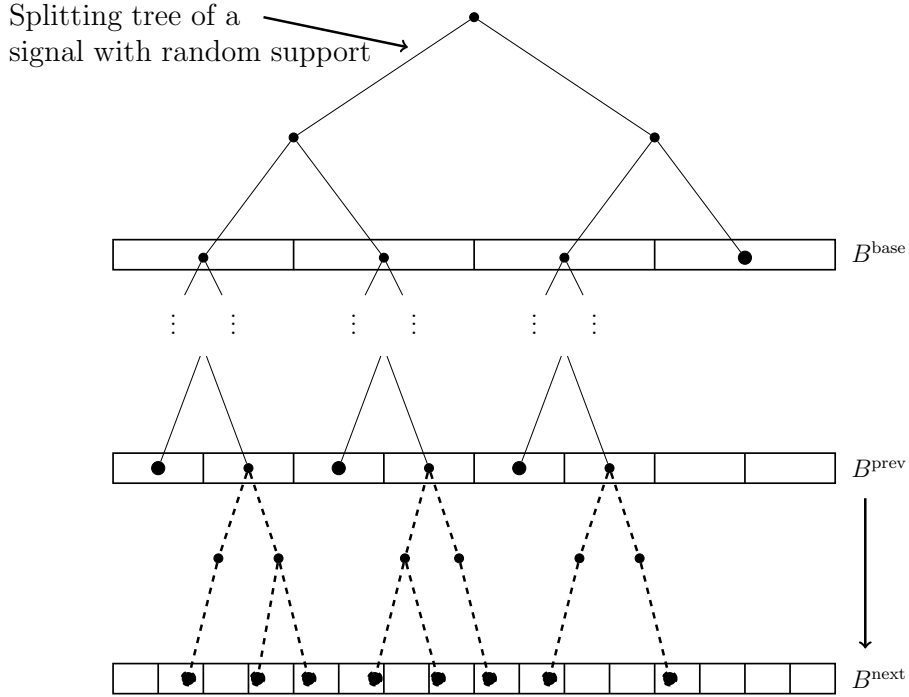
\begin{figure}
	\centering
	\scalebox{.8}{
		\begin{tikzpicture}[level/.style={sibling distance=60mm/#1,level distance = 2cm}]
		\node [arn] (z){}
		child {node [arn] {}
			child {node [arn] {}
				child [level distance=1cm,sibling distance=5mm]{node [] {$\vdots$}}
				child [draw=white,level distance=1.5cm]{node [] {}
					child [draw=black]{node [arn_l] {}}
					child [draw=black, sibling distance=15mm]{node [arn] {}
						child[level distance = 1.5cm,sibling distance=10mm]{node[arn] {}edge from parent [photon]
							child [sibling distance=10mm]{node [arn_l] {}edge from parent [photon]}			
							child [level distance=1cm, draw=white]{}
						}
						child[level distance = 1.5cm,sibling distance=10mm]{node[arn] {}edge from parent [photon] 
							child [sibling distance=7mm]{node [arn_l] {}edge from parent [photon]}			
							child [sibling distance=10mm]{node [arn_l] {}edge from parent [photon]}
						}
					}
				}
				child [level distance=1cm,sibling distance=5mm]{node [] {$\vdots$}}
			}
			child {node [arn] {}
				child [level distance=1cm,sibling distance=5mm]{node [] {$\vdots$}}
				child [draw=white,level distance=1.5cm]{node [] {}
					child [draw=black]{node [arn_l] {}}
					child [draw=black]{node [arn] {}
						child[level distance = 1.5cm,sibling distance=8mm]{node[arn] {}edge from parent [photon] 				
							child [sibling distance=10mm] {node [arn_l] {}edge from parent [photon]}
							child [sibling distance=11mm] {node [arn_l] {}edge from parent [photon]}
						}
						child[level distance = 1.5cm,sibling distance=10mm]{node[arn] {}edge from parent [photon]				
							child [draw=white,level distance = 1cm] {}
							child [sibling distance=10mm] {node [arn_l] {}edge from parent [photon]}
						}
					}
				}
				child [level distance=1cm,sibling distance=5mm]{node [] {$\vdots$}}
			}
		}
		child { node [arn] {} 
			child {node [arn] {}
				child [level distance=1cm,sibling distance=5mm]{node [] {$\vdots$}}
				child [draw=white,level distance=1.5cm]{node [] {}
					child [draw=black]{node [arn_l] {}
					}
					child [draw=black,sibling distance=15mm]{node [arn] {}
						child[level distance = 1.5cm,sibling distance=10mm]{node[arn] {}edge from parent [photon] 				
							child [sibling distance=10mm] {node [arn_l] {}edge from parent [photon]}
							child [sibling distance=5mm, level distance=1cm, draw=white] {}
						}
						child[level distance = 1.5cm,sibling distance=10mm]{node[arn] {}edge from parent [photon] 				
							child [sibling distance=5mm, level distance=1cm, draw=white] {}
							child [sibling distance=10mm] {node [arn_l] {}edge from parent [photon]}
						}
					}
				}
				child [level distance=1cm,sibling distance=5mm]{node [] {$\vdots$}}
			}
			child {node [arn_l] {}
			}
		};

		\draw[draw=black,thick] (-6,-4.2) rectangle ++(12,0.5);
		\foreach \x in {-3 , 0 , 3} \draw[draw=black,thick] (\x,-4.2) -- (\x,-3.7);
		\node []	at (6,-3.9)	[label=right:\large $\bbase$]	{};
		
		\draw[draw=black,thick] (-6,-7.75) rectangle ++(12,0.5);
		\foreach \x in {-4.5, -3 , -1.5, 0 , 1.5, 3, 4.5} \draw[draw=black,thick] (\x,-7.75) -- (\x,-7.25);
		\node []	at (6,-7.5)	[label=right:\large $\bprev$]	{};
		
		\draw[draw=black,thick] (-6,-11.25) rectangle ++(12,0.5);
		\foreach \x in { -5.25, -4.5, -3.75, -3 , -2.25, -1.5, -0.75, 0 , 0.75, 1.5, 2.25, 3, 3.75, 4.5, 5.25 } \draw[draw=black,thick] (\x,-11.25) -- (\x,-10.75);
		\node []	at (6,-11)	[label=right:\large $\bnext$]	{};
		
		\draw[draw=black,very  thick, ->] (6.5,-8.0) -- (6.5,-10.5);

		\node [] at (-8,0) [label=right:\Large{Splitting tree of a}]	{};
		\node [] at (-8,-0.6) [label=right:\Large{signal with random support}]	{};
		\draw[draw=black,very  thick, ->] (-3.5,0) -- (-1.1,-0.5);
		\end{tikzpicture}
	}
	\par
	
	\caption{The illustration of the transition from $\bprev$-bucketing to $\bnext$-bucketing on the  splitting tree of a signal with random support.}\label{fig:random-loc}
\end{figure}

We provide the intuition for this for the one-dimensional setting ($d=1$) to simplify notation (changes required in higher dimensions are minor). In this setting, property {\bf (b)} above is simply a manifestation of the fact that since the support is uniformly random, any given congruence classes modulo $B'=Ck$ for a large enough constant $C>1$ is likely to contain only a single element of the support of $\wh{x}$. Our adaptive aliasing filters provide a way to only partition frequency space along a carefully selected subset of bits in $[\log_2 N]$, but due to the randomness assumption, one can isolate most of the elements by simply partitioning using the bottom $\log_2 k+O(1)$ bits. This essentially corresponds to hashing $\wh{x}$ into $B=Ck$ buckets at computational cost $O(B'\log B')=O(k\log k)$. While this scheme is efficient, it unfortunately only recovers a constant fraction of coefficients. One solution would be to hash into $B=Ck^2$ buckets (i.e., consider congruence classes modulo $Ck^2$), which would result in perfect hashing with good constant probability, allowing us to recover the entire signal in a single round. However, this hashing scheme would result in a runtime of $\Omega(k^2\log k)$ and is, hence, not satisfactory. On the other hand, hashing into $Ck^2$ buckets is clearly wasteful, as most buckets would be empty. Our main algorithmic contribution is a way of ``implicitly'' hashing into $Ck^2$ buckets, i.e., getting access to the nonempty buckets, at an improved cost of $\widetilde O(k)$.

Our algorithm uses an iterative approach, and the main underlying observation is very simple. Suppose that we are given the ability to ``implicitly'' hash into $B$ buckets for some $B$, namely, get access to the nonempty buckets. If $B$ is at least $\text{min}(Ck^2, N)$, we know that there are no collisions with high probability and we are done. If not, then we show that, given access to nonempty buckets in the $B$-hashing (i.e. a hashing into $B$ buckets), we can get access to the nonempty buckets of a $(\Gamma B)$-hashing for some appropriately chosen constant $\Gamma > 1$ at a polylogarithmic cost in the size of each nonempty bucket of the $B$-bucketing by essentially computing the Fourier transform of the signal restricted to nonempty buckets in the $B$-bucketing. We then proceed iteratively in this manner, starting with $B=Ck$, for which we can perform the hashing explicitly. Since the number of nonzero frequencies remaining in the residual after $t$ iterations of this process decays geometrically in $t$, we can also afford to use a smaller number of buckets $B'$ in the hashing that we construct explicitly, ensuring that the runtime is dominated by the first iteration.

Ultimately, the algorithm takes the following form. At every iteration, we explicitly compute  a hashing into $\bbase\leq Ck$ buckets explicitly. Then, using a list of nonempty buckets in a $\bprev$-hashing from the previous iteration, we extend this list to a list of nonempty buckets in a $\bnext$-bucketing at polylogarithmic cost per bucket (by solving a well-conditioned linear system, see Algorithm~\ref{alg:hashing2}), where $\bnext=\Gamma\cdot \bprev$ for some large enough constant $\Gamma>1$. Meanwhile, we reduce $\bbase$ by a factor of $\Gamma$, thus maintaining the invariant $\bbase\cdot \bnext \approx k^2$ at all times (note that this is satisfied at the start, when $\bbase =\bprev\approx k$, and $\bbase\cdot\bnext$ remains invariant at each iteration). Therefore, after a logarithmic number of iterations, we have effectively emulated hashing into $\approx k^2$ but at a total cost of roughly one hashing computation into $\approx k$ buckets (see Figure~\ref{fig:random-loc} for an illustration).

\paragraph{Organization.}
In Section~\ref{sec:prelim}, we introduce basic definitions and notation that will be used throughout the paper.  
Section~\ref{sec:filters} introduces our main technical tool of adaptive aliasing filter, which are used in the various algorithms found in this paper. Section~\ref{sec:worstcase1d} shows how to use the adaptive aliasing filters to solve the problem of estimation for Fourier measurements for worst-case signals, i.e., problem \eqref{eq:exact-est}, thereby proving Theorem~\ref{thm:intro-est-main}. Section~\ref{hard-instance} then shows that the inherent tree pruning process used to subtract off recovered frequencies and access residual signals in the estimation algorithm is essentially optimal.

Section~\ref{sec:randomsign} proves our main theorem, Theorem~\ref{thm:sfft-worstcase}, for problem \eqref{eq:exact-sfft} on worst-case signals. Additionally, it shows how to improve on the runtime under the assumption that the signal is a worst-case signal with random phase, thereby proving Theorem~\ref{thm:sfft-rand-signs}.

Finally, Section~\ref{sec:avgcase} discusses how to obtain an algorithm for problem \eqref{eq:exact-sfft} on random support signals and proves Theorem~\ref{thrm:sparse-fft}.

\newpage

\section{Preliminaries and notation} \label{sec:prelim}
In this section, we introduce some notation and basic definitions that we will use in the paper.

For any positive integer $n$, we use the notation $[n]$ to denote the set of integer numbers $\{0, 1, \dots, n-1 \}$. We are interested in computing the Fourier transform of discrete signals of size $N$ in dimension $d$, where $N=n^d$ for some $n \ge 2$. Such a signal will be a function $[n]^d \to \C$. However, we will often identify $[n]^d \to \C$ with $\C^{n^d}$ for convenience (and often use the two interchangably depending on the context). This correspondence is formally defined later in Definition~\ref{def:flattening}. We first need the notion of an inner product.
\begin{definition}[Inner product]
	Let $\tv$ and $\ff$ be two vectors in dimension $d$. We denote the inner product of $\tv$ and $\ff$  by $\ff^T \tv = \sum_{q=1}^d f_qt_q$.
\end{definition}
Let us define the \emph{Fourier transform} of a signal.
\begin{definition}[Fourier transform] \label{def:ft}
	For any positive integers $d$ and $n$, the \emph{Fourier transform} of a signal $x \in \C^{n^d}$ is denoted by $\wh x$, where for any $\ff \in [n]^d$, we define
	$\wh x_{\ff} = \sum_{\tt \in [n]^d} x_{\tt} e^{-2\pi i \frac{\ff^T \tt}{n}}$.
\end{definition}
Note that in the case of $n=2$, the Fourier transform reduces to the Hadamard transform of size $N=2^d$.

\begin{claim}[Parseval's theorem]\label{parseval}
	For any positive integers $n$ and $d$, any signal $x\in  \C^{n^d}$ satisfies
	$\|\wh x\|_2^2 = n^d \cdot \| x\|_2^2$.
\end{claim}

\begin{definition}[Unit impulse]
	For any positive integers $n$ and $d$, the \emph{unit impulse function} $\delta \in \C^{n^d}$ is defined as the function given by $\delta(\tv) = 1$ for $\tv = 0$ and $\delta(\tv) = 0$ for $\tv\neq 0$.
\end{definition}

\begin{claim}\label{claim:delta-fft}
	For any positive integers $d$, $n$, and any $\aaa \in [n]^d$, the inverse Fourier transform of $\wh x: [n]^d\to \C$ given by $\wh x_\ff = e^{2\pi i \frac{\aaa^T \ff}{n}}$ is
	$x_\tv = \delta(\tv+\aaa)$.
\end{claim}

\begin{claim}[Convolution theorem]\label{conv-thrm}
	Suppose $d$ and $n$ are positive integers. Then, for any signals $x , y \in \C^{n^d}$,
	$ \wh{(x * y )} = \wh x \cdot \wh y$,
	where $x * y$ is the convolution of $x$ and $y$ which itself is a signal in $\C^{n^d}$ defined as,
	$(x * y )_{\tv} = \sum_{\bm{\tau} \in [n]^d} x_{\bm{\tau}} y_{\tv - \bm\tau}$ for all $\tv \in [n]^d$.
\end{claim}

We will require the notion of a \emph{tensor product} of signals. Given $d$ signals $G_1, G_2, \dots, G_d: [n] \to \C$, the tensor product constructs a signal in $\C^{n^d}$ that is defined as follows.
\begin{definition}[Tensor multiplication] \label{def:tensor}
	Suppose $d$ and $n$ are positive integers. Given functions $G_1 , G_2, \dots, G_d: [n]\to \C$, we define the \emph{tensor product} $(G_1 \times G_2 \times \cdots \times G_d) : [n]^d \to \C$ as 
	$ \left( G_1 \times G_2 \times \cdots \times G_d \right){(\jj)} = G_1({j_1}) \cdot G_2({j_2}) \cdots G_d({j_d})$ for all $\jj = (j_1, j_2, \dots, j_d) \in [n]^d$.
	
\end{definition}
Note that the tensor product is essentially a generalization of the usual outer product on two vectors to $d$ vectors.
\begin{claim}[Fourier transform of a tensor product]\label{claim:tensor}
	For any integers $n$ and $d$ and $G_1 , G_2, \dots, G_d \in \C^n$, let $G: [n]^d \to \C$ denote the tensor product $G = G_1 \times G_2 \times \cdots \times G_d$. Then, the $d$-dimensional Fourier transform $\wh{G}$ of $G$ is the tensor product of $\wh G_1, \wh G_2, \cdots, \wh G_d $, i.e.,
	$\wh G = \wh G_1 \times \wh G_2 \times \cdots \times \wh G_d$.
\end{claim}

\begin{definition}
	For any positive $d$, $n$, and $k$, a signal $x: [n]^d \to \C$ is called \emph{Fourier $k$-sparse} if $\|\wh x\|_0 = k$.
\end{definition}

\begin{definition}[The Restricted isometry property] \label{def:rip}
	We say that a matrix $A \in \C^{q \times n}$ satisfies the \emph{restricted isometry property (RIP)} of \emph{order} $k$ if for every $k$-sparse vector $x \in \C^n$, i.e., $\|x\|_0 \le k$, it holds that
	$  \frac{1}{2} \|x\|^2_2 \le \|Ax\|_2^2 \le \frac{3}{2} \|x\|^2_2$.
\end{definition}

We will use the following theorem from \cite{haviv2017restricted}.
\begin{theorem} \emph{(The Restricted Isometry Property \cite[Theorem 3.7]{haviv2017restricted})}\label{RIP-thrm}
	For sufficiently large $N$ and $k$, and a unitary matrix $M \in \C^{N\times N}$ satisfying $\|M\|_\infty = O\left( \frac{1}{\sqrt{N}} \right)$, the following holds. For some $q = O\left( k \log^2 k \log N \right)$ let $A \in \C^{q \times N}$ be a matrix whose $q$ rows are chosen uniformly and independently from the rows of $M$, multiplied by $\sqrt{\frac{N}{q}}$. Then, with probability $1 - \frac{1}{N^{10}}$, the matrix A satisfies the restricted
	isometry property of order $k$, as per Definition \ref{def:rip}.
\end{theorem}

\section{Adaptive aliasing filters}\label{sec:filters}
In this section, we introduce a new class of filters that forms the basis of our algorithm for estimation of worst case Fourier sparse signals. For simplicity, we begin by introducing the filters in the one-dimensional setting and then show how they naturally extend to the multidimensional setting (using tensoring). Throughout the section, we assume that the input is a signal $x \in \C^n$ with $\supp \wh{x}=S$ for some $S\subseteq [n]$.

\subsection{One-dimensional Fourier transform}\label{sec:filters-1d}
We restate the following definition for $\tfull_n$ and corresponding labels of vertices:
\deftfull*

Next, we recall the definition of the \emph{splitting tree} of a set.
\defsplit*

The splitting tree $T=\tree(S, n)$ can be constructed easily in $O(|S| \log n)$ time, given $S$. We provide simple pseudocode in Algorithm~\ref{alg:tree-construction}.

For every node $v \in T$, the \emph{level} of $v$, denoted by $l_T(v)$, is the distance from $v$ to the root. The following basic claim will be useful and follows immediately from the definition of $T=\tree(S, n)$:
\begin{claim}\label{cl:t}
	For every integer power of two $n$, if $T$ is a subtree of $\tfull$, then for every node $v\in T$, the labels of nodes that belong to the subtree $T_v$ of $T$ rooted at $v$ are 
	congruent to $f_v$ modulo $2^{l_T(v)}$. Furthermore, every node $u\in T$ at level $l_T(v)$ or higher which satisfies $f_u\equiv f_v \pmod{2^{l_T(v)}}$ belongs to $T_v$. 
\end{claim}

\defwt*

\begin{definition}[$(f, S)$-isolating filter]\label{def:f-s-isolating}
	For every power of two $n$, set $S \subseteq [n]$, and $f\in S$, a filter $G\in \C^n$ is called \emph{$(f,S)$-isolating} if $\wh{G}_f=1$, and $\wh{G}_{f'}=0$ for all $f'\in S\setminus \{f\}$. 
	
	In particular, if $G$ is $(f, S)$-isolating, then for every signal $x \in \C^n$ with $\supp\wh{x} \subseteq S$, we have
	\begin{align*}
	\sum_{j\in[n]} x_j G_{t-j} &= (x*G)_{t}\\
	&= \frac{1}{n}\sum_{f \in [n]} \wh x_f \cdot \wh G_f \cdot e^{2\pi i \frac{ft}{n}}\\
	&= \frac{1}{n} \wh{x}_f e^{2\pi i \frac{f t}{n}} 
	\end{align*}
	for all $t\in[n]$, by convolution theorem, see Claim \ref{conv-thrm}.
\end{definition}

While the definitions above suffice to state our estimation primitive, our Sparse FFT algorithm requires a filter $G$ that satisfies a more refined property due to the fact that throughout the execution of the algorithm, the identity of $\supp \wh{x}$ is only partially known.  We encode this knowledge as a subtree $T$ of $\tfull_n$ whose leaves are not necessarily at level $\log_2 n$. Hence, every leaf $v\in T$ corresponds to a set of frequencies in the support of $\wh{x}$ whose full identities have not been discovered yet. This is captured by the following definition:

\deffreqcone*

Note that under this definition, the frequency cone of a vertex $v$ of $T$ corresponds to the subtree rooted at $v$ when $T$ is embedded inside $\tfull_n$ (see Figure~\ref{fig:conjclass}).
\defvtisol*
Note that in particular, for all signals $x \in \C^n$ with $\supp\wh{x} \subseteq \bigcup_{u: \text{~leaf of~}T} \subtree_T(u)$ and $t\in[n]$,
\[ \sum_{j\in[n]} x_j G_{t-j} = \frac{1}{n} \sum_{f \in \subtree_T(v)}\wh{x}_f e^{2\pi i \frac{f t}{n}}. \]

\begin{algorithm}
	\caption{Filter construction in time and Fourier domain}
	\begin{algorithmic}[1]
		
		\Procedure{FilterPreProcess}{$T, v, n$}
		
		\State $r\gets$ root of $T$, $l\gets l_T(v)$, $f\gets f_v$
		\State $v_0,v_1,\ldots, v_l\gets $ path from $r$ to $v$ in $T$, where $v_0 = r$ and $v_l = v$
		\State $\mathbf{ g} \gets \{0\}^{\log_2n}$
		\For{\texttt{$j = 1$ to $l$ }}
		\If{$v_{j-1}$ has two children in $T$}
		\State $g_j \gets e^{-2\pi i \frac{f}{2^j}}$
		
		\EndIf
		
		\EndFor
		\State \textbf{return} $ \bf g$
		\EndProcedure

		\Procedure{FilterTime}{$\mathbf{ g}, n$}
		
		\State $G(t) \gets \delta(t) \text{ for all }t\in[n]$

		\For{\texttt{$l = 1$ to $\log_2n$ }}

		\If{\texttt{$g_l \neq 0$}}
		\State $G(t) \gets \frac{ G(t)}{2} + g_l \cdot  \frac{G(t+n/2^l) }{2} \text{ for all }t\in[n]$

		\EndIf
		
		\EndFor
		\State \textbf{return} $ G$
		\EndProcedure
		
		\Procedure{FilterFrequency}{$\mathbf{g}, n, \xi$}
		
		\State $\wh{G}_{\xi} \gets 1$

		\For{\texttt{$l=1$ to $\log_2n$}}
		
		\If{\texttt{$g_l \neq 0$}}
		\State $\wh{G}_\xi \gets \wh{G}_\xi \cdot {(1 + g_l \cdot e^{2\pi i \frac{\xi}{2^l}} )}\big/ {2}$
		
		\EndIf

		\EndFor
		\State \textbf{return} $ \wh{G}_\xi$
		\EndProcedure
	\end{algorithmic}
	\label{alg:filter}
	
\end{algorithm}

\begin{lemma}[Filter properties]\label{lem:filter-isolate}
	
	For every power of two $n$, subtree $T$ of $\tfull_n$, and leaf $v \in T$, the procedure \textsc{FilterPreProcess}$(T,v,n)$ outputs a static data structure $\mathbf{g} \in \C^{\log_2n}$ in time $O(\log_2n)$ such that, given $\mathbf{g}$, the following conditions hold:
	
	\begin{enumerate}
		\item The primitive \textsc{FilterTime}$(\mathbf{g}, n)$ outputs a filter $G$ such that $| \supp{G} | = 2^{w_T(v)}$ and $G$ is a $(v, T)$-isolating filter. Moreover, the procedure runs in time $O(2^{w_T(v)} + \log_2n)$. \label{tilter-time}
		
		\item For every $\xi \in [n]$, the primitive \textsc{FilterFrequency}$(\mathbf{ g},n,\xi)$ computes the Fourier transform of $G$ at frequency $\xi$, namely, $\wh G(\xi)$, in time $O(\log_2n)$. \label{filter-freq}
	\end{enumerate}

	
\end{lemma}
Before we prove Lemma~\ref{lem:filter-isolate}, we establish the following corollary, assuming the statement of Lemma~\ref{lem:filter-isolate} holds.
\begin{corollary}
	Suppose $n$ is a power of two, $S \subseteq [n]$, and $f\in S$. Then, let $T=\tree(S, n)$ be the splitting tree of $S$. If $v$ is the leaf of $T$ with label $f_v=f$, while $\mathbf{g}$ is the output of \textsc{FilterPreProcess}$(T,v,n)$, and $G$ is the filter computed by \textsc{FilterTime}$(\mathbf{g}, n)$, then the following conditions hold:
	
	\begin{description}
		\item[(1)] $G$ is an $(f, S)$-isolating filter.
		\item[(2)] $|\supp G| = 2^{w_T(v)}$.
	\end{description}
\end{corollary}

\begin{proof}
	Indeed, given a subset $S$ and $f\in S$, if $T=\tree(S, n)$, then all the leaves of $T$ are at level $\log_2 n$ and the set of labels of the leaves is exactly $S$. Hence, for every leaf $v$ of $T$, one has $\subtree_T(v)=\{f_v\}$. By Lemma~\ref{lem:filter-isolate}, $G$ is a $(v, T)$-isolating filter. Therefore, by Definition \ref{def:v-t-isolating},
	$$
	\emptyset = \supp \wh{G} \cap \left(\bigcup_{\substack{u\neq v \\ u:\text{~leaf of~}T}} \subtree_T(u)\right) = \supp \wh{G} \cap \left(\bigcup_{\substack{u\neq v \\ u:\text{~leaf of~}T}} \{f_u\}\right) = \supp \wh{G}\cap (S\setminus f_v),
	$$
	and $\wh{G}(f) = 1$ for all $f \in \subtree_T(v)=\{f_v\}$. This implies {\bf (1)}, see definition of $(f,S)$-isolating filters in \ref{def:f-s-isolating}. Property {\bf (2)} follows directly from Lemma~\ref{lem:filter-isolate}.
\end{proof}

Now, we prove Lemma~\ref{lem:filter-isolate}. 

\bigskip
\begin{proofof}{Lemma~\ref{lem:filter-isolate}}
	Let $v$ be a leaf of $T$, $l=l_T(v)$ denote the level of $v$ (i.e., distance from the root), $r$ denote the root of $T$, and $v_0, v_1,\ldots, v_l$ denote the path from root to $v$ in $T$, where $v_0 = r$ and $v_l = v$. 
	
	We first show how to efficiently construct a $(v,T)$-isolating filter in the \emph{Fourier} domain, i.e., how to efficiently construct $\wh{G}$. Then we derive the time domain representation of $G$.  We iteratively define a sequence of functions $G_0, G_1, \dots, G_l$ (with Fourier transforms $\wh{G}_0, \wh{G}_1, \dots, \wh{G}_l$, respectively) by traversing the path from the root to $v$ in $T$, after which we let $G$ be the final filter constructed on this path, i.e., $G := G_l$ (and $\wh{G} := \wh{G}_l$). We start with $\wh{G}_0(\xi) = 1$ for all $\xi\in[n]$. Then, we iteratively define 
	$\wh{G}_q$ in terms of $\wh{G}_{q-1}$ according to the following update rule for all $q =1,2,\ldots, l$: 
	\begin{equation}\label{eq:g-def}
	\wh{G}_q(\xi)=
	\begin{cases}
	\wh{G}_{q-1}(\xi) \cdot \frac{1+ e^{2\pi i \frac{\xi - f_v}{2^q}}}{2} &\qquad\text{if~$v_{q-1}$ has two children in $T$}\\
	\wh{G}_{q-1}(\xi) &\qquad\text{otherwise}
	\end{cases}.
	\end{equation}
	for every $\xi\in [n]$.
	
	We now show that $G = {G}_l$ is a $(v,T)$-isolating filter. It is enough to show that $G$ satisfies
	\begin{equation}\label{eq:gsupp}
	(\supp \wh{G} )\cap \left(\bigcup_{\substack{u\neq v \\ u:\text{~leaf of~}T}} \subtree_T(u)\right) = \emptyset 
	\end{equation}
	and
	\begin{equation}
	\wh{G}(f) = 1 \text{ for all } f \in \subtree_T(v). \label{eq:gval}
	\end{equation}
	
	We now prove \eqref{eq:gsupp}. Consider a leaf $u$ of $T$ distinct from $v$. Recall that $v_0, v_1,\dots, v_l$ denotes the root to $v$ path in $T$. Let $j$ be the largest integer such that $v_j$ is a common ancestor of $v$ and $u$. 
	
	By definition of tree $T$ (Definition~\ref{def:t-full}) and because $v_j$ is at level $j$, one has that the label of the right child $a$ of $v_j$ is $f_{v_j}$, and the label of the left child $b$ is $f_{v_j}+2^j$. Furthermore, using this together with Claim~\ref{cl:t}, we get that the labels of nodes in subtree $T_a$ of $T$ subtended at the right child $a$ of $v$ are 
	congruent to $f_a = f_{v_j}$ modulo $2^{j+1}$, and 
	labels in the subtree $T_b$ rooted at the left child $b$ of $v_j$ are all congruent to $f_b = f_{v_j}+2^j$ modulo $2^{j+1}$.
	
	Suppose that $v$ belongs to the right subtree of $v_j$, and $u$ belongs to the left subtree (the other case is symmetric). We thus get that $f_v\equiv f_{v_j} \pmod{2^{j+1}}$, and $f_u \equiv f_{v_j}+2^j \pmod{2^{j+1}}$. It now suffices to note that by construction of $\wh{G}$ (see~\eqref{eq:g-def}), we have that for all $\xi\in [n]$,
	$$
	\wh{G}_{j+1}(\xi)=\wh{G}_{j}(\xi) \cdot \frac{1+ e^{2\pi i \frac{\xi - f_v}{2^{j+1}}}}{2}.
	$$
	By Claim \ref{cl:t}, for all $f \in \subtree_T(u)$ one has that $f \equiv f_u \pmod{2^{l_T(u)}}$ and hence, $f \equiv f_u \pmod{2^{j+1}}$ because $j+1 \le l_T(u)$. Therefore, by substituting $\xi=f$ in the above, we get 
	$$
	\wh{G}_{j+1}(f) = \wh{G}_{j}(f) \cdot \frac{1+ e^{2\pi i \frac{f - f_v}{2^{j+1}}}}{2} = \wh{G}_{j}(f) \cdot \frac{1+ e^{2\pi i \frac{f_u - f_v}{2^{j+1}}}}{2}=0,
	$$
	implying that $\wh{G}_{j+1}(f)=0$ and, hence, $\wh{G}_l(f)=0$, as required.
	
	It remains to prove \eqref{eq:gval}. Consider any $f'\in \subtree_T(v)$, and note that by Claim~\ref{cl:t}, $f'\equiv f_v \pmod{2^{l}}$. Using this in~\eqref{eq:g-def}, we get
	\begin{equation*}
	\begin{split}
	\wh{G}(f')&=\prod_{\substack{q\in \{1, 2, \ldots, l\}\\v_{q-1}\text{~has two children in~}T}} \frac{1+ e^{2\pi i \frac{f' - f_v}{2^q}}}{2}\\
	&=1,\\
	\end{split}
	\end{equation*}
	since $f' -f_v \equiv 0 \pmod{2^q}$ for every $q=0,\ldots, l$.
	
	Next, note that the primitive \textsc{FilterPreProcess}($T,v, n$) preprocesses the tree $T$ by traversing the path from root to leaf $v$ in time $O(\log_2n)$.  Given $\mathbf{g}$, the primitive \textsc{FilterFrequency} $(\mathbf{ g},n,\xi)$ implements ~\eqref{eq:g-def} for successive values of $q$, and the runtime of this algorithm is $O(\log_2n)$ because of the \emph{for} loop passing through vector $\mathbf{g}$.
	
	Finally, it remains to show that the filter $G$ in \emph{time domain} can be computed efficiently and has a small support. First note that by Claim~\ref{claim:delta-fft},
	the inverse Fourier transform of $ \frac{1+ e^{2\pi i \frac{\xi - f_v}{2^q}}}{2}$ is $\frac{\delta(t)+e^{-2\pi i f_v/2^q}\delta\left(t+\frac{n}{2^q}\right)}{2}$. 
	
	By the duality of convolution in the time domain and multiplication in Fourier domain (see Claim~\ref{conv-thrm}), we can equivalently define $G$ (see~\eqref{eq:g-def}) by letting $G_0(t)=\delta(t)$ and setting
	\begin{equation}\label{eq:g-def-time}
	G_q(\xi)=
	\begin{cases}
	G_{q-1}(t) * \frac{\delta(t)+e^{-2\pi i {f_v}/{2^q}}  \delta\left(t+\frac{n}{2^q}\right)}{2} &\qquad \text{if~$v_{q-1}$ has two children in $T$}\\
	G_{q-1}(t) &\qquad \text{otherwise}
	\end{cases}
	\end{equation}
	for every $q=1,\dots, l$. Thus, $G=G_l$ is the time domain representation of the filter $\wh{G}$ defined in~\eqref{eq:g-def}. We now note that convolving any function with a function supported on two points, e.g., $\frac{1}{2}\left(\delta(t)+e^{-2\pi i {f_v}/{2^q}}  \delta(t+\frac{n}{2^q})\right)$, at most doubles the support. Since the number of times the convolution is performed in obtaining $G_l$ from $G_0$ (as per~\eqref{eq:g-def-time}) is $w_T(v)$, the support size of $G$ is at most $2^{w_T(v)}$. Given $\mathbf{ g}$, the primitive \textsc{FilterTime} $(\mathbf{g},n)$ implements the above algorithm for construction of $G$ and, therefore, runs in time $O(2^{w_T(v)} + \log_2n)$.
\end{proofof}

\subsection{$d$-dimensional Fourier transform}\label{sec:filters-d}

In this section, we show that our construction of adaptive aliasing filters from the previous section naturally extends to higher dimensions without any loss by tensoring.  

\begin{definition}[Flattening of ${[n]}^d$ to ${[{n}^d]}$. Unflattening of ${[n^d]}$ to ${[{n}]^d}$] \label{def:flattening}
	For every power of two $n$, positive integer $d$,  and $\ff=(f_1,\ldots, f_d) \in [n]^d$ we define the {\em flattening} of $\ff$ as 
	$$
	\wt{\ff} = \sum_{r=1}^{d} f_{r} \cdot n^{r-1}. 
	$$
	Similarly, for a subset $S\subseteq [n]^d$ we let $\wt{S}:=\{\wt{\ff}: \ff\in S\}$ denote the flattening of $S$.
	
	For $\wt{\bm\xi} \in [n^d]$, we define the {\em unflattening} of $\wt{\bm\xi}$ as $\bm\xi=(\xi_1,\dots, \xi_d) \in [n]^d$, where
	$$
	\xi_{q} = \frac{\wt{\bm\xi} - \wt{\bm\xi}\pmod{n^{q-1}}}{n^{q-1}} \pmod{n}.
	$$
	for every $q=1,\dots,d$. Similarly, for a subset $\wt{R}\subseteq [n^d]$, we let ${R}:=\{{\bm\xi} \in [n]^d: \wt{\bm\xi}\in \wt{R}\}$ denote the unflattening of $\wt{R}$.
\end{definition}
\begin{definition}[Multidimensional splitting tree] \label{splittree-weight-highdim}
	Suppose $d$ is a positive integer and $n$ is a power of two. For every $S \subseteq [n]^d$, the \emph{flattened splitting tree} of $S$ is defined as $\wt{T}= \tree(\wt{S},n^d)$ where $\wt{S}$ is flattening of $S$. 
	
	The unflattened splitting tree of $S$ is denoted by $T$ and is obtained from the flattened splitting tree $\wt{T}$ by unflattening the labels $\wt{\ff}_v$ of all nodes $v \in \wt{T}$. 
\end{definition}



\begin{definition}[Multidimensional $(\ff, S)$-isolating filter] \label{IsolatingFilter-highdim}
	Suppose $n$ is a power of two integer and $S \subseteq [n]^d$ for a positive integer $d$.  Then, for any frequency $\ff\in S$, a filter $G: [n]^d \to \C$ is called \emph{$(\ff,S)$-isolating} if $\wh{G}_\ff=1$ and $\wh{G}_{\ff'}=0$ for every $\ff'\in S\setminus \{\ff\}$. 
\end{definition}

\begin{definition}[Frequency cone of a leaf of $T$ in high dimensions]\label{def:iso-t-highdim}
	Suppose $d$ is a positive integer, $n$ is a power of two, and $N = n^d$. For every unflattened subtree $T$ of $\tfull_N$ and $v\in T$, we define the {\em frequency cone of $v$} as
	$$
	\subtree_T(v):=\left\{ \ff \in [n]^d: \wt{\ff} \equiv \wt{\ff}_v \pmod{2^{l_T(v)}} \right\},
	$$
	where $l_T(v)$ denotes the level of $v$ in $T$ (i.e., the distance from the root).
\end{definition}

\begin{claim} \label{claim:sub-tree}
	For every positive integer $d$, power of two $n$, and every subtree $T$ of $\tfull_{n^d}$ and every leaf $v \in T$ of height $l_T(v) < d\log_2n$, let $T' = T \cup \{\text{left child $u$ of $v$}\}\cup \{\text{right child $w$ of $v$}\}$. Then the following holds,
	$$\subtree_T(v)=\subtree_{T'}(u)\cup \subtree_{T'}(w)$$
\end{claim}

\begin{definition}[Multidimensional $(v, T)$-isolating filter]\label{def:v-t-isolating-highdim}
	Suppose $d$ is a positive integer, $n$ is a power of two, and $N = n^d$. For every subtree $T$ of $\tfull_N$ and vertex $v\in T$, a filter $G\in \C^{n^d}$ is called \emph{$(v,T)$-isolating} if $\wh{G}_{\ff}=1$ for all $\ff\in \subtree_T(v)$ and for every $\ff'\in \bigcup_{\substack{u \neq v \\ u: \text{~leaf of~}T}} \subtree_T(u)$ one has $\wh{G}_{\ff'}=0$.
	
	In particular, for every signal $x \in \C^{n^d}$ with $\supp\wh{x} \subseteq \bigcup_{u: \text{~leaf of~}T} \subtree_T(u)$ and for all $\tv\in[n]^d$,
	\[ \sum_{\jj\in[n]^d} x_{\jj} G_{\tv-\jj} = \frac{1}{N} \sum_{\ff \in \subtree_T(v)}\wh{x}_{\ff} e^{2\pi i \frac{\ff^T \tv}{n}}. \]
\end{definition}

\begin{lemma}[Construction of a multidimensional isolating filter]\label{lem:isolate-filter-highdim}
	Suppose $n$ is a power of two integer and $d$ is a positive integer. Let $N=n^d$. For every subtree $T$ of $T_{N}^{full}$ and every leaf $v \in T$, there exists a $(v,T)$-isolating filter $G$ such that $|\supp G| = 2^{w_T(v)}$. Such a filter $G$ can be constructed in time $O(2^{w_T(v)} + \log_2N)$. Moreover, for any frequency $\bm\xi \in [n]^d$, the Fourier transform of $G$ at frequency $\bm\xi$, i.e., $\wh G(\bm\xi)$, can be computed in time $O( \log_2N)$.
\end{lemma}
The proof of Lemma~\ref{lem:isolate-filter-highdim} appears in Appendix \ref{appx:A}. The key idea is to choose $q^*$ to be the smallest positive integer such that $l_T(v)\leq q^* \cdot \log_2 n$. One then defines successive filters $G^{(0)}, G^{(1)}, \dots, G^{(q^*)}$ by letting $\wh G^{(0)} = 1$ and 
\[
\wh{G}^{(q)}(\ff)=\wh{G}^{(q-1)}(\ff)\cdot \wh{G}_q(f_q)
\]
for $q = 1,2,\dots, q^*$, where $\wh G_q$ is an isolating filter corresponding to the projection of the leaves of tree $T$ into coordinate $q$. The final filter $G = G^{(q^*)}$ turns out to be $(v,T)$-isolating.

\subsection{Putting it together}

\begin{claim}\label{kraftsum}
	For any binary tree $T$ let $L$ be the set of leaves of $T$. There exists a leaf $v\in L$ such that $w_T(v) \leq \log_2 |L|$.
\end{claim}
\begin{proof}
	Let $T'$ be the tree obtained by ``collapsing'' $T$, i.e., removing all nodes (and incident edges) of $T$ that have exactly one child. Then, observe that the leaves of $T$ are still preserved in $T'$, except that they are at possibly varying levels. In particular, a leaf $v$ in $T'$ will be at level $w_T(v)$. Thus, by applying Kraft's inequality to $T'$ (which is an equality because every node in $T'$ is either a leaf or has two children), we see that
	\[
	\sum_{v\in L} 2^{-w_T(v)} = 1.
	\]
	Therefore, there exists a $v\in L$ such that $2^{-w_T(v)} \geq \frac{1}{|L|}$ and, therefore, $w_T(v) \le \log_2|L|$, as desired.
\end{proof}

This gives us the main result of this section, and the main technical lemma of the paper:
\begin{corollary} \label{cor:isofilter}
	For every integer $n\geq 1$ a power of two and every positive integer $d$, every $S \subseteq [n]^d$, there exists an $\ff \in S$ and an $(\ff,S)$-isolating filter $G$ (as defined in Definition~\ref{IsolatingFilter-highdim}) such that $|\supp G| \leq |S|$. 
\end{corollary}
\begin{proof}
	Follows by combining Lemma~\ref{lem:isolate-filter-highdim} with Claim~\ref{kraftsum}.
\end{proof}

\section{Estimation of sparse high-dimensional signals in quadratic time} \label{sec:worstcase1d}
In this section, we use the filters that we have constructed in Section~\ref{sec:filters} in order to show the first result of the paper, a deterministic algorithm for estimation of Fourier-sparse signals in time which is quadratic in the sparsity.

\begin{algorithm}
	\caption{$d$-dimensional Estimation for Sparse FFT with sample and time complexity $k^2$}
	\begin{algorithmic}[1]
		
		\Procedure{Estimate}{$x, S, n, d$} 
		\State $\wt{T} \gets \textsc{Tree}(\wt{S},n^d)$ 
		\Comment $\wt{S}$: flattening of $S$
		\Comment $\wt{T}$: flattened splitting tree of $S$
		\State Let $T$ be the unflattening of $\wh T$ 
		
		\While{\texttt{$T\neq \emptyset$}}
		\State $v \gets \argmin_{u : \text{ leaf of }{T}} w_{{T}}(u)$, ${\ff} \gets {\ff}_v$		\Comment $\ff$ is label of node $v$
		\State $v_0,v_1,\ldots, v_{d \cdot \log_2n}\gets $ path from $r$ to $v$ in ${T}$, where $v_0 = r$ and $v_{d \cdot \log_2n} = v$

		\For{\texttt{$q = 1 $ to $d$ }}

		\State $T_q^v \gets $ subtree of $T$ rooted at $v_{(q-1)\cdot \log_2n}$
		
		\State Remove all nodes of $T^v_q$ which are at distance more than $\log_2n$ from $v_{(q-1)\cdot \log_2n}$
		
		\State Label every node $u \in T_q^v$ as $f_u = (\ff_{u})_q$

		\State $\mathbf{ g} \gets \textsc{FilterPreProcess} (T_q^v , v_{q\cdot \log_2n}, n)$
		\State $G_{q} \gets  \textsc{{FilterTime}}(\mathbf{ g}_{q}, n)$
		
		\State $\wh G_{q}(\xi_{q}) = \textsc{FilterFrequency}(\mathbf{ g}_{q},n,\xi_{q})$
		\EndFor
		
		\State $G \gets G_{1} \times G_{2} \times ... \times G_{d}$		
		
		\State $h_{\ff} \gets  \sum_{\bm\xi \in [n]^d} \left(\wh \chi_{\bm\xi} \cdot \prod_{q=1}^{d}\wh G_{q}(\xi_{q})\right)$

		\State $\wh \chi_{\ff} \gets \wh \chi_{\ff} + \left(n^d \cdot \sum_{\jj \in [n]^d} x_\jj \cdot G_{-\jj}\right) - h_{\ff} $

		\State ${T} \gets \tree.\textsc{remove}({T},v)$
		
		\EndWhile
		
		\State \textbf{return} $ \wh \chi$
		
		\EndProcedure

	\end{algorithmic}
	\label{alg:high-dim-Est-k2}
	
\end{algorithm}

\begin{theorem}[Estimation guarantee]\label{thm:est-main}
	Suppose $n$ is a power of two integer and $d$ is a positive integer and $S\subseteq [n]^d$. Then, for any signal $x \in \C^{n^d}$ with $\supp{\wh x} \subseteq S$, the procedure \textsc{Estimate}$(x,S,n,d)$ (see Algorithm~\ref{alg:high-dim-Est-k2}) returns $\wh{x}$. Moreover, the sample complexity of this procedure is $O(|S|^2)$ and its runtime is $O(|S|^2 \cdot d \log_2n)$.
\end{theorem}

\begin{proof}
	The proof is by induction on the iteration number $t=0,1,2,...$ of the \emph{while} loop in Algorithm~\ref{alg:high-dim-Est-k2}. One can see that since at each iteration the tree $T$ looses one of its leaves, the algorithm terminates after $|S|$ iterations, since initially the number of leaves of $T$ is $|S|$. Let $\wh \chi ^{(t)}$ denote the signal $\wh \chi$ after iteration $t$, and let ${T}^{(t)}$ denote the tree ${T}$ after iteration $t$ and let $S^{(t)}$ denote the set of frequencies corresponding to leaves of $T^{(t)}$, i.e., $S^{(t)} = \{\ff_u: u \text{ is a leaf of } T^{(t)}\}$. In particular, $\wh\chi^{(0)}=0$ and ${T}^{(0)}$ is the unflattened spltting tree of ${S}$ and $S^{(0)} = S$.  
	
	We claim that for each $t = 0,1,\dots,|S|$, we have
	\begin{equation}
	\begin{split}
	\supp (\wh x - \wh \chi^{(t)}) \subseteq S^{(t)} \text{~and~}|S^{(t)}|=|S|-t\label{eq:xsinduct-hd}
	\end{split}
	\end{equation}
	
	\paragraph{{Base case of induction}:} We have $S^{(0)} = S$ and $\wh \chi^{(0)} \equiv 0 $, which immediately implies \eqref{eq:xsinduct-hd} for $t = 0$.
	
	\paragraph{Inductive step:} For the inductive hypothesis, let $r \geq 1$ and assume that \eqref{eq:xsinduct-hd} holds for $t = r-1$. The main loop of the algorithm finds $v = \argmin_{u : \text{ leaf of } {T}^{(r-1)}} w_{{T}^{(r-1)}}(u)$. By Claim~\ref{kraftsum} along with inductive hypothesis, $w_{{T}^{(r-1)}}(v) \le \log_2|S^{(r-1)}| \leq \log_2 |S|$. Note that the main loop of the algorithm constructs a $(\ff_v, S^{(r-1)})$-isolating filter $G$, along with $\wh G$. In order to do so, the algorithm constructs trees $T^v_q$ for all $q\in \{1,...,d\}$ which in total takes time $O(|S| d \log_2n )$. Given $T^v_q$'s, the algorithm constructs filter $G$ and $\wh G$ in time $O \left( 2^{w_{{T}^{(r-1)}}(v)} + d\log_2 n \right) = O \left( |S| + d\log_2 n \right)$, by Lemma~\ref{lem:isolate-filter-highdim}. Moreover, the filter $G$ has support size $2^{w_{{T}^{(r-1)}}(v)} \leq |S|$ by Lemma~\ref{lem:isolate-filter-highdim}.
	
	By Lemma \ref{lem:isolate-filter-highdim} computing the quantity $h_{\ff} = \sum_{\bm\xi \in [n]^d} \wh \chi^{(r-1)}_{\bm\xi} \cdot \wh G(\bm\xi)$ in line~15 of Algorithm~\ref{alg:high-dim-Est-k2} can be done in time $O(\|\wh \chi^{(r-1)}\|_0 \cdot d\log_2n) = O(|S| \cdot d\log_2n)$. By convolution theorem \ref{conv-thrm}, the quantity $h_{\ff}$ satisfies $h_{\ff}= n^d \cdot (\chi^{(r-1)} \ast G)_0$, and thus
	\begin{align*}
	\left(n^d \cdot \sum_{\jj\in [n]^d} x_{\jj} \cdot G_{-\jj}\right) - h_{\ff} &= n^d \cdot \left( \left(x - \chi^{(r-1)} \right) \ast G \right)_0\\
	&= \wh x_{\ff_v} - \wh\chi_{\ff_v}^{(r-1)},
	\end{align*}
	where the last transition is due to the fact that $G$ is $\left( \ff_v, S^{(r-1)} \right)$-isolating along with the inductive hypothesis of $\supp{\left( \wh x - \wh \chi^{(r-1)} \right)} \subseteq S^{(r-1)}$.
	
	We thus get that $\wh\chi^{(r)}(\cdot) \gets \wh\chi^{(r-1)}(\cdot) +  \left( \wh{x} - \wh{\chi}^{(r-1)} \right)_{\ff_v} \cdot \delta_{\ff_v}(\cdot)$. Moreover, it updates the tree ${T}^{(r)} \gets \tree.\textsc{remove}({T}^{(r-1)},v)$. Also note that the set $S^{(r)}$ gets updated to $S^{(r-1)} \setminus \{\ff_v\}$ accordingly. This establishes \eqref{eq:xsinduct-hd} for $t=r$, thereby completing the inductive step.
	
	\bigskip
	
	The number of steps is exactly $|S|$, as follows from  the inductive claim. Thus, the total runtime is $O(|S|^2 \cdot d\log_2 n)$.
\end{proof}

\section{A lower bound of $k^{1-o(1)}$ rounds of tree pruning}\label{hard-instance}
One apparent disadvantage of our algorithm presented in the previous section is the fact that it only estimates  elements of the Fourier spectrum one at a time, thereby taking $k$ rounds to estimate all elements in the spectrum. Since the isolation of one element takes up to $k$ time due to the support size of $G$, the resulting bound on the runtime is quadratic in $k$. A natural conjecture is that our analysis is not tight, and one can achieve better runtime by removing several nodes of weight at most $\log_2 k+O(1)$ at once. If one could argue that the filters $G$ that isolate the nodes removed in one round have nontrivial overlap, runtime improvements could be achieved. In this section we present a class of signals on which $k^{1-o(1)}$ rounds of  pruning the tree are required, showing that our analysis is essentially optimal.

\paragraph{Tree pruning process} Suppose $n$ is a power of two integer and $\tau$ is a positive integer. Let $T$ be a subtree of $\tfull_n$. The \emph{tree pruning process, $\mathcal{P}(T,\tau,n)$}, is an iterative algorithm that performs the following operations on $T$ successively until $T$ is empty:
\begin{enumerate}
	\item Find $\tilde{S}_\tau = \{ \text{leaves $v$ of } T : w_T(v) \le \tau \}$, i.e., set of vertices of weight no more than $\tau$.
	
	\item For each $v \in \tilde{S}_\tau$ (in an arbitrary order) remove $v$ from $T$ together with the path from $v$ to its closest ancestor that has two children (i.e., run $T.remove(v)$; see Algorithm~\ref{alg:tree-construction}).
\end{enumerate}
We show that for every $k$ and sufficiently large integer $n$ there exists a tree $T$ with $k$ leaves such that $\mathcal{P}(T,\tau,n)$ with $\tau=\log_2 k+O(1)$ requires $k^{1-o(1)}$ rounds to terminate. This in particular shows that our $k^2$ runtime analysis from section~\ref{sec:worstcase1d} cannot be improved by reusing work done in a single iteration, and hence our analysis is essentially optimal. Our construction is one-dimensional, although higher dimensional extensions can be readily obtained.

\begin{theorem}\label{lem:process}
	For any integer constant $c\geq 1$, sufficiently large power of two integer $n$ there exists $k=\Theta(\log^c n)$ such that if  $\tau=\log_2 k+O(1)$, the following condition holds. There exists a subtree $T$ of $\tfull_n$ with $k$ leaves such that the tree pruning process $\mathcal{P}(T,\tau,n)$ requires $k^{1-o(1)}$ iterations to terminate. 
	
\end{theorem}

The following simple lemma is crucial to our analysis
\begin{lemma}[Monotonicity of tree pruning process] \label{lem:monotonic}
	Suppose $n$ is a power of two integer $T'$ a subtree of $\tfull_n$ and $T$ a subtree of $T'$. Then for every integer $\tau$ the number of rounds that it takes $\mathcal{P}(T,\tau,n)$ to collapse $T$ is at most the number of rounds that it takes $\mathcal{P}(T',\tau,n)$  to collapse $T'$.
\end{lemma}
\begin{proof}
	For $j = 0,1,2,\dots$, let $T^{(j)}$ (respectively $T'^{(j)}$) denote the tree obtained by performing $j$ rounds of the tree pruning process (with threshold $\tau$) to $T$ (respectively $T'$). Note that $T^{(0)} = T$ and $T'^{(0)} = T'$.
	
	We claim that $T^{(j)}$ is a subtree of $T'^{(j)}$ for all $j=0,1,\dots$, which will obviously imply the desired conclusion. We use induction on $j$. Note that the {\bf base} of induction is trivial for $j=0$. Now, we prove the {\bf inductive step}. Suppose $j > 0$. By the inductive hypothesis, we have that $T^{(j-1)}$ is a subtree of $T'^{(j-1)}$. Thus, for any leaf $v$ that appears in both $T^{(j-1)}$ and $T'^{(j-1)}$, we have $w_T(v) \leq w_{T'}(v)$ (this is because any node in $T'^{(j-1)}$ along the path from the root to $v$ that has exactly one child will also have exactly one child in $T^{(j-1)}$). Hence, if $v$ is removed from $T'^{(j-1)}$ in the $j$-th iteration of the process, then it is also removed from $T^{(j-1)}$ during the $j$-th iteration. Hence, $T^{(j)}$ is a subtree of $T'^{(j)}$, which completes the inductive step and, therefore, proves the claim.
\end{proof}

We recall a few definitions.
\lowhamming*
Note that $|H_c^n| = \sum_{j=0}^c \binom{\log_2 n}{j}$.

\lowhammingsupport*
Note that for any $x \in \mathcal{X}^n_c$ we have that $\| x \|_0 = \sum_{i=0}^c \binom{\log_2n}{i}$, so for any $c \le (\frac{1}{2} - \epsilon)\log_2n$, the signals that are contained in class $\mathcal{X}^n_c$ are $\Theta\Big( \binom{\log_2n}{c} \Big)$-sparse.

\lowhammingweighttrees*
It is not hard to see that $T^n_c$ is in fact the splitting tree for the set $H^n_c$ and, hence, the number of its leaves is $\sum_{i=0}^c \binom{\log_2n}{i}$. 

Now, we are ready to prove Theorem~\ref{lem:process}.

\begin{proofof}{Theorem~\ref{lem:process}}
	Let us choose the tree $T$ to be $T^n_c$ for some positive integer $c$. We will set parameter $c$ at the end of the proof. Let $D(n, c, \tau)$ denote the number of iterations required to collapse $T_c^n$ with threshold $\tau$. We prove that
	\begin{equation}
	D(n, c, \tau) \ge \frac{\log_2^c n}{c!\cdot \tau^c} \label{eq:itercollapse}
	\end{equation}
	for any power of two integer $n$, any integer $0 \leq c \leq \log_2 n$, and any positive integer $\tau $. We use induction on $c$.
	
	{\bf Base:} Note that for $c=0$, the tree $T_c^n$ has one leaf, which gets removed in the first iteration of the tree pruning process. Thus, $D(n,0,\tau) = 1$ for any power of two $n$ and $\tau\geq 1$, and so, \eqref{eq:itercollapse} holds for $c=0$.
	
	{\bf Inductive step:} Suppose $c > 0$. For any $T_c^n$, we label the nodes along the path from the root to the rightmost leaf (i.e., the path formed by starting at the root and repeatedly following the right child) in order as $0,1,\dots, \log_2 n$.
	
	Note that if $n \leq 2^\tau$, then 
	\[
	\frac{\log_2 ^c n}{c! \cdot \tau^c} \leq \frac{\tau^c}{c!\cdot \tau^c} \leq 1.
	\]
	Thus, \eqref{eq:itercollapse} does indeed hold for $n \leq 2^\tau$.
	
	Now, suppose $n > 2^\tau$. Recall that a copy of $T_{c-1}^{n/2^{j+1}}$ is rooted at the left child of node $j$ of $T_c^n$ for all $j=0,1,\dots,\tau-1$. We divide the pruning process on $T_c^n$ into two phases. The first phase consists of the process up until the point at which the left subtree of node $j$ in $T_c^n$ completely collapses for some $j \in \{0,1,\dots,\tau-1\}$, while the second phases consists of the process thereafter. Thus, the number of rounds in the first phase is just the number of rounds till the top $\tau$ left subtrees collapses.
	
	Note that during the first phase, the behavior of the collapsing process on the left subtree of node $j$ corresponds to running a collapsing process with threshold $\tau-j-1$ on $T_{c-1}^{n/2^{j+1}}$. Thus, the number of rounds in the first phase is,
	\[
	R = \min_{0\leq j < \tau} \{D(n/2^{j+1}, c-1, \tau-j-1)\}.
	\]

	By the inductive hypothesis (on $c$), we have that for $j=0,1,\dots, \tau-1$
	\[
	D(n/2^{j+1}, c-1, \tau-j-1) \ge \frac{1}{(c-1)!} \cdot \left(\frac{\log_2 n -j-1}{\tau-j-1}\right)^{c-1},
	\]
	which implies that $ R \ge \frac{1}{(c-1)!}\cdot \left(\frac{\log_2 n - 1}{\tau - 1}\right)^{c-1}$ since we assumed $\tau\leq \log_2 n$.

	Now, let $T'$ be the tree obtained after performing $R$ rounds of the collapsing process on $T_c^n$. Moreover, let $T''$ be the tree obtained by further removing any left subtrees of nodes $0,1, \dots, \tau-1$. By Lemma~\ref{lem:monotonic}, we have that the number of rounds needed to collapse $T'$ is at least the number of rounds needed to collapse $T''$. Moreover, observe that the number of rounds needed to collapse $T''$ is precisely $D(n/2^\tau, c, \tau)$, thus, the number of rounds in the second phase is at least $D(n/2^\tau, c, \tau)$, and so,
	\begin{align*}
	D(n,c,\tau) &\geq R + D(n/2^\tau, c, \tau)\\
	&\geq \frac{1}{(c-1)!}\cdot \left(\frac{\log_2 n - 1}{\tau - 1}\right)^{c-1} + D(n/2^\tau, c, \tau).
	\end{align*}
	Note that a similar argument gives us
	\begin{align*}
	D(n/2^{a\tau},c,\tau) \geq \frac{1}{(c-1)!}\cdot \left(\frac{\log_2 n - a\tau - 1}{\tau - 1}\right)^{c-1} + D(n/2^{(a+1)\tau}, c, \tau) 
	\end{align*}
	for all $a=0,1,\dots, \lfloor (\log_2 n -1)/\tau\rfloor - 1$ (this condition ensures that $\tau\leq \log_2 (n/2^{a\tau})$, as required by our argument above). Hence, it follows that
	\begin{align*}
	D(n,c,\tau) &\geq \sum_{a=0}^{\lfloor (\log_2 n -1)/\tau\rfloor - 1} \frac{1}{(c-1)!}\cdot \left(\frac{\log_2 n - a\tau - 1}{\tau - 1}\right)^{c-1} + D(n/2^{\tau \cdot \lfloor (\log_2 n -1)/\tau\rfloor}, c, \tau)\\
	&\geq \frac{1}{(c-1)!} \sum_{a=0}^{\lfloor (\log_2 n -1)/\tau\rfloor - 1} \left(\frac{\log_2 n}{\tau} - a\right)^{c-1} + 1\\
	&\geq \frac{1}{(c-1)!} \cdot \int_1^{\frac{\log_2 n}{\tau}} u^{c-1}\,du + 1\\
	&= \frac{1}{(c-1)!} \cdot \frac{1}{c} \left(\left(\frac{\log_2 n}{\tau}\right)^c - 1\right) + 1\\
	&\geq \frac{\log_2^c n}{c! \cdot \tau^c},
	\end{align*}
	which establishes \eqref{eq:itercollapse} for $n > 2^\tau$. This completes the inductive step.
	
	Recall that $k = \Theta\Big( \binom{\log_2n}{c} \Big)$, so for any constant $c$ one has $k=\Theta(\binom{\log_2n}{c})\leq (e \log_2 n/c)^c$. Setting $\tau=\log_2 k+O(1)$, we get
	$$
	D(n, c, \tau) \ge \frac{\log_2^c n}{c!\cdot \tau^c}=\Theta(k/(\log_2 k)^c)=k^{1-o(1)}, 
	$$
	as required.
	
\end{proofof}

\section{Sparse FFT for worst-case sparse signals and worst case signals with random phase}\label{sec:randomsign}

In this section we prove the main result of the paper, namely
\sfftworstcase*

We also study Fourier sparse signals $x$ whose nonzero frequencies are distributed arbitrarily (worst-case) and whose values at the nonzero frequencies are independently chosen to have a uniformly random phase. Recall Definition~\ref{random-sign}:

\randsign*

For this model we prove the stronger result: 
\sfftrandphase*

The main property that allows us to obtain the stronger result is the fact that a small number of time domain samples from such a signal suffice to approximate its energy with high confidence (whereas $\Omega(k)$ samples are required in general for a worst-case $k$-sparse signal). This is reflected by the following
\begin{lemma}\label{lem:random-sign}
	For any positive integer $d$, power of two $n$, and worst-case signal with random phase $x$, we have
	\[
	\Pr\left[\frac{1}{2} \cdot \frac{\|\beta\|_2^2}{n^{2d}} \leq \frac{1}{s} \sum_{j=1}^s |x_{\tv_j}|^2 \leq \frac{3}{2} \cdot \frac{\|\beta\|_2^2}{n^{2d}}\right] \geq 1 - \frac{1}{n^{4d}},
	\]
	where $s = Cd^3 \log_2^3 n$ for some absolute constant $C > 0$ and $\tv_1, \tv_2, \dots, \tv_s \sim \unif([n]^d)$ are i.i.d. random variables. The probability is over the randomness in choosing the various $\tv_j$ as well the randomness in the choice of phase for each frequency of $\wh x$.
\end{lemma}
For completeness we present a proof for this lemma in Appendix \ref{appx:A}.

\subsection{Proofs of Theorems~\ref{thm:sfft-worstcase} and ~\ref{thm:sfft-rand-signs}}
Given the construction of our adaptive aliasing filter from the previous section, our sparse recovery algorithms follow by a reduction to the estimation problem.  We find the vertex $v^*=\text{argmin}_{v\in T} w_T(v)$, which, by Kraft's inequality, satisfies $w_T(v^*)\leq \log_2 k$. We then define an auxiliary tree $T'$ by appending a left $a$ and a right child $b$ to $v$. Then for each of the children $a, b$, we, in turn, construct a filter $G$ that isolates them from the rest of $T$ (i.e., from the frequency cones of other nodes in $T$) and check whether the corresponding restricted signals are nonzero. The latter is unfortunately a nontrivial task, since the sparsity of these signals can be as high as $k$, and detecting whether a $k$-sparse signal is nonzero requires $\Omega(k)$ samples. However, a fixed set of $k\log^3 N$ locations that satisfies the restricted isometry property (RIP) can be selected, and accessing the signal on those values suffices to test whether it is nonzero. If the signal is further assumed to be a worst case random phase signal, then a polylogarithmic number of samples suffices. The following lemma (Lemma~\ref{lem:zero-test}) makes the latter claim formal. The algorithm is presented as Algorithm~\ref{alg:fullsparsefft}.

\begin{lemma}[\textsc{ZeroTest} guarantee]\label{lem:zero-test}
	Suppose $d$ is a positive integer and $n$ is a power of two. Assume $T$ is a subtree of $\tfull_{n^d}$. Suppose that signals $x , \wh{\chi} \in \C^{n^d}$ satisfy $\supp{ (\wh{x} - \wh{\chi}) } \subseteq \bigcup_{u: \text{~leaf of~}T} \subtree_T(u)$. Suppose that $\mathbf{\Delta}$ is a multiset of sample from $[n]^d$ which satisfies the following for every leaf $v$ of $T$:
	\[ \frac{1}{2} \cdot \frac{\|\wh{y} \|_2^2}{n^{2d}} \leq \frac{1}{|\mathbf{\Delta}|} \cdot \sum_{\Delta \in \mathbf{\Delta}} |y_\Delta|^2 \leq \frac{3}{2} \cdot \frac{\|\wh{y} \|_2^2}{n^{2d}} \]
	where $y = (\wh{x} - \wh{\chi})_{\subtree_T(v)}$ is the signal obtained by restricting $\wh{x} - \wh{\chi}$ to frequencies $\bm\xi \in \subtree_T(v)$ and zeroing it out on all other frequencies.
	
	Then the following conditions hold:
	\begin{itemize}
		\item \textsc{ZeroTest}$(x, \wh{\chi}, T, v, n, d, \mathbf{\Delta})$ outputs {\bf true} if $\supp{(\wh{x}-\wh\chi)} \cap {\subtree_T(v)}\neq \emptyset$; otherwise, it outputs {\bf false}.
		
		\item The sample complexity of this procedure is $O(2^{w_T(v)} \cdot |\mathbf{\Delta}|)$, where $w_T(v)$ is the weight of leaf $v$ in $T$ (see Definition~\ref{def:weight}).
		
		\item The runtime of the \textsc{ZeroTest} procedure is \[ O \left( \|\wh{\chi}\|_0 \cdot |\mathbf{\Delta}| + |T| \cdot d\log_2n + 2^{w_T(v)} \cdot |\mathbf{\Delta}| \right), \] where $|T|$ denotes the number of leaves of $T$.
	\end{itemize}

\end{lemma}
\begin{proof}
	Consider lines 14-15 in Algorithm~\ref{alg:zero-test}. By Claim~\ref{conv-thrm}, we have that
	\begin{align*}
	h_{\ff}^\Delta &= \frac{1}{n^d} \sum_{\bm\xi \in [n]^d} e^{2\pi i \frac{\bm{\xi}^T \Delta}{n}} \cdot \wh\chi_{\bm{\xi}} \wh G_{\bm{\xi}}  \\
	&= \sum_{\jj \in [n]^d} G_{\Delta-\jj} \cdot \chi_{\jj}.
	\end{align*}
	Thus,
	\begin{align*}
	H_{\ff}^\Delta &= \left( \sum_{\jj \in [n]^d} G_{\Delta-\jj} \cdot x_{\jj} \right) - h_{\ff}^\Delta \\
	&= \sum_{\jj \in [n]^d} G_{\Delta-\jj} \cdot (x - \chi)_{\jj}.
	\end{align*}
	
	Note that, by Lemma \ref{lem:isolate-filter-highdim}, the filter $G$ used in Algorithm~\ref{alg:zero-test} is a $(v, T)$-isolating filter. Therefore, by the assumption $\supp (\wh{x} - \wh{\chi}) \subseteq \bigcup_{u:\text{~leaf of~}T} \subtree_T(u)$ and the definition of a $(v, T)$-isolating filter (see Definition~\ref{def:v-t-isolating-highdim}), we have
	\begin{align*}
	H_{\ff}^\Delta	&= \sum_{\jj \in [n]^d} G_{\Delta-\jj} \cdot (x - \chi)_{\jj} \\
	&=\frac{1}{n^d} \sum_{\bm\xi \in \subtree_T(v) } (\wh{x} - \wh{\chi})_{\bm\xi} \cdot e^{2\pi i \frac{\bm\xi^T \Delta}{n}}.
	\end{align*}
	Note that $H_{\ff}^\Delta$ is essentially the inverse Fourier transform of $(\wh{x} - \wh{\chi})_{\subtree_T(v)}$, where $(\wh{x} - \wh{\chi})_{\subtree_T(v)}$ denotes the signal obtained by restricting $\wh{x} - \wh{\chi}$ to frequencies $\bm\xi \in \subtree_T(v)$ and zeroing out the signal on all other frequencies. By the assumption of the lemma we have the following:
	
	\[ \frac{1}{2} \cdot \frac{\|(\wh{x} - \wh{\chi})_{\subtree_T(v)}\|_2^2}{n^{2d}} \leq \frac{1}{|\mathbf{\Delta}|} \cdot \sum_{\Delta \in \mathbf{\Delta}} |H^\Delta_{\ff}|^2 \leq \frac{3}{2} \cdot \frac{\|(\wh{x} - \wh{\chi})_{\subtree_T(v)}\|_2^2}{n^{2d}}. \]
	Therefore the first claim of the lemma holds. 
	
	Note that in order to construct a $(v, T)$-isolating filter $G$, along with $\wh G$, the algorithm constructs trees $T^v_q$ for all $q\in \{1,...,d\}$, which has total time complexity $O(|T| d \log_2n )$. Given $T^v_q$'s, the algorithm constructs filter $G$ and $\wh G$ in time $O \left( 2^{w_{{T}}(v)} + d\log_2 n \right)$, by Lemma~\ref{lem:isolate-filter-highdim}. Moreover, the filter $G$ has support size $2^{w_{{T}}(v)}$, by Lemma~\ref{lem:isolate-filter-highdim}.
	
	By Lemma \ref{lem:isolate-filter-highdim}, computing the quantities $h_{\ff}^\Delta = \frac{1}{n^d} \sum_{\bm\xi \in [n]^d} e^{2\pi i \frac{\bm{\xi}^T \Delta}{n}} \cdot \wh\chi_{\bm{\xi}} \wh G_{\bm{\xi}}$ for all $\Delta$ in line~14 of Algorithm~\ref{alg:zero-test} can be done in time $O\left( \|\wh\chi\|_0 \cdot (|\mathbf{\Delta}| + d \log_2 n) \right) = O \left( \|\wh \chi\|_0 \cdot |\mathbf{\Delta}| \right)$.
	Given the values of $h_{\ff}^\Delta$ for various $\Delta$, computing all $\big\{|H^\Delta_{f^*}|^2\big\}_{\Delta \in \mathbf{\Delta}}$ in line~15 takes time $O \left( 2^{w_T(v)}\cdot |\mathbf{\Delta}| \right)$. Therefore the total runtime of this procedure is
	\[ O \left( |T|d \log_2n + 2^{w_T(v)}\cdot |\mathbf{\Delta}| + \|\wh{\chi}\|_0 \cdot |\mathbf{\Delta}| \right), \]
	as desired.
	
	Because support size of $G$ is $2^{w_T(v)}$, computing all $\big\{|H^\Delta_{f^*}|^2\big\}_{\Delta \in \mathbf{\Delta}}$ in line~15 of the algorithm requires $O(2^{w_T(v)} \cdot |\mathbf{\Delta}|)$ samples from $x$ which proves the second claim of the lemma.
\end{proof}

\begin{algorithm}
	\caption{Procedure for testing zero hypothesis}
	\begin{algorithmic}[1]
		
		\Procedure{ZeroTest}{$x, \wh{\chi}, T, v, n, d, \mathbf{\Delta}$} \Comment $\bm{\Delta}$: multiset of elements from $[n]^d$
		\State $\ff \gets \ff_v$, $l \gets l_T(v)$, $q^* \gets \left\lceil \frac{l}{\log_2 n} \right\rceil$

		\State $v_0,v_1,\ldots, v_{l}\gets $ path from $r$ to $v$ in ${T}$, where $v_0 = r$ and $v_{l} = v$
		
		\State $(u_1, u_2, \cdots, u_{q^*-1}, u_{q^*}) \gets (v_{\log_2n}, v_{2\log_2n} , \cdots, v_{(q^*-1)\cdot \log_2n}, v_l)$
		
		\For{\texttt{$q = 1 $ to $q^*$ }}

		\State $T_q^v \gets $ subtree of $T$ rooted at $v_{(q-1)\cdot \log_2n}$
		\State Remove all nodes of $T^v_q$ which are at distance more than $\log_2n$ from $v_{(q-1)\cdot \log_2n}$
		\State Label every node $w \in T_q^v$ as $f_w = (\ff_{w})_q$

		\State $\mathbf{ g}_{q} \gets \textsc{FilterPreProcess}(T_q^v, u_{q}, n)$
		\State $G_{q} \gets  \textsc{{FilterTime}}(\mathbf{ g}_{q}, n)$
		
		\State $\wh G_{q}(\xi_{q}) = \textsc{FilterFrequency}(\mathbf{ g}_{q},n,\xi_{q})$
		\EndFor
		
		\State $G \gets G_{1} \times G_{2} \times ... \times G_{d}$		
		
		\State $h^{\Delta}_{\ff} \gets \frac{1}{n^d}\sum_{\bm\xi \in [n]^d} \big( e^{2\pi i \frac{\bm\xi^T \Delta}{n}} \cdot \wh \chi_{\bm\xi} \cdot \prod_{q=1}^{d}\wh G_{q}(\xi_{q})\big)$ for all $\Delta \in \mathbf{\Delta}$
		
		\State $H^{\Delta}_{\ff} \gets \big(\sum_{\jj \in [n]^d} G(\Delta-\jj) \cdot x_{\jj}\big) - h^\Delta_{\ff}$ for all $\Delta \in \mathbf{\Delta}$

		\If{$\frac{1}{|\bm{\Delta}|}\sum_{\Delta \in \mathbf{\Delta}} |H^\Delta_{\ff}|^2 = 0$}
		\State \textbf{return} {\bf false}
		
		\Else	
		\State \textbf{return} {\bf true}

		\EndIf

		\EndProcedure
		
	\end{algorithmic}
	\label{alg:zero-test}
	
\end{algorithm}

\begin{algorithm}
	\caption{Sparse FFT for worst-case sparse signals}
	\begin{algorithmic}[1]
		
		\Procedure{SparseFFT}{$x, n, d, k$} 
		
		\State $\mathbf{ \Delta} \gets$ Multiset of $[n]^d$ corresponding to Fourier measurements satisfying RIP of order $k$ \Comment{$|\mathbf{ \Delta}| = O(k \log_2^2k \cdot d\log_2n)$, see Theorem \ref{RIP-thrm}}
		
		\State $T\gets \{r\}, \ff_r\gets 0$

		\While{$T \neq \emptyset$}
		\State $v \gets \argmin_{u: \text{ leaf of } T} w_T(u)$, $\ff\gets \ff_v$ and $l\gets l_T(v)$
		\If{\texttt{$l= d \log_2 n$}}\Comment{All bits of $v$ have been discovered}
		
		\State $v_0,v_1,\ldots, v_{d \cdot \log_2n}\gets $ path from $r$ to $v$ in ${T}$, where $v_0 = r$ and $v_{d \cdot \log_2n} = v$
		
		\For{\texttt{$q = 1 $ to $d$ }}

		\State $T_q^v \gets $ subtree of $T$ rooted at $v_{(q-1)\cdot \log_2n}$
		\State Remove all nodes of $T^v_q$ which are at distance more than $\log_2n$ from $v_{(q-1)\cdot \log_2n}$
		\State Label every node $u \in T_q^v$ as $f_u = (\ff_{u})_q$ 		
		
		\State $\mathbf{ g}_{q} \gets \textsc{FilterPreProcess}(T_q^v, v_{q\cdot \log_2n}, n)$
		\State $G_{q} \gets  \textsc{{FilterTime}}(\mathbf{ g}_{q}, n)$
		
		\State $\wh G_{q}(\xi_{q}) = \textsc{FilterFrequency}(\mathbf{ g}_{q},n,\xi_{q})$
		\EndFor
		
		\State $G \gets G_{1} \times G_{2} \times ... \times G_{d}$		
		
		\State $h_{\ff} \gets  \sum_{\bm\xi \in [n]^d} \big(\wh \chi_{\bm\xi} \cdot \prod_{q=1}^{d}\wh G_{q}(\xi_{q})\big)$

		\State $\wh \chi_{\ff} \gets \wh \chi_{\ff} + \big(n^d \cdot \sum_{\jj \in [n]^d} x_\jj \cdot G_{-\jj}\big) - h_{\ff} $
		
		\State $T \gets \tree.\textsc{remove}(T,v)$

		\Else
		\State $T'\gets T\cup \{\text{left child $u$ of $v$}\}\cup \{\text{right child $w$ of $v$}\}$ 
		\If{\textsc{ZeroTest}$(x, \wh{\chi}, T', w , n , d, \mathbf{ \Delta})$}
		
		\State $\text{Add } w \text{ as the right child of node } v \text{ to tree } T$	
		\State $\ff_{w} \gets \ff$ \Comment{Frequency corresponding to node $w$}
		\EndIf
		
		\If{\textsc{ZeroTest}$(x, \wh{\chi},  T', u , n , d, \mathbf{ \Delta})$}
		
		\State $\text{Add } u \text{ as the left child of node } v \text{ to tree } T$	
		\State $\ff_{u} \gets\ff+2^l;$ \Comment{Frequency corresponding to node $u$}
		\EndIf
		\EndIf
		\EndWhile
		\State \textbf{return} $\wh{\chi};$
		\EndProcedure
		
	\end{algorithmic}
	\label{alg:fullsparsefft}
	
\end{algorithm}

\begin{algorithm}
	\caption{Sparse FFT for worst-case sparse signals with random phase}
	\begin{algorithmic}[1]
		
		\Procedure{SparseFFT-RandomPhase}{$x, n, d, k$} 
		
		\State $\mathbf{ \Delta} \gets \left\{ \Delta_i: \Delta_i \sim \unif([n]^d),  \forall i \in [C d^3\log_2^3n] \right\}$\Comment{$C$: constant} \Comment{$\mathbf{ \Delta}$: Multiset}
		
		\State $T\gets \{r\}, \ff_r\gets 0$

		\While{$T \neq \emptyset$}
		\State $v \gets \argmin_{u: \text{ leaf of } T} w_T(u)$, $\ff\gets \ff_v$ and $l\gets l_T(v)$
		\If{\texttt{$l= d \log_2 n$}}\Comment{All bits of $v$ have been discovered}
		
		\State $v_0,v_1,\ldots, v_{d \cdot \log_2n}\gets $ path from $r$ to $v$ in ${T}$, where $v_0 = r$ and $v_{d \cdot \log_2n} = v$
		
		\For{\texttt{$q = 1 $ to $d$ }}

		\State $T_q^v \gets $ subtree of $T$ rooted at $v_{(q-1)\cdot \log_2n}$
		\State Remove all nodes of $T^v_q$ which are at distance more than $log_2n$ from $v_{(q-1)\cdot \log_2n}$
		\State Label every node $u \in T_q^v$ as $f_u = (\ff_{u})_q$ 		
		
		\State $\mathbf{ g}_{q} \gets \textsc{FilterPreProcess}(T_q^v, v_{q\cdot \log_2n}, n)$
		\State $G_{q} \gets  \textsc{{FilterTime}}(\mathbf{ g}_{q}, n)$
		
		\State $\wh G_{q}(\xi_{q}) = \textsc{FilterFrequency}(\mathbf{ g}_{q},n,\xi_{q})$
		\EndFor
		
		\State $G \gets G_{1} \times G_{2} \times ... \times G_{d}$		
		
		\State $h_{\ff} \gets  \sum_{\bm\xi \in [n]^d} \big(\wh \chi_{\bm\xi} \cdot \prod_{q=1}^{d}\wh G_{q}(\xi_{q})\big)$

		\State $\wh \chi_{\ff} \gets \wh \chi_{\ff} + \big(n^d \cdot \sum_{\jj \in [n]^d} x_\jj \cdot G_{-\jj}\big) - h_{\ff} $
		
		\State $T \gets \tree.\textsc{remove}(T,v)$

		\Else
		\State $T'\gets T\cup \{\text{left child $u$ of $v$}\}\cup \{\text{right child $w$ of $v$}\}$ 
		\If{\textsc{ZeroTest}$(x, \wh{\chi}, T', w , n , d, \mathbf{ \Delta})$}
		
		\State $\text{Add } w \text{ as the right child of node } v \text{ to tree } T$	
		\State $\ff_{w} \gets \ff$ \Comment{Frequency corresponding to node $w$}
		\EndIf
		
		\If{\textsc{ZeroTest}$(x, \wh{\chi},  T', u , n , d, \mathbf{ \Delta})$}
		
		\State $\text{Add } u \text{ as the left child of node } v \text{ to tree } T$	
		\State $\ff_{u} \gets\ff+2^l;$ \Comment{Frequency corresponding to node $u$}
		\EndIf
		\EndIf
		\EndWhile
		\State \textbf{return} $\wh{\chi};$
		\EndProcedure
		
	\end{algorithmic}
	\label{alg:sparsefftphase}
	
\end{algorithm}

We now prove our main result:

\begin{proofof}{Theorems \ref{thm:sfft-worstcase} and \ref{thm:sfft-rand-signs}}
	Note that Algorithms~\ref{alg:fullsparsefft} and \ref{alg:sparsefftphase} are identical except in line 2. We first analyze the common code of the algorithms (after line 2) under the assumption that the set $\mathbf{ \Delta}$ in all calls to \textsc{ZeroTest} are replaced with a more powerful set which satisfies the precondition of Lemma~\ref{lem:zero-test} hence \textsc{ZeroTest} correctly tests the zero hypothesis on its input signal with probability $1$. We then establish a coupling between this idealized execution and the actual execution for both Algorithms \ref{alg:fullsparsefft} and \ref{alg:sparsefftphase}, leading to our result.
	
	Let $m$ denote the size of the set $m = |\mathbf{ \Delta}|$. We prove that the following properties are maintained throughout the execution of \textsc{SparseFFT} (Algorithm~\ref{alg:fullsparsefft}) and \textsc{SparseFFT-RandomPhase} (Algorithm~\ref{alg:sparsefftphase}):
	\begin{description}
		\item[(1)] $\supp{ (\wh{x} - \wh{\chi})} \subseteq \bigcup_{u: \text{ leaf of } T} \subtree_T(u)$; 
		\item[(2)] For every leaf $u$ of tree $T$ one has  $\supp{(\wh{x}- \wh{\chi})}\cap \subtree_T(u)\neq \emptyset$; 
		
		\item[(3)] If $\wh x$ is a worst-case signal with random phase, then $\wh x - \wh \chi$ is a worst-case signal with random phase; 
		
		\item[(4)] The quantity $\phi=(d\log_2n + 1)\|\wh x - \wh \chi\|_0 - \sum_{u : \text{ leaf of } T} l_T(u)$ always decreases by at least 1 on every iteration of Algorithm \ref{alg:fullsparsefft} or \ref{alg:sparsefftphase};
		
		\item[(5)] Always $\|\wh x - \wh \chi\|_0 \le k$;
	\end{description}

	The {\bf base} of the induction is provided by the first iteration, at which point $T$ is a single vertex $ T = \{r\}$ where $r$ is the root with $\ff_r=0$ and $\wh \chi=0$. The conditions {\bf (1)} and {\bf (2)} and {\bf (3)} and {\bf (5)} are satisfied since $\subtree_T(r)=[n]^d$ and $\supp (\wh x- \wh\chi)=\supp{ \wh x}\neq \emptyset$ and $\wh x - \wh \chi = \wh x$ is a worst-case signal with random phase if $\wh x$ is a worst-case signal with random phase.
	
	We now prove the {\bf inductive step}.  We assume that conditions {\bf (1)} and {\bf (2)} and {\bf (3)} and {\bf (5)} of the inductive hypothesis are satisfied at the beginning of a certain iteration and argue that conditions {\bf (1)} and {\bf (2)} and {\bf (3)} and {\bf (5)} are maintained at the end of the iteration. We also show that the value of the quantity $\phi$ defined in {\bf (4)}, at the end of the loop is smaller than its value at the start of the loop by at least one. Let $v\in T$ be the smallest weight leaf chosen by the algorithm in line~4. We now consider two cases.
	
	{\bf Case 1: $l_T(v)=d\log_2 n$.}  Since $G$ is a $(v, T)$-isolating filter, we have by Definition~\ref{def:v-t-isolating} that 
	for every signal $z \in \C^{n^d}$ with Fourier support $\supp\wh{z} \subseteq \bigcup_{u: \text{~leaf in~} T} \subtree_T(u)$ and for all $t\in[n]^d$,
	\begin{equation}\label{eq:23h9hgg}
	\sum_{\jj\in[n]^d} z_{\jj} G_{\tv-\jj} = \frac{1}{n^d} \sum_{\ff' \in \subtree_T(v)}\wh{z}_{\ff'} e^{2\pi i \frac{\ff'^T \tv}{n}}.
	\end{equation}
	By condition {\bf (1)} of the inductive hypothesis  one has 
	$\supp (\wh{x} - \wh{\chi}) \subseteq \bigcup_{u: \text{ leaf of } T} \subtree_T(u),
	$
	and thus we can apply ~\eqref{eq:23h9hgg} with $z=x-\chi$ and $\tv = 0$, obtaining
	
	\begin{equation}\label{eq:est-intermediate}
	\begin{split}
	\sum_{\jj\in[n]^d} (x-\chi)_{\jj} G_{-\jj} &= \frac{1}{n^d} \sum_{\ff' \in \subtree_T(v)}\wh{(x-\chi)}_{\ff'}.\\
	\end{split}
	\end{equation}
	
	Note that by Claim~\ref{conv-thrm},
	\begin{equation*}
	\begin{split}
	n^d \cdot \sum_{\jj\in[n]^d} \chi_{\jj} G_{-\jj} &= \sum_{f\in [n]} \wh{\chi}_f \wh{G}_f=h_{\ff},
	\end{split}
	\end{equation*}
	where $h_{\ff}$ is the quantity computed in line~15. We thus get that 
	\begin{align*}
	n^d \sum_{\jj \in [n]^d} x_{\jj} \cdot G_{-\jj} - h_{\ff} &= \sum_{\ff' \in \subtree_T(v)}\wh{(x-\chi)}_{\ff'}\\
	&= \wh{(x - \chi)}_{\ff_v},
	\end{align*}
	because $ \subtree_T(v)=\{\ff_v\}$ due to the assumption that $l_T(v)=d \log_2 n$. Thus we get that $\wh \chi (\cdot) \gets \wh \chi (\cdot) + \wh{(x - \chi)}_{\ff_v} \delta_{\ff_v}(\cdot)$ therefore at the end of the loop we have $\wh{(x - \chi)}_{\ff_v} = 0$ which means that $\ff_v$ will no longer be in $\supp{\wh{(x - \chi)}}$. And also $v$ gets removed from tree $T$ implying that $\{\ff_v\} = \subtree_T(v)$ will be excluded from $\bigcup_{u: \text{ leaf of } T} \subtree_T(u)$. Note that this also implies that $\wh {(x - \chi)}$ will remain a worst-case signal with random phase. Therefore, condition {\bf (1)} and {\bf (2)} and {\bf (3)} hold. 
	
	Now, note that $\| \wh{(x - \chi)} \|_0$ will decrease by 1 exactly because $\ff_v$ is no longer in $\supp{\wh{(x - \chi)}}$ and the rest of the support is unchanged. This shows that condition {\bf (5)} holds. Also, $\sum_{u : \text{ leaf of } T} l_T(u)$ decreases by exactly $d\log_2n$ because the level of $v$ was $l_T(v) = d\log_2n$ and $v$ gets removed from $T$. So $\phi$ will decrease by exactly one as required in condition {\bf (4)}.

	{\bf Case 2} Suppose that $l_T(v) < d\log_2 n$. We first check that the invocation of $\textsc{ZeroTest}$ satisfies preconditions of Lemma~\ref{lem:zero-test}. We need to ensure that for the residual signal $\wh x - \wh{\chi}$  one has
	$$
	\supp{ (\wh{x} - \wh{\chi}) } \subseteq \bigcup_{u: \text{~leaf of~}T'} \subtree_{T'}(u),
	$$
	where $T'$ is the tree obtained from $T$ by adding two children of $v$ (line~19). This follows, since by the inductive hypothesis we have
	$$
	\supp{ (\wh{x} - \wh{\chi}) } \subseteq \bigcup_{u:\text{~leaf of~}T} \subtree_T(u),
	$$
	and because by claim \ref{claim:sub-tree} we have,
	$$
	\subtree_T(v)=\subtree_{T'}(u)\cup \subtree_{T'}(v).
	$$
	We thus get that the preconditions of Lemma~\ref{lem:zero-test} are satisfied, and the output of \textsc{ZeroTest}$(x, \wh{\chi}, T', w, n , d, \mathbf{ \Delta})$ is {\bf true} if $(\wh{x}-\wh{\chi})_{\subtree_{T'}(w)}\neq 0$ and {\bf false} otherwise. A similar analysis shows that the algorithm correctly tests the zero hypothesis on $(\wh{x}-\wh{\chi})_{\subtree_{T'}(u)}$. We thus get, letting $T_{new}$ denote the tree $T$ at the end of the while loop, that 
	$$
	\supp{ (\wh{x} - \wh{\chi}) } \subseteq \bigcup_{u: \text{ leaf of } T_{new}} \subtree_{T_{new}}(u), 
	$$
	and for every $v \in T_{new}$ one has $\supp{(\wh{x}-\wh{\chi})} \cap {\subtree_{T_{new}}(v)}\neq \emptyset$. Hence, because $\wh{(x-\chi)}$ remains unchanged, conditions {\bf (1)} and {\bf (2)} and {\bf (3)} hold at the end of the loop.
	
	Now, we show $\phi$ is decreased at least by one. By inductive hypothesis $\supp{(\wh{x}-\wh{\chi})} \cap {\subtree_{T}(v)}\neq \emptyset$ and at least one of $w$ or $u$ will be added to $T$ because $\subtree_T(v)=\subtree_{T'}(u)\cup \subtree_{T'}(v)$. Note that $l_{T_{new}}(w) = l_{T_{new}}(u) = l_{T}(v)+1$ hence $\sum_{u' : \text{ leaf of } T} l_{T_{new}}(u') \ge \sum_{u' : \text{ leaf of } T} l_{T}(u') + 1$. Because $\|\wh x - \wh \chi\|_0$ remains unchanged, the value of $\phi$ will decrease by at least one hence conditions {\bf (4)} and {\bf (d)} hold.
	
	Because $l_T(u) \le d\log_2n$ for every leaf $u \in T$, it follows from condition (2) that the quantity $\phi = (d\log_2n + 1)\|\wh x - \wh \chi\|_0 - \sum_{u : \text{ leaf of } T} l_T(u)$ is non-negative. At the first iteration, $\wh \chi=0$ and $ T = \{r\}$ where $r$ is the root with $l_T(f) = 0$. Hence, $\phi = \|\wh x\|_0(1+d\log_2n)$ at first iteration. Because $\phi$ is decreasing by at least 1 at each iteration, the algorithm terminates after $O(\|\wh x\|_0 \cdot d\log_2n)$ iterations. By Lemma \ref{lem:zero-test} along with Claim \ref{kraftsum}, the runtime of each iteration of algorithm is $O(km)$ and also sample complexity of each iteration is $O(k m)$ therefore the total runtime and sample complexity both will be $O(k^2 m \cdot d\log_2n)$.
	
	Finally, observe that throughout this analysis we have assumed that the set $\mathbf{ \Delta}$ satisfies the precondition of Lemma~\ref{lem:zero-test} for all the invocations of \textsc{ZeroTest} by our algorithm.

	In reality, there are two cases. The first case is for worst-case signals (Algorithm~\ref{alg:fullsparsefft}, Theorem~\ref{thm:sfft-worstcase}). In this case, the algorithm chooses $\mathbf{\Delta}$ to be a multiset which corresponds to the Fourier measurements with RIP of order $k$. Let $F^{-1}_N$ be the $d$ dimensional inverse Fourier transform's matrix with $N=n^d$ points. The matrix $M = {\sqrt{N}} F^{-1}_N$ is a unitary matrix. If you let $M_{\bm{\Delta}}$ denote the submatrix of $M$ whose rows are sampled from $M$ according to set $\bm{\Delta}$ defined in line~5 of Algorithm~\ref{alg:fullsparsefft} then by Theorem \ref{RIP-thrm} there exists a multiset $\bm{\Delta}$ of size $m= O(k\log_2^2k\cdot d\log_2n)$ such that $\frac{\sqrt{N}}{|\bm{\Delta}|}M_{\bm{\Delta}}$ satisfies the restricted isometry property of order $k$.
	Therefore, for every signal $y \in \C^{n^d}$:
	\[ \frac{1}{2} \cdot \frac{\|\wh{y} \|_2^2}{n^{2d}} \leq \frac{1}{|\mathbf{\Delta}|} \cdot \sum_{\Delta \in \mathbf{\Delta}} |y_\Delta|^2 \leq \frac{3}{2} \cdot \frac{\|\wh{y} \|_2^2}{n^{2d}} \]
	As we have shown in condition {\bf (5)} of the induction, $\|\wh x - \wh \chi\|_0 \le k$. Hence for every leaf $v$ of the tree $\|(\wh{x} - \wh{\chi})_{\subtree_T(v)}\|_0 \le k$ therefore the precondition of lemma \ref{lem:zero-test} is satisfied.

	The second case is for worst-case signals with random phase (Algorithm~\ref{alg:sparsefftphase}, Theorem~\ref{thm:sfft-rand-signs}). We have shown in condition {\bf (3)} of the induction that $x - \chi$ is a worst-case signal with random phase in every iteration of the algorithm. Therefore for every leaf $v$ of the tree it is true that $(\wh{x} - \wh{\chi})_{\subtree_T(v)}$ is a worst-case signal with random phase. In this case, the multiset $\bm{\Delta}$ is defined in line~5 of Algorithm~\ref{alg:sparsefftphase} therefore by Lemma~\ref{lem:random-sign} for a fixed leaf $v$ of tree $T$ with probability at least $1-1/n^{4d}$ the following holds:
	
	\[ \frac{1}{2} \cdot \frac{\|\wh{y} \|_2^2}{n^{2d}} \leq \frac{1}{|\mathbf{\Delta}|} \cdot \sum_{\Delta \in \mathbf{\Delta}} |y_\Delta|^2 \leq \frac{3}{2} \cdot \frac{\|\wh{y} \|_2^2}{n^{2d}} \]
	where $y=(\wh{x} - \wh{\chi})_{\subtree_T(v)}$.

	This shows that in the second case which corresponds to theorem \ref{thm:sfft-rand-signs}, the failure probability of procedure \textsc{ZeroTest} is at most $\frac{1}{n^{4d}}$. Moreover, the above analysis shows that \textsc{SFFT} makes at most $O(kd\log_2 n)$ calls to \textsc{ZeroTest}. Therefore, by a union bound, the overall failure probability of the calls to $\textsc{ZeroTest}$ is $O\left( (kd\log_2n) \frac{1}{n^{4d}} \right) \leq O(n^{-2d})$. Hence, we obtain the desired result.

\end{proofof}

\section{Signals with random support in high dimension}\label{sec:avgcase}
In this section, we consider Fourier sparse signals whose support in the frequency domain is chosen randomly, while the values at the nonzero frequencies are chosen arbitrarily (worst-case). In other words, we assume a Bernoulli model for $\supp \wh x$, while the values at the frequencies that are chosen to be in the support are arbitrary. We will present an algorithm that runs in time $O(k \log^{O(1)} N)$. The model for random support signals can be found in Definition \ref{def:randsupport} (Section~\ref{sec:overview}), which we restate here for convenience of the reader:

\defrandsupp*

\subsection{Outline of our approach}
The algorithm is motivated by the idea of speeding up our algorithm for worst-case signals (Algorithm~\ref{alg:fullsparsefft}, also see Theorem~\ref{thm:sfft-worstcase}) by reducing the number of iterations of the process from $\Theta(k)$ down to $O(\log N)$. Such a reduction (which we show to be impossible for worst-case signals in Section~\ref{hard-instance}) requires the ability to peel off many elements of the residual in a single phase of the algorithm, which turns out to be possible if the support of $\wh{x}$ is chosen uniformly at random as in Definition~\ref{def:randsupport}. Indeed, if one considers the splitting tree $T$ of a signal with uniformly random support, one sees that 
\begin{description}
	\item [{\bf (a)}] a large constant fraction of nodes $v\in T$ satisfy $w_T(v)\leq \log_2 k+O(1)$;
	\item [{\bf (b)}] the adaptive aliasing filters $G$ constructed for such nodes will have significantly overlapping support in time domain.
\end{description}

We provide the intuition for this for the one-dimensional setting ($d=1$) to simplify notation (changes required in higher dimensions are minor). In this setting, property {\bf (b)} above is simply a manifestation of the fact that since the support is uniformly random, any given congruence classes modulo $B'=Ck$ for a large enough constant $C>1$ is likely to contain only a single element of the support of $\wh{x}$. Our adaptive aliasing filters provide a way to only partition frequency space along a carefully selected subset of bits in $[\log_2 N]$, but due to the randomness assumption, one can isolate most of the elements by simply partitioning using the bottom $\log_2 k+O(1)$ bits. This essentially corresponds to hashing $\wh{x}$ into $B=Ck$ buckets at computational cost $O(B'\log B')=O(k\log k)$. While this scheme is efficient, it unfortunately only recovers a constant fraction of coefficients. One solution would be to hash into $B=Ck^2$ buckets (i.e., consider congruence classes modulo $Ck^2$), which would result in perfect hashing with good constant probability, allowing us to recover the entire signal in a single round. However, this hashing scheme would result in a runtime of $\Omega(k^2\log k)$ and is, hence, not satisfactory. On the other hand, hashing into $Ck^2$ buckets is clearly wasteful, as most buckets would be empty. Our main algorithmic contribution is a way of ``implicitly'' hashing into $Ck^2$ buckets, i.e., getting access to the nonempty buckets, at an improved cost of $\widetilde O(k)$.

Our algorithm uses an iterative approach, and the main underlying observation is very simple. Suppose that we are given the ability to ``implicitly'' hash into $B$ buckets for some $B$, namely, get access to the nonempty buckets. If $B$ is at least $\text{min}(Ck^2, N)$, we know that there are no collisions with high probability and we are done. If not, then we show that, given access to nonempty buckets in the $B$-hashing (i.e. a hashing into $B$ buckets), we can get access to the nonempty buckets of a $(\Gamma B)$-hashing for some appropriately chosen constant $\Gamma > 1$ at a polylogarithmic cost in the size of each nonempty bucket of the $B$-bucketing by essentially computing the Fourier transform of the signal restricted to nonempty buckets in the $B$-bucketing. We then proceed iteratively in this manner, starting with $B=Ck$, for which we can perform the hashing explicitly. Since the number of nonzero frequencies remaining in the residual after $t$ iterations of this process decays geometrically in $t$, we can also afford to use a smaller number of buckets $B'$ in the hashing that we construct explicitly, ensuring that the runtime is dominated by the first iteration.

Ultimately, the algorithm takes the following form. At every iteration, we explicitly compute  a hashing into $\bbase\leq Ck$ buckets explicitly. Then, using a list of nonempty buckets in a $\bprev$-hashing from the previous iteration, we extend this list to a list of nonempty buckets in a $\bnext$-bucketing at polylogarithmic cost per bucket (by solving a well-conditioned linear system, see Algorithm~\ref{alg:hashing2}), where $\bnext=\Gamma\cdot \bprev$ for some large enough constant $\Gamma>1$. Meanwhile, we reduce $\bbase$ by a factor of $\Gamma$, thus maintaining the invariant $\bbase\cdot \bnext \approx k^2$ at all times (note that this is satisfied at the start, when $\bbase =\bprev\approx k$, and $\bbase\cdot\bnext$ remains invariant at each iteration). Therefore, after a logarithmic number of iterations, we have effectively emulated hashing into $\approx k^2$ but at a total cost of roughly one hashing computation into $\approx k$ buckets (see Figure~\ref{fig:random-loc} for an illustration).  

\paragraph{Bucketing in high dimensions (\textsc{MakeBucket} function).} We note that our vectorial notation for buckets in high dimensions (see section~\ref{sec:vectorial}) allows us to continue talking about bucketings with $\bvec^{base}, \bvec^{prev}, \bvec^{next}$ buckets, even though now the number of buckets is in fact a vector of length $d$. In fact in dimension $d$ the only property of the bucketing that matters for our analysis is the number of buckets $|\bvec^{base}|, |\bvec^{prev}|, |\bvec^{next}|$ and the shape of each bucket is not important (this is due to the fact that the support is sampled from a permutation invariant distribution). In order to avoid unnecessary notation overload, in Algorithm \ref{alg:FFThighD-k1/r} we introduce procedure \textsc{MakeBucket} that constructs a bucketing $\mathbf{B}$ of size $|\mathbf{B}|=b$ of the following simple form. The vector $\bvec$ is defined by
$$
\bvec_p=\left\lbrace
\begin{array}{cc}
n&\text{~if~}p \le \lfloor \log_n b \rfloor\\
\frac{b}{n^p}&\text{~if~}p=\lfloor \log_n b \rfloor+1\\
1&\text{o.w.}
\end{array}
\right.
$$
The pseudocode for \textsc{MakeBucket}, which implements the formula above, is given in Algorithm~\ref{alg:FFThighD-k1/r}.

\subsection{Notation} \label{sec:vectorial}
We will need notations for vectorial operations, e.g., entrywise multiplication and/or division of vectors, which is defined in the following definition.

\begin{definition}[Entrywise vectorial arithmetic]
	Suppose that $\bv = (B_1, B_2, \cdots, B_d)$, $\jj=(j_1, j_2, \cdots, j_d)$ and $\tv=(t_1,t_2,\cdots,t_d)$ are $d$-dimensional vectors and $a$ is a scalar value. Then we define the following operations,
	\begin{center}
		\def\arraystretch{1.2}
		\begin{tabular}{  m{3cm}  m{15cm} }
			$\jj \cdot \tv$ & $d$-dimensional vector $(j_1 \cdot t_1, j_2 \cdot t_2 ,\cdots,j_d \cdot t_d)$.\\
			
			${\jj }/{ \tv}$ &$d$-dimensional vector $(j_1 / t_1, j_2 / t_2 ,\cdots,j_d / t_d)$.\\
			
			$a / \tv$ &$d$-dimensional vector $(a / t_1, a / t_2 ,\cdots, a / t_d)$.\\
			
			$\jj \equiv \tv \pmod{\bv}$  & entrywise congruence mod $\bv$, i.e., for all $q=1,\dots,d$, one has $j_q \equiv t_q \pmod{B_q}$.\\
			
			$\jj \pmod{\bv}$ &$d$-dimensional vector $\left( j_1 \pmod{B_1}, j_2 \pmod{B_2}, \cdots, j_d \pmod{B_d} \right)$.\\
			
			$\jj \pmod{a}$ &$d$-dimensional vector $\left( j_1 \pmod{a}, j_2 \pmod{a}, \cdots, j_d \pmod{a} \right)$.\\
			
			$| \jj |$ &the product of all the entries of vector $\jj$, i.e., $| \jj | = j_1 j_2 \cdots j_d$.\\
			
			$[\bv]$ &Cartesian product $[B_1] \times [B_2] \times \dots \times  [B_d]$.
		\end{tabular}
	\end{center}
\end{definition}

\subsection{Outline of our approach}
The algorithm is motivated by the idea of speeding up our algorithm for worst-case signals (Algorithm~\ref{alg:fullsparsefft}, also see Theorem~\ref{thm:sfft-worstcase}) by reducing the number of iterations of the process from $\Theta(k)$ down to $O(\log k)$. Such a reduction (which is shown to be impossible for worst-case signals in Section~\ref{hard-instance}) requires the ability to peel off many elements of the residual in a single phase of the algorithm, which turns out to be possible if the support of $\wh{x}$ is chosen uniformly at random as in Definition~\ref{def:randsupport}. Indeed, if one considers the splitting tree $T$ of a signal with uniformly random support (see Figure~\ref{fig:split-rand} for an illustration), one sees that 
\begin{description}
	\item [{\bf (a)}] a large constant fraction of nodes $v\in T$ satisfy $w_T(v)\leq \log_2 k+O(1)$;
	\item [{\bf (b)}] the adaptive aliasing filters $G$ constructed for such nodes will have significantly overlapping support in time domain.
\end{description}

We provide the intuition for this for the one-dimensional setting ($d=1$) to simplify notation (changes required in higher dimensions are minor). In this setting, property {\bf (b)} above is simply a manifestation of the fact that since the support is uniformly random, any given non-empty congruence class modulo $B'=Ck$ for a large enough constant $C>1$ is likely to contain only a single element of the support of $\wh{x}$. Our adaptive aliasing filters provide a way to only partition frequency space along a carefully selected subset of bits in $[\log_2 N]$, but due to the randomness assumption, one can isolate most of the elements by simply partitioning using the bottom $\log_2 k+O(1)$ bits. This essentially corresponds to hashing $\wh{x}$ into $B=Ck$ buckets at computational cost $O(B'\log B')=O(k\log k)$. While this scheme is efficient, it unfortunately only recovers a constant fraction of coefficients. One solution would be to hash into $B=Ck^2$ buckets (i.e., consider congruence classes modulo $Ck^2$), which would result in perfect hashing with good constant probability, allowing us to recover the entire signal in a single round. However, this hashing scheme would result in a runtime of $\Omega(k^2\log k)$ and is, hence, not satisfactory. On the other hand, hashing into $Ck^2$ buckets is clearly wasteful, as most buckets would be empty. Our main algorithmic contribution is a way of ``implicitly'' hashing into $Ck^2$ buckets, i.e., getting access to the nonempty buckets, at an improved cost of $\widetilde O(k)$.

Our algorithm uses an iterative approach, and the main underlying observation is very simple. Suppose that we are given the ability to ``implicitly'' hash into $B$ buckets for some $B$, namely, get access to the nonempty buckets. If $B$ is at least $\text{min}(Ck^2, N)$, we know that there are no collisions with high probability and we are done. If not, then we show that, given access to nonempty buckets in the $B$-hashing (i.e. a hashing into $B$ buckets), we can get access to the nonempty buckets of a $(\Gamma B)$-hashing for some appropriately chosen constant $\Gamma > 1$ at a polylogarithmic cost in the size of each nonempty bucket of the $B$-bucketing by essentially computing the Fourier transform of the signal restricted to nonempty buckets in the $B$-bucketing. We then proceed iteratively in this manner, starting with $B=Ck$, for which we can perform the hashing explicitly. Since the number of nonzero frequencies remaining in the residual after $t$ iterations of this process decays geometrically in $t$, we can also afford to use a smaller number of buckets $B'$ in the hashing that we construct explicitly, ensuring that the runtime is dominated by the first iteration.

Ultimately, the algorithm takes the following form. At every iteration, we explicitly compute  a hashing into $\bbase\leq Ck$ buckets explicitly. Then, using a list of nonempty buckets in a $\bprev$-hashing from the previous iteration, we extend this list to a list of nonempty buckets in a $\bnext$-bucketing at polylogarithmic cost per bucket (by solving a well-conditioned linear system, see Algorithm~\ref{alg:hashing2}), where $\bnext=\Gamma\cdot \bprev$ for some large enough constant $\Gamma>1$. Meanwhile, we reduce $\bbase$ by a factor of $\Gamma$, thus maintaining the invariant $\bbase\cdot \bnext \approx k^2$ at all times (note that this is satisfied at the start, when $\bbase =\bprev\approx k$, and $\bbase\cdot\bnext$ remains invariant at each iteration). Therefore, after a logarithmic number of iterations, we have effectively emulated hashing into $\approx k^2$ but at a total cost of roughly one hashing computation into $\approx k$ buckets (see Figure~\ref{fig:random-loc} for an illustration).  

\paragraph{Bucketing in high dimensions (\textsc{MakeBucket} function).} We note that our vectorial notation for buckets in high dimensions (see section~\ref{sec:vectorial}) allows us to continue talking about bucketings with $\bvec^{base}, \bvec^{prev}, \bvec^{next}$ buckets, even though now the number of buckets is in fact a vector of length $d$. In fact in dimension $d$ the only property of the bucketing that matters for our analysis is the number of buckets $|\bvec^{base}|, |\bvec^{prev}|, |\bvec^{next}|$ and the shape of each bucket is not important (this is due to the fact that the support is sampled from a permutation invariant distribution). In order to avoid unnecessary notation overload, in Algorithm \ref{alg:FFThighD-k1/r} we introduce procedure \textsc{MakeBucket} that constructs a bucketing $\mathbf{B}$ of size $|\mathbf{B}|=b$ of the following simple form. The vector $\bvec$ is defined by
$$
\bvec_p=\left\lbrace
\begin{array}{cc}
n&\text{~if~}p \le \lfloor \log_n b \rfloor\\
\frac{b}{n^p}&\text{~if~}p=\lfloor \log_n b \rfloor+1\\
1&\text{o.w.}
\end{array}
\right.
$$
The pseudocode for \textsc{MakeBucket}, which implements the formula above, is given in Algorithm~\ref{alg:FFThighD-k1/r}.

\subsection{Filtering, hashing, and bucketing in high dimensions}

We introduce the main definitions here. Our techniques in this section use a version of our adaptive aliasing filters that is taylored to the assumption that the support of $\wh{x}$ is chosen uniformly at random. Since the signal is assumed to be sampled from a distribution, we are able to design a fast algorithm by adapting to a distribution as opposed to a given realization of the support of $\wh{x}$.  The next definition is essentially a simplified version of the definition of a frequency cone from Section~\ref{sec:overview} (see Definition~\ref{def:iso-t}):

\begin{definition}[Congruence classes of support]\label{def:congruence-class}
	Suppose $d$ and $n$ are positive integers such that $n$ is a power of two. Let $\bv = (B_1, B_2, \dots, B_d)$ be a vector of powers of two such that $B_j \mid n$ for $j=1,2,\dots,d$. For every $\bb\in [\bv]$, and signal $x\in \C^{n^d}$, we define the \emph{$(\bv, \bb)$-congruence class} of $\supp{ \wh{x}}$ to be the set $S_x(\bv, \bb)$, given by
	\[
	S_x(\bv, \bb) = \{ \ff \in \supp{\wh x} : \ff \equiv\bb \pmod{\bv}\} .
	\]
\end{definition}

We access the signal using a bucketing operation, defined  below. 

\begin{definition}[Bucketing in high dimensions] \label{def:bucketing-highdim}
	Suppose $d$ and $n$ are positive integers such that $n$ is a power of two. Let $\bv = (B_1, B_2, \dots, B_d)$ be a vector of powers of two such that $B_j \mid n$ for all $j=1,2,\dots,d$. For every $\aaa \in [n]^d$, $\bb\in [\bv]$, and signal $x\in \C^{n^d}$, we define the \emph{$(\bv, \bb)$-bucketing} of $x$ with \emph{shift} $\aaa$ to be $U_x^\aaa(\bv, \bb)$, given by
	\begin{equation*}
	\begin{split}
	U_x^\aaa(\bv, \bb) &= \sum_{\ff\equiv\bb \pmod{\bv}} \wh x_{\ff} \cdot e^{2\pi i \frac{\ff^T \aaa}{n}}\\
	&= \sum_{\ff\in S_x(\bv, \bb)} \wh x_{\ff} \cdot e^{2\pi i \frac{\ff^T \aaa}{n}}.
	\end{split}
	\end{equation*}
\end{definition}

The following definition of  Bernoulli set provides a compact way of referring to the distribution of $\text{supp~}\wh{x}$:
\begin{definition}[Bernoulli set]\label{bernoulli-set}
	For every power of two $n$ and positive integer $d$, let $N=n^d$ and set $S\subseteq [n]^d$ be a random set such that each $\jj\in [n]^d$ is independently chosen to be in $S$ with probability $k/N$,
	$$\Pr[ \jj \in S ] = \frac{k}{N}.$$
	Moreover, for any $\bvec = (B_1, B_2, \dots, B_d)$ such that $B_1, B_2, \dots, B_d \mid n$, we define
	\begin{align*}
	S^{(\bvec)} = \{ \ff \in [\bvec]: \exists\,\gg, \hh \in S \text{ such that } \gg\neq\hh, \ff\equiv\gg\equiv\hh \pmod{\bvec} \}.
	\end{align*}
\end{definition}

The next lemma is crucial for our analysis. The lemma considers two bucketings $\bvec$ and $\bvec'$, where $\bvec$ is a refinement of $\bvec$. The object of interest is the number of buckets in the bucketing $\bvec$ that contain at least two elements of a Bernoulli set $S$, i.e. {\bf non-singleton buckets}. The lemma shows that as long as the product of the number of buckets in $\bvec$ and $\bvec'$ is at least $k^2$,  the elements (i.e. frequencies) of a Bernoulli set $S$ that belong to non-singleton buckets in $\bvec$ must be rather uniformly spread over the coarser bucketing $\bvec'$. Specifically, no bucket in $\bvec'$ contains more than $O(\log N)$ such frequencies with high probability. We will use this lemma with $\bvec'=\bvec^{base}$ and $\bvec=\bvec^{next}$ and $\bvec=\bvec^{prev}$ (see proof of Theorem~\ref{thrm:sparse-fft} below).

\begin{lemma}[Refinement lemma]\label{lem:liftsize}
	For any power of two integers $n$ and $k$, suppose $\bvec = (B_1, B_2, \dots B_d)$ and $\bvec' = (B_1', B_2', \dots, B_d')$ satisfy $B_j' \mid B_j \mid n$ for all $j = 1,2,\dots, d$ as well as $|\bv| \cdot |\bv'| \geq k^2$. Then, with probability at least $1 - \frac{1}{N^3}$, we have that for all $\bb\in [\bvec']$,
	\[ \left| S^{(\bvec)} \cap \{\ff\in [\bvec]: \ff \equiv \bb \pmod{\bvec'}\} \right| = O(\log N). \]
\end{lemma}
The proof is a simple probabilistic argument and is given in Appendix~\ref{app:bernoulli}.

\subsection{Hashing in high dimensions by using downsampling}
The main primitive we need for developing an algorithm with $\tilde{O}(k)$ sample complexity is a hashing function based on downsampling, presented below as Algorithm~\ref{alg:hashing2}. The algorithm takes as input the list $R$ of buckets in the hashing into $\bprevv$ (that will later be guaranteed to be superset of the nonempty buckets in a $\bprevv$-hashing of the residual signal $x-\chi$) and outputs a list of potentially nonempty buckets in the hashing into $\bnextv$, together with evaluations of the corresponding hashed signals at a point $\alphav$ that is given as a parameter.

\begin{algorithm}
	\caption{Procedure for hashing a signal's Fourier transform using downsampling}
	\begin{algorithmic}[1]
		
		\Procedure{Hashing}{$x, \wh \chi, n, d, \bbasev, \bprevv, \bnextv, \alphav, R$} \Comment $\alphav$ is a point in time domain
		\State \Comment{$R$ is the set of nonempty buckets in $\bprevv$}
		\State $m \gets \frac{|\bnextv|}{|\bprevv|} \max_{\bb \in [\bbasev]} \left|\left\{ \rr \in R : \rr \equiv \bb \pmod{\bbase} \right\}\right|$
		\State $S \gets \left\{ \betav_i: \betav_i \sim \unif([\bnextv]),  \forall i \in [C m \log^2_2m \cdot d \log_2 n] \right\}$ \Comment $S$: multiset 
		\Comment $C$: constant

		\For{$\betav\in S$}
		\State $\aaa \gets \alphav + \betav\cdot \frac{n}{\bnextv}$
		\State $Z^{(\alphav,\betav)}_\bb \gets \frac{n^d}{|\bbasev|} x_{\bb\cdot\frac{n}{\bbasev} + \aaa} \text{ for all } \bb\in [\bbasev]$
		\State $\wh Z^{(\alphav,\betav)} \gets FFT\left(Z^{(\alphav, \betav)}\right)$
		\State $\wh Z^{(\alphav,\betav)}_\bb \gets \wh Z^{(\alphav,\betav)}_\bb - \sum_{\rr\in \left[ \frac{n}{\bbasev}\right]} \wh\chi_{\bb + \rr\cdot\bbasev} e^{2\pi i \frac{\left(\bb + \rr\cdot\bbasev\right)^T \aaa}{n}} \text{ for all } \bb\in[\bbasev]$
		\EndFor
		
		\State $W\gets\emptyset;$
		
		\For{\texttt{$\bb \in [\bbase]$}}
		\State $w_\betav \gets \wh Z_\bb^{(\alphav, \betav)} \text{ for all } \betav\in S$
		\State $R_\bb \gets \{(\rr, \sv): \bb+\rr \cdot \bbasev \in R \text{ and } \sv\in [\bnextv/\bprevv]\}$
		\State $\left(\matA_\bb\right)_{\betav, (\rr, \sv)} \gets e^{2\pi i (\bb+\rr\bbasev+\sv\bprevv)^T  \left(\frac{\betav}{\bnextv}\right)} \text{ for all } \betav\in S \text{ and } (\rr, \sv)\in R_\bb$
		\State $\vv \gets \textsc{LeastSquaresSolver}(\matA_\bb, \ww)$
		
		\For{$(\rr,\sv) \in R_\bb$}
		\State $W\gets W \cup \left\{\left(\bb + \rr\bbasev + \sv\bprevv, \vv_{(\rr, \sv)}^{(\bb)}\right)\right\}$
		\EndFor
		
		\EndFor
		
		\State \textbf{return} $W$
		
		\EndProcedure

		\Procedure{LeastSquaresSolver}{$A, b$} 

		\State \textbf{return} $(A^TA)^{-1} A^Tb$
		
		\EndProcedure
	\end{algorithmic}
	\label{alg:hashing2}
\end{algorithm}

\begin{lemma}
	\emph{(\textsc{Hashing} in high dimensions)} \label{lem:hashing2-highdim}
	Suppose $d$ and $n$ are positive integers such that $n$ is a power of two. Suppose $\bbasev = (\bbase_1, \bbase_2, \dots, \bbase_d)$, $\bprevv = (\bprev_1, \bprev_2, \dots, \bprev_d)$ and $\bnextv = (\bnext_1, \bnext_2, \dots, \bnext_d)$ are vectors of powers of two such that $\bprev_j \mid \bnext_j$ and $\bbase_j \mid \bprev_j$ for all $j$. Moreover, let $\alphav  \in [n]^d$ be a shift vector. For any signals $x , \wh{\chi} \in \C^{n^d}$ suppose that
	\begin{align}
	R \supseteq \{\bb\in[\bprevv]:  S_{x - \chi }(\bprevv,\bb) \neq \emptyset \}, \label{eq:rcond}
	\end{align}
	where $S_{x - \chi}(\bprevv,\bb)$ is $(\bprevv,\bb)$-congruence class of $\supp{ \wh{(x-\chi)}}$ as per Definition \ref{def:congruence-class}. Then the procedure \textsc{Hashing} $(x, \wh{\chi}, n, d, \bbasev, \bprevv, \bnextv, \alphav, R)$ outputs a set $W$ that with probability $1 - \frac{1}{n^{10d}}$ satisfies
	\begin{align}
	W = \{(\bb', U_{x-\chi}^\alphav (\bnextv, \bb')): \bb'\in [\bnextv] \text{ s.t. } \, \bb'\equiv\bb \pmod{\bprevv} \text{ for some } \bb\in R \}. \label{eq:tcond}
	\end{align}
	Moreover, the sample complexity of this procedure is 
	\[ O\left( m\log^2 m \cdot d\log n \cdot \left| \bbasev \right| \right), \]
	while the time complexity of this procedure is
	\[ O\left( (m\log^2 m \cdot d\log n)^3 \cdot \left| \bbasev \right| + (m\log^2 m \cdot d\log n) \cdot \left(\left| \bbasev \right|\log_2\left( | \bbasev| \right) + \|\wh \chi\|_0 \right) \right), \]
	where $m = \frac{|\bnextv|}{|\bprevv|} \max_{\bb \in [\bbasev]} \left|\left\{ \rr \in R : \rr \equiv \bb \pmod{\bbase} \right\}\right| $.
\end{lemma}

\begin{proof}
	Recall that the algorithm uses a coarse $\bbasev$-bucketing to refine a $\bprevv$-bucketing, for which only the non-empty buckets are computed, to a $\bnextv$-bucketing. 
	Let $x' = x-\chi$ and $\aaa = \alphav + \betav \cdot \frac{n}{\bnextv}$, where $\alphav = \aaa\bmod \frac{n}{\bnextv}$. Suppose \eqref{eq:rcond} holds.
	
	Note that for any $\bb\in [\bbasev]$,
	\begin{align}
	U_{x'}^\aaa (\bbasev, \bb) &= \sum_{\ff\equiv\bb \pmod{\bbasev}} \wh x'_\ff e^{2\pi i \frac{\ff^T \aaa}{n}} \nonumber\\
	&= \sum_{\rr\in \left[\frac{\bprevv}{\bbasev}\right]} \sum_{\sv\in \left[\frac{\bnextv}{\bprevv}\right]} \sum_{\ff \equiv \phib(\bb, \rr, \sv) \pmod{\bnextv}} \wh x'_\ff e^{2\pi i \frac{\ff^T \aaa}{n}} \nonumber\\
	&= \sum_{\rr\in \left[\frac{\bprevv}{\bbasev}\right]} \sum_{\sv\in \left[\frac{\bnextv}{\bprevv}\right]} e^{2\pi i \phib(\bb,\rr,\sv)^T \left(\frac{\betav}{\bnextv}\right)} \sum_{\qq \in \left[\frac{n}{\bnextv}\right]} \wh x'_{\phib(\bb,\rr,\sv)+\qq\cdot \bnextv} e^{2\pi i \frac{(\phib(\bb,\rr,\sv) + \qq\cdot \bnextv)^T\alphav}{n}} \nonumber\\
	&= \sum_{\rr\in \left[\frac{\bprevv}{\bbasev}\right]} \sum_{\sv\in \left[\frac{\bnextv}{\bprevv}\right]} e^{2\pi i \phib(\bb,\rr,\sv)^T \frac{\betav}{\bnextv}} \cdot U_{x'}^{\alphav}(\bnextv, \phib(\bb,\rr,\sv)), \label{eq:buckettransform}
	\end{align}
	where $\phib(\bb, \rr, \sv)$ denotes $\bb + \rr\bbasev + \sv\bprevv \in [\bnextv]$ for ease of notation. Hence, $\bbasev$-bucketing and $\bnextv$-bucketing are related via a linear system which is defined in \eqref{eq:buckettransform} for any collection of values of $\betav$. We also show how to choose a relatively small number of values of $\betav$ such that the above linear system will be well-conditioned. Note that $\frac{\bnextv}{\bprevv}$ is a vector of length $d$, as per our vectorial notations in section \ref{sec:vectorial}, which denotes by how much we want to further refine each bucket of $\bprevv$-bucketing.
	
	\paragraph{Choosing $\betav$'s that make the linear system in \eqref{eq:buckettransform} well-conditioned:}
	For every $\bb \in [\bbasev]$, let $\matA_\bb$ denote the $\left| \bnextv \right| \times \left| \frac{\bnextv}{\bbasev} \right|$ matrix whose rows are indexed by $\betav \in [\bnextv]$ and columns are indexed by $(\rr, \sv)\in [\frac{\bprevv}{\bbasev}]\times[\frac{\bnextv}{\bprevv}]$, with entries defined by
	\[
	(\matA_\bb)_{\betav, (\rr, \sv)} = e^{2\pi i \phib(\bb,\rr,\sv)^T  \left(\frac{\betav}{\bnextv}\right)}.
	\]
	Moreover, let $\vv_b$ denote the column vector of length $\left| \frac{\bnextv}{\bbasev} \right|$ with entries indexed by $(\rr, \sv)\in [\frac{\bprevv}{\bbasev}]\times[\frac{\bnextv}{\bprevv}]$ and given by
	\begin{equation}
	(\vv_\bb)_{(\rr,\sv)} = U_{x'}^{\alphav}(\bnextv, \phi(\bb, \rr, \sv)), \label{next-bucketing}
	\end{equation}
	while we let $\ww_\bb$ denote the column vector of length $\left|\bnextv\right|$ with entries indexed by $\betav \in [\bnextv]$ and given by
	\begin{equation}
	\ww_{\betav} = U_{x'}^{\alphav + \betav\cdot\frac{n}{\bnextv}}(\bbasev, \bb). \label{base-bucketing}
	\end{equation}
	Then, \eqref{eq:buckettransform} implies the following linear system of equations,
	\begin{equation}
	\matA_\bb \cdot \vv^{(\bb)} = \ww^{(\bb)}. \label{linear-sys}
	\end{equation}

	Next, for any $\bb\in[\bbasev]$, let
	\[
	R_\bb = \{(\rr, \sv): \bb+\rr\bbasev \in R \text{ and } \sv\in [\bnextv/\bprevv]\}.
	\]
	By the assumption of the lemma on the set $R$, one has that the vectors ${\vv_{\bb}}$ in \eqref{next-bucketing} satisfy $\supp{\vv_{\bb}} \subseteq R_\bb$ for every $\bb \in [\bbasev] $. This shows that the linear system in \eqref{linear-sys} can be solved very efficiently by randomly sampling its rows. 	
	More formally, suppose that $S$ is a multiset such that $S = \left\{ \betav_i: \betav_i \sim \unif([\bnextv]),  \forall i \in [C m \log^2_2m \cdot d \log_2 n] \right\}$ where $m = \max_{\bb' \in [\bbasev]} |R_{\bb'}|$ and let $\matA^{S}_\bb$ denote the submatrix of $\matA_\bb$ whose rows are selected with respect to set $S$. Then by Theorem \ref{RIP-thrm}, the matrix $\frac{1}{|S|} \matA_\bb^S$ satisfies RIP of order $m$. We will use this property to solve the system \eqref{linear-sys} efficiently.
	
	Let $\wt{\matA}_\bb$ be a submatrix of $\matA_\bb$ with size $|S|\times |R_\bb|$. Suppose that its rows are selected with respect to set $S$ and its columns are selected with respect to $R_\bb$. More specifically, $\wt{\matA}_\bb$ is a $|S|\times |R_\bb|$ matrix whose rows are indexed by $\betav\in S$ and columns are indexed by $(\rr, \sv)\in R_\bb$, with entries defined by
	\[
	(\wt \matA_\bb)_{\betav, (\rr, \sv)} = e^{2\pi i \phib(\bb,\rr,\sv)^T  \left(\frac{\betav}{\bnextv}\right)}.
	\]
	
	Moreover, let $\vv^{(\bb)}$ denote the column vector of length $|R_\bb|$ with entries indexed by elements of $R_\bb$ and given by
	\begin{equation}
	\vv^{(\bb)}_{(\rr,\sv)} = U_{x-\chi}^{\alphav}(\bnextv, \phi(\bb, \rr, \sv)), \label{eq:vvdef}
	\end{equation}
	while we let $\ww^{(\bb)}$ denote the column vector of length $|S|$ with entries indexed by elements of $S$ and given by
	\[
	\ww^{(\bb)}_{\betav} = U_{x-\chi}^{\alphav + \betav\cdot\frac{n}{\bnextv}}(\bbasev, \bb).
	\]
	Then, \eqref{linear-sys} implies that
	\begin{equation}
	\wt\matA_\bb \cdot \vv^{(\bb)} = \ww^{(\bb)}. \label{linearsys-reduced}
	\end{equation}
	
	Now, because the matrix 
	$$\frac{1}{|S|} \matA_\bb^S$$  
	satisfies RIP of order $|R_\bb|$ for all $\bb \in [\bbasev]$, one has that the condition number of $\wt{ \matA_\bb}^T \wt{ \matA_\bb}$ is at most $\sqrt{3}$. Therefore, a linear least squares solver can compute $\vv^{(\bb)}$ efficiently and in a numerically stable way using the reduced linear system in \eqref{linearsys-reduced}. Note that lines 11-14 of Algorithm~\ref{alg:hashing2} carry out this procedure and compute $\vv^{(\bb)}$ for each $\bb \in [\bbasev]$.
	
	\paragraph{Computing $U^\aaa_{x-\chi}(\bbasev,\bb)$:}
	Now we show how to compute $U_{x-\chi}^\aaa (\bbasev, \bb)$ for any $\bb\in [\bbasev]$ and $\aaa = \alphav + \betav\cdot\frac{n}{\bnextv}$. By standard downsampling properties, we have that if $Z^{(\alphav,\betav)}: [\bbasev] \to \C$ is the signal defined by
	\[
	Z^{(\alphav,\betav)}_{\tt} = \frac{n^d}{|\bbasev|} x_{\tt \cdot \frac{n}{\bbasev} + \aaa}
	\]
	then its Fourier transform is given by
	\[
	\wh Z^{(\alphav, \betav)}_{\bb} = \sum_{\ff\equiv\bb \pmod{\bbase}} \wh x_\ff \cdot e^{2\pi i \frac{\ff^T \aaa}{n}} = U_x^\aaa (\bbase, \bb).
	\]
	Hence,
	\begin{align}
	U_{x-\chi}^{\aaa}(\bbase, \bb) = \wh Z_\bb^{(\alphav, \betav)} - \sum_{\ff\equiv\bb \pmod{\bbase}} \wh \chi_\ff e^{2\pi i \frac{\ff^T \aaa}{n}}, \label{eq:bucketcompute}
	\end{align}
	which demonstrates how to compute $U_{x-\chi}^\aaa (\bbasev, \bb)$. Note that lines 6-8 of Algorithm~\ref{alg:hashing2} simply compute  $U_{x-\chi}^\aaa (\bbase, \bb)$, for some $\bb\in[\bbase]$ with $\aaa = \alphav + \betav \cdot \frac{n}{\bnextv}$ for some $\betav$.

	\paragraph{Sample complexity and Runtime:}
	Lines 6-8 of Algorithm~\ref{alg:hashing2} compute  $U_{x-\chi}^\aaa (\bbase, \bb)$, for some $\bb\in[\bbase]$ with $\aaa = \alphav + \betav \cdot \frac{n}{\bnextv}$ for some $\betav$, in time $O(\left| \bbasev \right|\log_2\left( | \bbasev| \right) + \|\wh \chi\|_0)$ and with sample complexity $O\left( \left| \bbasev \right| \right)$, according to the rule \eqref{eq:bucketcompute}. This shows that the vector $\ww_\bb$ in \eqref{base-bucketing} can be constructed efficiently.
	
	Note that lines 11-14 of Algorithm~\ref{alg:hashing2} carry out a least squares linear system procedure and compute $\vv^{(\bb)}$ for each $\bb \in [\bbasev]$ in time $O(|S|^3)$, as the time complexity of \textsc{LeastSquaresSolver} procedure is $O(|S|^3)$. 
	Moreover, by \eqref{eq:vvdef}, it follows that for a fixed $\bb\in [\bbasev]$, line 16 simply adds all pairs $(\bb', U_{x-\chi}^\alphav(\bnextv, \bb'))$ with $\bb' \in[\bnextv]$ satisfying $(\bb' \bmod \bprev) \in R$. Also, any $\bb'\in [\bnextv]$ for which there exists $\ff \in \supp \wh{x-\chi}$ with $\ff\equiv\bb' \pmod{\bnextv}$ must satisfy $\ff\equiv\bb' \pmod{\bprevv}$ and, hence, $S_{x-\chi}(\bprevv, \bb' \pmod{\bprevv}) \neq \emptyset$. This shows that $(\bb' \bmod \bprevv) \in R$ (by \eqref{eq:rcond}). Thus, it follows that after looping over all $\bb \in [\bbasev]$, the final $W$ satisfies \eqref{eq:tcond}, as desired.
	
	Note that the sample complexity of Algorithm~\ref{alg:hashing2} is determined by the total number of samples required to construct the various $Z^{(\alphav, \betav)}$. For any fixed $\betav \in S$, constructing $Z^{(\alphav, \betav)}_\jj$ requires $|\bbasev|$ samples from $X$. Since there are $|S|$ values of $\betav$ that are relevant, it follows that the total sample complexity is
	\[
	O\left( |S| \cdot \left| \bbasev \right| \right) = O\left( m\log^2 m \cdot d\log n \cdot \left| \bbasev \right| \right).
	\]
	
	The time complexity of this procedure is due to two computations. First, constructing each $\wh Z^{(\alphav, \betav)}$ for a fixed $\betav$ takes time $O(\left| \bbasev \right|\log_2\left( | \bbasev| \right) + \|\wh \chi\|_0)$. Second, computing the $\vv^{(\bb)}$ vector for each fixed $\bb \in [\bbasev]$ requires time $O(|S|^3)$. Therefore, the total time complexity is
	\[ O\left( |S|^3 \cdot \left| \bbasev \right| + |S| \cdot \left(\left| \bbasev \right|\log_2\left( | \bbasev| \right) + \|\wh \chi\|_0\right) \right). \]
\end{proof}

\subsection{Resolving buckets in the hashed signal}
The other major building block we need for developing a sparse FFT algorithm is a function for testing bucketings of signals with various shifts for emptyness and one-sparsity. Such a primitive takes in a list of buckets of a hashed signal and determines whether each bucket is empty or not. If a bucket is not empty, then we determine whether the bucket consists of exactly one frequency using a one-sparse test. If so, we can determine this frequency and the value of the signal at this frequency from the bucketed signals. If not, then we retain the bucket for the next iteration, in which we will hash to more buckets.

\begin{algorithm}
	\caption{Procedure for testing a hashed signal}
	\begin{algorithmic}[1]
		
		\Procedure{TestBuckets}{$\{W_\alphav\}_{\alphav\in\shiftset \cup \{ e_1,\dots,e_d \}}, n, d, \bv$}\Comment $1$-sparse test and zero test
		
		\Comment $e_1,\dots,e_d$ are standard basis vectors in $[n]^d$
		
		\State $  \wh\chi  \gets \{ 0 \}^{n^d};$ 
		\State $R \gets \emptyset$
		\For{$\bb \in \dom\left( W_\alphav\right)$} \Comment $\dom\left( W_\alphav\right)$: set of all first coordinates in $W_\alphav$
		\If{\texttt{$ \sum_{\alphav \in \shiftset} |W_\alphav(\bb)|^2 > 0 $}} \Comment Zero test on $U_{x} (\bv, \bb)$
		
		\State $f_q \gets \frac{n}{2\pi} \cdot \phi \left( \frac{W_{e_q}(\bb)}{W_{0}(\bb)} \right)$ for every $q \in \{1,...,d\}$ \Comment $e_1, ..., e_d$ are standard bases
		\State $v \gets W_{0}(\bb)$
		
		\If{\texttt{$ \sum_{\alphav \in \shiftset} |v e^{2\pi i \frac{\ff^T \alphav}{n}} - W_\alphav(\bb)|^2 = 0 $}} \Comment One sparse test on $U_{x} (\bv, \bb)$
		\State $\wh\chi_\ff \gets v;$
		\Else
		\State $R \gets R \cup \{\bb\};$
		\EndIf
		\EndIf
		\EndFor
		
		%
		%
		
		\State \textbf{return} $(\wh\chi , R)$
		\EndProcedure
		
	\end{algorithmic}
	\label{alg:resolve-buckets}
	
\end{algorithm}

\begin{lemma}
	\emph{(\textsc{TestBuckets} in high dimensions)} \label{lem:resolve-buckets-highdim}
	Suppose $d$ and $n$ are positive integers such that $n$ is a power of two. Suppose $\bv = (B_1, B_2, \dots, B_d)$ is a vector of powers of two such that $B_j \mid n$ for all $j$. Suppose $x \in \C^{n^d}$ is a signal such that  $W_\alphav(\bb) = U_x^\alphav(\bv, \bb)$ is a $(\bv, \bb)$-bucketing of $x$ with shift $\alphav$ for all $\alphav\in\shiftset$ and $\bb\in \dom(W_\alphav)$ where $\shiftset$ is a multiset of $q$ i.i.d samples from $\unif([n]^d)$ for some 
	$$q = \Omega\left( \max_{\bb \in [\bv]} |S_x(\bv,\bb)| \cdot \log^2 (\max_{\bb \in [\bv]} |S_x(\bv,\bb)|)\cdot (d \log n) \right),$$
	and $S_x(\bv,\bb)$ for all $\bb \in [\bv]$ are Congruence classes of $\supp{\wh x}$. Also suppose that Algorithm \ref{alg:resolve-buckets} takes in the quantities $W_\alphav(\bb)$ for all $\alphav \in \{e_1, \dots, e_d\}$, standard basis vectors in $[n]^d$. Then, \textsc{TestBuckets}$(\{W_\alphav\}_{\alphav\in\shiftset\cup \{ e_1,\dots,e_d \} }, n,d,\bv)$ returns $\wh\chi$ and $R$ that with probability $1-\frac{1}{n^{10d}}$ satisfy the following:
	\begin{itemize}
		\item For any $\bb\in \dom(W_\alphav)$ such that $S_x(\bv,\bb)$ is a singleton set, $S_x(\bv,\bb)=\{\ff\}$, we have $\wh\chi_\ff = \wh x_\ff$.
		\item We have $R = \left\{\bb\in\dom(W_\alphav): |S_x(\bv,\bb)| \ge 2 \right\}$.
	\end{itemize}
	Moreover, the runtime of this procedure is $O \left( |\shiftset| \cdot |\dom(W_\alphav)| \right)$.
\end{lemma}

\begin{proof}
	Let $F^{-1}_N$ be the $d$ dimensional inverse Fourier transform's matrix with $N=n^d$ points. The matrix $M = {\sqrt{N}} F^{-1}_N$ is a unitary matrix and all of its elements have absolute value $\frac{1}{\sqrt{N}}$. If you let $M_{\shiftset}$ denote the submatrix of $M$ whose rows are sampled from $M$ according to set $\shiftset$ then by Theorem \ref{RIP-thrm}, $\frac{\sqrt{N}}{|\shiftset|}M_\shiftset$ satisfies the restricted isometry property of order $\max_{\bb \in [\bv]}|S_x(\bv,\bb)|+1$ with probability $1- \frac{1}{N^{10}}$. In the rest we condition on the event corresponding to matrix $\frac{\sqrt{N}}{|\shiftset|}M_\shiftset$ satisfying RIP of order $\max_{\bb \in [\bv]}|S_x(\bv,\bb)|+1$.
	
	Now, note that by definition of $U_x(\bv,\bb)$ and $S_x(\bv,\bb)$ (definitions \ref{def:bucketing-highdim} and \ref{def:congruence-class} respectively) we have,
	\begin{align*}
	W_\alphav(\bb) &= U_x^\alphav (\bv,\bb) \\
	&= \sum_{\ff \in S_x(\bv,\bb)} \wh x_{\ff} \cdot e^{2\pi i \frac{\ff^T \aaa}{n}}\\
	&= \left( (N F_N^{-1}) \cdot \wh{x}_{S_x(\bv,\bb)} \right)_{\alphav} 
	\end{align*}
	therefore, for every $\bb \in \dom(W_\alphav)$ the following holds true,
	$$ \frac{|\shiftset|}{2} \|\wh{x}_{S_x(\bv,\bb)}\|_2^2  \le \sum_{\alphav \in \shiftset} |W_\alphav(\bb)|^2 \le \frac{3|\shiftset|}{2}\|\wh{x}_{S_x(\bv,\bb)}\|_2^2$$
	thus the zero test in line~5 of Algorithm \ref{alg:resolve-buckets} works correctly for all buckets.
	
	Now note that if $S_x(\bv,\bb)$ is a singleton set $\{\ff\}$ then,
	\begin{align*}
	W_\alphav(\bb) &= U_x^\alphav (\bv,\bb) \\
	&= \sum_{\jj \in S_x(\bv,\bb)} \wh x_{\jj} \cdot e^{2\pi i \frac{\jj^T \aaa}{n}}\\
	&= \sum_{\jj \in \{\ff\}} \wh x_{\jj} \cdot e^{2\pi i \frac{\jj^T \aaa}{n}}\\
	&= x_{\ff} \cdot e^{2\pi i \frac{\ff^T \aaa}{n}}.
	\end{align*}
	Therefore, for every $q=1,2,\dots,d$, the following holds,
	$$\frac{W_{e_q}(\bb)}{W_{0}(\bb)} = e^{2\pi i \frac{\ff^T e_q}{n}} = e^{2\pi i \frac{f_q}{n}}$$
	thus, $\frac{n}{2\pi} \cdot \phi \left( \frac{W_{e_q}(\bb)}{W_{0}(\bb)} \right) = f_q$. Also note that,
	$W_{0}(\bb) = \wh x_\ff$. This is precisely implemented in line~6-7 of Algorithm \ref{alg:resolve-buckets}. On the other hand if the hypothesis that $S_x(\bv,\bb)$ is a singleton set is incorrect our algorithm will find it. Using the notation $v = W_{0}(\bb)$ as in line~7 of Algorithm \ref{alg:resolve-buckets},
	\begin{align*}
	W_\alphav(\bb) - v e^{2\pi i \frac{\ff^T \aaa}{n}}
	&= \sum_{\jj \in S_x(\bv,\bb)} \wh x_{\jj} \cdot e^{2\pi i \frac{\jj^T \aaa}{n}} - v e^{2\pi i \frac{\ff^T \aaa}{n}}\\
	&= \sum_{\jj \in S_x(\bv,\bb) \cup \{\ff\}} \wh x'_{\jj} \cdot e^{2\pi i \frac{\jj^T \aaa}{n}}\\
	&= \left( (N F_N^{-1}) \cdot \wh{x}'_{S_x(\bv,\bb) \cup \{\ff\}} \right)_{\alphav} 
	\end{align*}
	where, $\wh x'\textbf{} = \wh x(\cdot) - v \delta_\ff(\cdot)$. Because, $\wh{x}'_{S_x(\bv,\bb) \cup \{\ff\}}$ is at most $\max_{\bb \in [\bv]}|S_x(\bv,\bb)|+1$ sparse and matrix $\frac{\sqrt{N}}{|\shiftset|}M_\shiftset$ satisfies RIP of order $\max_{\bb \in [\bv]}|S_x(\bv,\bb)|+1$ we have that,
	
	$$ \frac{|\shiftset|}{2} \|\wh{x}'_{S_x(\bv,\bb) \cup \{\ff\}}\|_2^2  \le \sum_{\alphav \in \shiftset} |W_\alphav(\bb) v e^{2\pi i \frac{\ff^T \aaa}{n}} |^2 \le \frac{3|\shiftset|}{2}\|\wh{x}_{S_x(\bv,\bb) \cup \{\ff\} }\|_2^2$$
	thus, the one sparse test in line~8 of Algorithm \ref{alg:resolve-buckets} works correctly for all buckets.
	
	It is straightforward to see that the runtime of this procedure is $O \left( |\shiftset| \cdot |\dom(W_\alphav)| \right)$.
\end{proof}

\subsection{Sparse FFT for signals with random support in nearly linear time}

\begin{algorithm}
	\caption{Procedure for Sparse FFT on random support signals with nearly linear sample complexity and runtime}
	\begin{algorithmic}[1]
		
		\Procedure{SparseFFT}{$x, n, d, k$} 
		\State $\Gamma \gets \Theta(1)$

		\State $\bbasev \gets \textsc{MakeBucket}({\Gamma}k, n,  d)$
		\State $\bprevv  \gets \bbasev$
		
		\State $\bnextv \gets \textsc{MakeBucket}({\Gamma}^2 k, n, d)$
		
		\State $\wh \chi^0 \gets \{0\}^{n^d}$
		\State $R \gets [\bprevv]$
		
		\State $L \gets \log_{\Gamma} k$
		
		\For{\texttt{$t=0$ to $L$}}
		
		\If{$R = \emptyset$}
		\State \textbf{return} $ \wh \chi$		
		\EndIf
		
		\State $\shiftset \gets \{\alphav_j : \alphav_j \text{ is an i.i.d. sample from  } \unif([n]^d) \text{ for all } j \in [C \log^2 N \cdot (\log \log N)^2] \}$ 
		
		\State $W_\alphav \gets \textsc{Hashing}(x, \wh\chi, n, d, \bbasev, \bprevv, \bnextv, \alphav, R) \text{ for all } \alphav\in\shiftset \cup \{e_1,...,e_d\} $
		
		\State $\left( \wh \chi' , R \right) \gets \textsc{TestBuckets} \left( \{W_\alphav\}_{\alphav\in\shiftset \cup \{e_1,...,e_d\}}, n, d, \bnextv \right)$
		\State $\wh\chi \gets \wh\chi + \wh\chi'$

		\State $\bprevv \gets \bnextv$	
		
		\State $\bbasev \gets \textsc{MakeBucket}({\Gamma}^{-t} k, n, d)$
		
		\State $\bnextv \gets \textsc{MakeBucket}({\Gamma}^{t+3} k , n ,d)$

		\EndFor

		\State \textbf{return} $ \wh \chi$
		
		\EndProcedure

		\Procedure{MakeBucket}{$k, n, d$} 
		
		\State $p \gets \lfloor \log_n k \rfloor$
		
		\State $r \gets \frac{k}{n^p}$
		
		\State $B_1, ..., B_p \gets n$
		
		\State $B_{p+1} \gets r$
		
		\State $B_{p+2}, ..., B_{d} \gets 1$
		
		\State \textbf{return} $\bv$
		
		\EndProcedure
		
	\end{algorithmic}
	\label{alg:FFThighD-k1/r}
	
\end{algorithm}

Now, we are ready to present the main theorem of this section.

\sfftalg*

The theorem is a consequence of the following lemma.
\begin{lemma}\label{lem:inductevent}
	Let $\bbasevt{t}$, $\bprevvt{t}$, $\bnextvt{t}$, $R^{(t)}$, and $\chi^{(t)}$ denote the values of $\bbasev$, $\bprevv$, $\bnextv$, $R$, and $\chi$, respectively, at the start of iteration $t$ of the main \emph{for} loop in Algorithm~\ref{alg:FFThighD-k1/r}. Then, for all  $t=0,1,\dots, L$, we define the event $\mathcal{E}_t$ to be the occurrence of the following statements:
	
	\begin{enumerate}
		\item $R^{(t)} =  \{\bb\in[\bprevvt{t}]:  S_{x - \chi^{(t)} }(\bprevvt{t},\bb) \neq \emptyset \}$.
		\item $\supp (\wh{x-\chi^{(t)}}) \subseteq \supp \wh x$ and $\supp{\wh \chi^{(t)}} \cap \supp (\wh{x-\chi^{(t)}}) = \emptyset$.
		\item If $t > 0$ then $\left| S_{x - \chi^{(t)} } \left(\bprevvt{t},\bxi \pmod{\bprevvt{t}} \right) \right| \ge 2$ for every $\bxi \in \supp {(\wh X - \wh \chi^{(t)})}$.	
	\end{enumerate}
	
	Then, $\EE_0$ holds with probability 1, while $\Pr[\mathcal{E}_t \mid \mathcal{E}_0, \mathcal{E}_1, \dots, \mathcal{E}_{t-1}] \geq 1 - \frac{1}{n^{2d}}$ for $t=1,\dots,L$.
\end{lemma}
\begin{proof}
	Note that for $t=0$, we have $R^{(t)} = \left[\bprevvt{0}\right]$. Thus, condition (1) trivially holds. Condition (2) also trivially holds, since $\wh{\chi}^{(0)} = 0$. Furthermore, (3) does not apply for event $\EE_0$. Thus, $\EE_0$ holds with probability 1.
	
	Now, assume that $\EE_0, \EE_1, \dots, \EE_m$ hold for some $m\geq 0$. We consider the probability of $\EE_{m+1}$ occurring, conditioned on the aforementioned events. Note that it follows from the values that Algorithm \ref{alg:FFThighD-k1/r} assigns to $\bbasevt{m}$, $\bprevvt{m}$ and $\bnextvt{m}$ along with condition (1) of the inductive hypothesis that Lemma~\ref{lem:hashing2-highdim} can be applied to invocations of \textsc{Hashing} $(x, \wh\chi^{(m)}, n, d, \bbasevt{m}, \bprevvt{m}, \bnextvt{m}, \alphav, R)$ in line 13 of Algorithm \ref{alg:FFThighD-k1/r}, and hence the output of \textsc{Hashing} procedure  satisfies the following,
	
	\begin{equation}
	W_{\alphav} = \left\{(\bb', U_{x-\chi^{(m)}}^\alphav (\bnextvt{m}, \bb')): \bb'\in [\bnextvt{m}] \text{ s.t. } \, \bb'\equiv\bb \pmod{\bprevvt{m}} \text{ for some } \bb\in R^{(m)} \right\}. \label{eq:w}
	\end{equation}
	
	Moreover, it is clear that for every $\alphav\in [n]^d$ and every $\bb' \in [\bnextvt{m}]$, one has that $W_\alphav(\bb')$ is $(\bnextvt{m}, \bb')$-bucketing of $x-\chi^{(m)}$. 
	
	Now, note that by condition (2) of the inductive hypothesis along with definition of Congruence class presented in Definition \ref{def:congruence-class}, it follows that $S_{x-\chi^{(m)}}(\bnextvt{m}, \bb) \subseteq S_{x}(\bnextvt{m}, \bb)$ for all $\bb \in [\bnextvt{m}]$. Hence by Lemma \ref{lem:sxsize}, with probability $1- \frac{1}{n^{3d}}$,
	$$S_{x-\chi^{(m)}}(\bnextvt{m}, \bb) = O(d\log_2n)$$
	This show that set $\shiftset$ defined in line 12 of Algorithm \ref{alg:FFThighD-k1/r} satisfies the precondition of Lemma \ref{lem:resolve-buckets-highdim}.
	Therefore by Lemma \ref{lem:resolve-buckets-highdim}, a call to $\textsc{TestBuckets}$ procedure in line 14 of Algorithm \ref{alg:FFThighD-k1/r}, with probability $1 - \frac{1}{n^{10d}}$, outputs $(\wh\chi', R^{(m+1)})$ such that the following hold:
	
	\begin{enumerate}
		\item[(a)] For any $\bb' \in \dom(W_\alphav)$ such that $S_{x - \chi^{(m)}}(\bnextvt{m},\bb')$ is a singleton set, $S_{x - \chi^{(m)}}(\bnextvt{m},\bb') = \{\ff\}$, one has $\wh\chi'_\ff = (\wh {x - \chi^{(m)}})_\ff$.
		\item[(b)] $R^{(m+1)} = \left\{\bb' \in\dom(W_\alphav): |S_{x - \chi^{(m)}}(\bnextvt{m},\bb')| \ge 2 \right\}$.
	\end{enumerate}
	In order to complete the inductive step, it suffices to show that (a) and (b) imply conditions (1), (2), and (3) in the definition of $\EE_t$ for $t = m+1$.
	
	Observe that condition (a) imply that $\supp{\wh \chi'} \subseteq \supp( \wh{x- \chi^{(m)}})$ therefore	since $\chi^{(m+1)} = \chi^{(m)} + \chi'$ (by line 15), it follows that $\supp( \wh{x- \chi^{(m+1)}})\subseteq \supp( \wh{x- \chi^{(m)}})$. Hence, by inductive hypothesis $\EE_m$ we have $\supp( \wh{x- \chi^{(m+1)}})\subseteq \supp{\wh x}$. Also note that for every $\ff \in \supp( \wh{x- \chi^{(m+1)}})$ we have that $\ff \in \supp( \wh{x- \chi^{(m)}})$ and hence by condition (2) of the inductive hypothesis $\EE_m$, one has $\wh \chi^{(m)}_\ff = 0$. Condition (a) implies that $\wh \chi'_\ff = 0$ for every $\ff \in \supp( \wh{x- \chi^{(m+1)}})$ and hence $\chi^{(m+1)}_\ff = \chi^{(m)}_\ff + \chi'_\ff = 0$ for every such every $\ff$. This establishes condition (2) of $\EE_{t}$ for $t = m+1$.
	
	Next note that $\bprevvt{m+1} = \bnextvt{m}$ (by line 16). Conditions (a) along with condition (1) of the inductive hypothesis for $\EE_m$ and \eqref{eq:w} imply that there exists no $\bb' \in [\bprevvt{m+1}]$ such that $|S_{x - \chi^{(m+1)}}(\bprevvt{m+1},\bb')| = 1$. Also note that condition (b) implies that $|S_{x - \chi^{(m+1)}}(\bprevvt{m+1},\bb')| \ge 2$ for every $\bb' \in R^{(m+1)}$, therefore $R^{(m+1)}$ satisfies condition (1) of the induction $\EE_{t}$ for $t = m+1$. This also establishes condition (3) of $\EE_{t}$ for $t = m+1$.
	
	By a union bound we have that with probability $1 - \frac{1}{n^{3d}} - \frac{1}{n^{10d}} \ge 1 - \frac{1}{n^{2d}}$, event $\EE_{m+1}$ holds true as desired. 
\end{proof}
Now we are ready to prove Theorem~\ref{thrm:sparse-fft}.
\begin{proof}[Proof of Theorem~\ref{thrm:sparse-fft}]
	Note that by Lemma~\ref{lem:inductevent}, there exist events $\EE_0, \EE_1, \dots, \EE_L$ such that $\Pr[\EE_0] = 1$ and $\Pr[\EE_t \mid \EE_0, \EE_1, \dots, \EE_{t-1}] \geq 1 - \frac{1}{n^{2d}}$ for $t=1,2,\dots, L$. Observe that
	\begin{align*}
	\Pr[\EE_L] &\geq \Pr[\EE_0, \EE_1, \dots, \EE_L]\\
	&\geq 1 - \sum_{t=0}^L \Pr[\EE_0, \dots, \EE_{t-1}] \cdot \Pr[\overline{\EE_t} \mid \EE_0, \dots, \EE_{t-1}]\\
	&\geq 1 - \sum_{t=0}^L \Pr[\overline{\EE_t} \mid \EE_0, \dots, \EE_{t-1}]\\
	&\geq 1 - \sum_{t=1}^L \frac{1}{n^{2d}}\\
	&\geq 1 - \frac{L}{n^{2d}}.
	\end{align*}
	Note that condition (3) of $\EE_L$ implies that the existence of a $\bxi \in [n]^d$ such that $\wh \chi^{(L)}_{\bxi} \neq \wh x_{\bxi}$ requires $S^{(\bprevvt{L})} \neq \emptyset$.
	Now, recall that after the main \emph{for} loop in Algorithm~\ref{alg:FFThighD-k1/r} finishes execution,we have
	\[ |\bprevv| = k \cdot \Gamma^{L+2}. \]
	Thus, by Lemma~\ref{lem:expcollision}, we have that $\E\left[S^{\left(\bprevv\right)}\right] \leq \frac{k^2}{k\cdot \Gamma^{L+2}} = \frac{k}{\Gamma^{L+2}} \leq \frac{1}{100}$, by our choice of $\Gamma$ and $L$.  Thus, by Markov's inequality, with probability $\geq \frac{99}{100}$ over the randomness in the choice of $S = \supp{\wh x}$, we have that $S^{\left(\bprevv\right)} = \emptyset$. Hence, by a union bound, we have that $\Pr \left[ \EE_L \land \left(S^{(\bprevv)} = \emptyset\right) \right] \geq \frac{9}{10}$. Thus, by condition (3) in Lemma~\ref{lem:inductevent}, we see that with probability $\geq \frac{9}{10}$, the output $\wh\chi$ of Algorithm~\ref{alg:FFThighD-k1/r} satisfies $\wh\chi_\ff = \wh x_\ff$ for all $\ff\in [n]^d$, which proves the correctness of Algorithm~\ref{alg:FFThighD-k1/r}.
	
	Now, let us compute the sample complexity of Algorithm~\ref{alg:FFThighD-k1/r}. Note that for each iteration $t$ of the main \emph{for} loop in \textsc{SparseFFT}, we have $\frac{|\bnextv|}{|\bprevv|} = \Gamma$. Also, By condition (1) of $\EE_t$ in Lemma \ref{lem:inductevent},
	$$R^{(t)} =  \{\bb\in[\bprevvt{t}]:  S_{x - \chi^{(t)} }(\bprevvt{t},\bb) \neq \emptyset \}$$
	Therefore, since $|\bprevv| \cdot |\bbasev| \ge k^2$, by Lemma~\ref{lem:liftsize}, we have that
	\begin{align*}
	&\max_{\bb\in\bbasevt{t}} \left| \left\{ \rr\in R^{(t)}: \rr \equiv \bb \pmod{\bbasevt{t}} \right\} \right| \\
	& \qquad = \max_{\bb\in\bbasevt{t}} \left| S^{(\bvec)} \cap \{\ff\in [\bvec]: \ff \equiv \bb \pmod{\bvec'}\} \right| \\
	&\qquad = O(\log N)
	\end{align*}
	with probability $\geq 1 - \frac{1}{N^3}$.
	
	Moreover, $|\bbasevt{t}| = O\left(\frac{k}{\Gamma^t}\right)$. Therefore, by Lemma~\ref{lem:hashing2-highdim}, each call to \textsc{Hashing} in the $t$-th iteration has sample complexity
	\[
	O\left(\Gamma \frac{k}{\Gamma^t} (\log^2 N) (\log \log N)^2 \right).
	\]
	Hence, because in each iteration \textsc{Hashing} is invoked $O((\log^2 N) (\log \log N)^2)$ times, the total sample complexity for the algorithm is
	\begin{equation*}
	O\left( \sum_{t=0}^{L-1} \Gamma \cdot \frac{k}{\Gamma^t} \log^4 N (\log \log N)^4 \right) = O\left( k (\log^4 N) (\log \log N)^4 \right),
	\end{equation*}
	since $\Gamma = O(1)$.
	
	Finally, we compute the time complexity of Algorithm~\ref{alg:FFThighD-k1/r}. By Lemma~\ref{lem:hashing2-highdim}, we have that the time complexity for each call to \textsc{Hashing} in the $t$-th iteration of the main \emph{for} loop is
	\[
	O\left( \left(\Gamma (\log^2 N) (\log \log N)^2 \right)^3 \cdot \frac{k}{\Gamma^t} + \left(\Gamma (\log^2 N) (\log \log N)^2  \right) \cdot \left(\frac{k}{\Gamma^t} \cdot \log k + k\right)\right).
	\]
	Thus, the total time complexity due to calls to \textsc{Hashing} is
	\begin{align*}
	O\left(\log^2 N (\log \log N)^2 \sum_{t=0}^{L-1} \left(\left(\Gamma (\log^2 N) (\log \log N)^2 \right)^3 \cdot \frac{k}{\Gamma^t} + \left(\Gamma (\log^2 N) (\log \log N)^2  \right) \cdot \left(\frac{k}{\Gamma^t}   \log k + k\right)\right)\right),
	\end{align*}
	which can be simplified as
	\begin{equation}
	O\left(k (\log^8 N) (\log \log N)^8\right), \label{eq:timefromhashing}
	\end{equation}
	since $\Gamma = O(1)$.
	Moreover, the call to \textsc{TestBucket} in the $t$-th iteration of the main \emph{for} loop has time complexity
	\[
	O\left(|\shiftset| \cdot \max_{\alphav\in\shiftset} |\dom(W_\alphav)|\right) = O\left(\frac{k}{\Gamma^t} \log^2 N \cdot (\log \log N)^2\right).
	\]
	Hence, the total time complexity due to calls to \textsc{TestBucket} is
	\begin{equation}
	O\left(\sum_{t=0}^{L-1} \frac{k}{\Gamma^t} \log^2 N \cdot (\log \log N)^2 \right) = O\left(k \log^2 N \cdot (\log \log N)^2\right). \label{eq:timefromresolvebucket}
	\end{equation}
	Therefore, by \eqref{eq:timefromhashing} and \eqref{eq:timefromresolvebucket}, the total time complexity of Algorithm~\ref{alg:FFThighD-k1/r} is
	\[
	O\left(k (\log^8 N)(\log \log N)^8\right),
	\]
	as desired.
\end{proof}

\subsection*{Acknowledgements}\label{sec:ack}
Michael Kapralov is supported in part by ERC Starting Grant 759471.

\newcommand{\etalchar}[1]{$^{#1}$}

\begin{appendix}

\section{Proofs and pseudocode omitted from section~\ref{sec:filters}} \label{appx:A}

\begin{algorithm}
	\caption{Splitting tree construction in time $O(|S| \log n)$}\label{alg:tree-construction}
	\begin{algorithmic}[1]
		
		\Procedure{Tree}{$S, n$} 

		\State $\mathcal{C}_{0} \gets \{(r, S)\}$
		
		\State Let $T$ be a tree with one node, labeled $f_{r} = 0$
		
		\For{\texttt{$j =1$ to $ \log_2 n$ }}
		
		\State $\mathcal{C}_j \gets \emptyset$
		
		\For{all $(v , S_v) \in \mathcal{C}_{j-1}$}\Comment{$\mathcal{C}_{j-1}$ contains nodes at level $j-1$ and their corresponding set of frequencies}
		
		\State $R \gets \{ g \in S_v : g = f_v \pmod {2^j}\}$
		\State $L \gets \{ g \in S_v: g = f_v+2^{j-1} \pmod {2^j}\}$
		
		\If{$R \neq \emptyset$}
		\State Add $u$ as right child of $v$ in $T$
		\State $\mathcal{C}_j \gets \mathcal{C}_j \cup \{ (f_v, R) \}$
		\EndIf
		\If{$L \neq \emptyset$}
		\State Add  $w$ as left child of $u$ in $T$
		\State $\mathcal{C}_j \gets \mathcal{C}_j \cup \{ (f_v+2^{j-1}, L)\}$
		\EndIf
		\EndFor
		\EndFor
		\State \textbf{return} $T$
		\EndProcedure

		\Procedure{Tree.remove}{$T, v$} 
		
		\State $r\gets$ root of $T$, $l\gets l_T(v)$
		\State $v_0,v_1,\ldots, v_l\gets $ path from $r$ to $v$ in $T$, where $v_0 = r$ and $v_l = v$
		
		\State $q \gets $ largest integer $j \le l$ such that $v_j$ has two children
		\State Remove $v_{q+1}, ..., v_{l}$ and their connecting edges from $T$
		
		\State \textbf{return} $T$
		\EndProcedure
		
	\end{algorithmic}
	\label{alg:tree-const}
	
\end{algorithm}

\begin{proofof}{Lemma~\ref{lem:isolate-filter-highdim}}
	Let $v$ be a leaf of $T$, let $l=l_T(v)$ denote the level of $v$, let $r$ denote the root of $T$, and let $v_0, v_1,\ldots, v_l$ denote the path from root to $v$ in $T$, where $v_0 = r$ and $v_l = v$. Let $q^*$ denote the smallest positive integer such that $l \le q^* \cdot \log_2n$. Note that $q^* \le d$.
	
	For $q\in \{0, 1, \ldots, d\}$ let ${T}^{(q)}$ be a subtree of $\tfull_{N}$ which denotes the result of truncating $T$ to contain only the nodes that are at distance at most $q\log_2 n$ from the root.

	We construct the $(v, {T})$-isolating filter $\wh{G}$ iteratively by starting with $\wh{G}^{(0)}=1$ and refining $\wh{G}^{(q-1)}$ to $\wh{G}^{(q)}$ over $q^*$ steps. The filters $\wh{G}^{(q)}$ will be $(v_{q\cdot \log_2 n}, {T}^{(q)})$-isolating for $q=0,1,\ldots, q^*-1 $ and $\wh{G}^{(q^*)}$ will be $(v_{l}, {T}^{(q^*)})$-isolating. Since ${T}^{(q^*)}={T}$ and $v_{l}=v$, the filter $\wh{G}^{(q^*)}$ will be $(v, {T})$-isolating, as required. 
	
	For every $q \in \{ 1,...,q^* \}$ let $T_q^v$ be the subtree of $T$ which is rooted at $v_{(q-1)\cdot \log_2n}$ and is restricted to contain only the nodes that are at distance at most $\log_2n$ from $v_{(q-1)\cdot \log_2n}$. For every node $u \in T_q^v$ the label of $u$ is defined to be $f_u = (\ff_u)_q$, i.e., the $q$th coordinate of $\ff_u$, where $\ff_u$ is the label of node $u$ in tree $T$.
	
	We now define $\wh{G}^{(q)}$ for $q=1,\ldots, q^*$. We start by letting $\wh{G}^{(0)}=1$ and letting for every $\ff=(f_1,\ldots, f_q, \ldots, f_d)\in [n]^d$
	\begin{equation}\label{eq:g-q-def}
	\wh{G}^{(q)}(\ff)=\wh{G}^{(q-1)}(\ff)\cdot \wh{G}_q(f_q).
	\end{equation}
	where $\wh G_q$ is a $(v_{q \cdot \log_2n}, T_q^v)$-isolating filter for all $q=1,..., q^*-1$ and $\wh G_{q^*}$ is a $(v_l, T_{q^*}^v)$-isolating filter. By lemma \ref{lem:filter-isolate}, for every $q=1,\ldots, q^*$ there exists such $ G_q$ with $|\supp{G_q}| = 2^{w_{T_q^v}(v_{q \cdot \log_2n})}$ and can be constructed in time $O(2^{w_{T_q^v}(v_{q \cdot \log_2n})} + \log_2n)$. Such a filter can be computed in Fourier domain at any desired frequency in time $O(\log_2n)$.
	Note that $\wh{G}^{(q)}$ is a tensor product of $q$ filters in dimension one. We now show by induction on $q$ that $\wh{G}^{(q)}$ is a $(v_{q\cdot \log_2 n}, {T}^{(q)})$-isolating filter.

	The {\bf base} of the induction is provided by $q=0$: since $v_0$ is the root of ${T}^{(0)}$, we have that $\subtree_{{T}^{(0)}}(v_0)=[n]^d$ and $\wh{G}^{(0)}\equiv 1$ as required. 
	
	We now prove the {\bf inductive step}: $q-1\to q$. We first show that $\wh{G}^{(q)}_{\ff'}=0$ for every $\ff'\in \bigcup_{\substack{u \neq v_{q\cdot \log_2n} \\ u: \text{~leaf of~}T^{(q)}}} \subtree_{T^{(q)}}(u)$. Let $u$ be a leaf of $T^{(q)}$ distinct from $v_{q\cdot \log_2n}$. Let $u'$ denote the leaf of ${T}^{(q-1)}$ which is the ancestor of $u$. We consider two cases. 
	
	\begin{description}
		\item[Case 1: $\ff'\not \in \subtree_{{T}^{(q-1)}}(v_{(q-1)\log_2 n})$] Suppose that $u' \neq v_{(q-1)\cdot \log_2n}$. Note that $l_{T}(u') \le (q-1)\log_2 n$, and also note that 
		$$
		\subtree_{{T}^{(q)}}(u)\subseteq \subtree_{{T}^{(q-1)}}(u'),
		$$ 
		Thus for every $\ff' \in \subtree_{T^{(q)}}(u)$ it is true that $\ff'\in \subtree_{{T}^{q-1}}(u')$. By the inductive hypothesis we have that $\wh{G}^{(q-1)}$ is $(v_{(q-1)\log_2 n}, {T}^{(q-1)})$-isolating, and hence by the assumption of $u' \neq v_{(q-1)\cdot \log_2n}$, one has $\wh{G}^{(q-1)}(\ff')=0$ for every such $\ff'$, and thus $\wh{G}^{(q)}(\ff')=\wh{G}^{(q-1)}(\ff')\cdot \wh{G}_q(f'_q)=0$ as required.

		\item[Case 2: $\ff' \in \subtree_{{T}^{(q-1)}}(v_{(q-1)\log_2 n})$] Suppose that $ v_{(q-1)\cdot \log_2n}$ is ancestor of $u$. Therefore, by definition of $T_q^v$, one can see that $u$ is a leaf in $T_q^v$. Hence, by definition of $T^v_q$, for every $\ff' \in \subtree_{T^{(q)}}(u)$, it is true that $f'_q \in \subtree_{T_q^v}(u)$. Recall that $\wh{G}_q$ is a $(v_{q\cdot \log_2n},T_q^v)$-isolating filter and therefore, $\wh{G}_q(f'_q)=0$, and thus $\wh{G}^{(q)}(\ff')=\wh{G}^{(q-1)}(\ff')\cdot \wh{G}_q(f'_q)=0$ as required. 
		
	\end{description}

	Now we show that $\wh{G}^{(q)}_{\ff}=1$ for all $\ff\in \subtree_{T^{(q)}}(v_{q\cdot \log_2n})$. Note that $v_{q\cdot \log_2n}$ is a leaf in $T_q^v$. Hence, for every $\ff \in \subtree_{T^{(q)}}(v_{q\cdot \log_2n})$, it is true that $f_q \in \subtree_{T_q^v}(v_{q\cdot \log_2n})$. Since $\wh{G}_q$ is a $(v_{q\cdot \log_2n},T_q^v)$-isolating filter, $\wh{G}_q(f_q)=1$. Now, note that 
	$$
	\subtree_{{T}^{(q)}}(v_{q\cdot \log_2n})\subseteq \subtree_{{T}^{(q-1)}}(v_{(q-1)\cdot \log_2n}),
	$$ 
	Thus for every $\ff \in \subtree_{T^{(q)}}(v_{q\cdot \log_2n})$ it is true that $\ff\in \subtree_{{T}^{(q-1)}}(v_{(q-1)\cdot \log_2n})$. By the inductive hypothesis we have that $\wh{G}^{(q-1)}$ is $(v_{(q-1)\log_2 n}, {T}^{(q-1)})$-isolating, and hence
	$\wh{G}^{(q-1)}(\ff)=1$, and thus $\wh{G}^{(q)}(\ff)=\wh{G}^{(q-1)}(\ff)\cdot \wh{G}_q(f_q)=1$ as required. 
	
	It remains to note that $w_{{T}}({v})=\sum_{q=1}^{q^*} w_{T_q^v}({v_{q\cdot \log_2n}})$. By Lemma~\ref{lem:filter-isolate}, for every $q \in \{1,...,q^*\}$ one has $|\supp G_q| = 2^{w_{T_q^v}(v_{q\cdot \log_2n})}$, so $|\supp G| =2^{w_{T}(v)}$, as required (note that the support size of the convolution of two filters is at most the product of support sizes of each filter).
	
	The total runtime for constructing this filter has two parts; First part is the computation time of $G_q$'s for all $q \in \{ 1,..., q^*\}$ which takes $\sum_{q=1}^{q^*} O\left( 2^{w_{T_q^v}({v_{q\cdot \log_2n}})}  + \log_2n \right) = O\left( 2^{w_T(v)} +d\log_2n \right)$ by Lemma \ref{lem:filter-isolate}. Second part is the time needed for computing the tensor product of all $G_q$'s which is $O\left( \|G_1\|_0 \cdot ... \cdot \|G_{q^*}\|_0 \right) = O(2^{w_T(v)})$. Therefore the total runtime is $O\left( 2^{w_T(v)} +d\log_2n \right)$. Moreover, the total time for computing $\wh G(\bm\xi)$ is the sum of the times needed for computing all $\wh G_q(\xi_q)$'s for $q=1, \cdots, q^*$, which is $O(d\log_2n)$ by Lemma \ref{lem:filter-isolate}. 	
	
\end{proofof}

\begin{proofof}{Lemma \ref{lem:random-sign}}
	Let $N = n^d$. Recall that for every $\tv \in [n]^d$,
	$$x_\tv = \frac{1}{N} \sum_{\ff \in [n]^d} \wh x_\ff \cdot e^{2\pi i \frac{\ff^T \tv}{n}}$$
	Because all $\wh x_\ff$'s are zero mean independent random variables, for every fixed $\tv \in [n]^d$ one has that for every $\ff \in [n]^d$ the random variables $\wh x_\ff \cdot e^{2\pi i \frac{\ff^T \tv}{n}}$ are zero mean and independent. Observe that for all $\ff \in [n]^d$, we have $|\wh x_\ff \cdot e^{2\pi i \frac{\ff^T \tv}{n}}| = |\beta_\ff| \le \|\beta\|_\infty$ and also $\E\left[ |\wh x_\ff \cdot e^{2\pi i \frac{\ff^T \tv}{n}}|^2 \right] = |\beta_\ff|^2$. Therefore by Bernstein's inequality we have that for every fixed $\tv \in [n]^d$,
	
	\begin{align*}
	\Pr\left[\left| x_\tv\right| > \frac{\theta}{N} \right] &\leq 2 \exp\left(- \frac{\frac{1}{2}\theta^2}{\|\beta\|_2^2 + \frac{1}{3} \|\beta\|_\infty \cdot \theta} \right)\\
	&\leq 2 \exp\left(- \frac{\frac{1}{2}\theta^2}{\|\beta\|_2^2 + \frac{1}{3} \|\beta\|_2 \cdot \theta} \right)
	\end{align*}
	If we choose $\theta = C_1 \log_2 N \cdot \|\beta\|_2$ for some absolute constant $C_1 > 0$,
	\begin{align*}
	\Pr\left[\left| x_\tv\right| > \frac{C_1 \log_2 N \cdot \|\beta\|_2}{N} \right] &\leq 2 \exp\left(- \frac{\frac{1}{2} C_1^2 \log^2_2 N}{ 1 + \frac{1}{3} C_1 \log_2 N} \right) \\
	&\le \frac{1}{2N^5}
	\end{align*}
	for large enough constant $C_1$.
	By a union bound over all $\tv \in [n]^d$ we get that, $| x_\tv|^2 \le \frac{C_1^2 \log^2_2 N }{N^2} \|\beta\|^2_2$ for all $\tv \in [n]^d$ with probability $1 - \frac{1}{2N^4}$.
	
	Now note that by Parseval's theorem, Claim \ref{parseval},
	\[
	\sum_{\jj \in [n]^d} |x_\jj|^2 = \frac{1}{N} \sum_{\ff\in [n]^d} |\beta_\ff|^2.
	\]
	Conditioning on $\|x\|_\infty \le \frac{C_1^2 \log^2_2 N }{N^2} \|\beta\|^2_2$, by Chernoff-Hoeffding Bound we have,
	
	\begin{align*}
	\Pr\left[\frac{1}{2} \cdot \frac{\|\beta\|_2^2}{N^{2}} \leq \frac{1}{s} \sum_{j=1}^s |x_{\tv_j}|^2	\le \frac{3}{2} \cdot \frac{\|\beta\|_2^2}{N^{2}} \right] &\ge 1 - 2e^{- \frac{C_2 \cdot s \cdot \|\beta \|_2^2}{N^2 \|x\|_\infty}}\\
	&\ge 1 - 2 e^{- \frac{C_2 \cdot s }{C_1^2 \log^2_2 N}}
	\end{align*}
	where the probability is over the i.i.d. random variables $\tv_1, \tv_2, \dots, \tv_s \sim \unif([n]^d)$ and $C_2$ is some positive constant. Therefore, by the choice of $s = C \log_2^3 N$ for some large enough constant $C$ we have that,
	$$\Pr\left[\frac{1}{2} \cdot \frac{\|\beta\|_2^2}{N^{2}} \leq \frac{1}{s} \sum_{j=1}^s |x_{\tv_j}|^2	\le \frac{3}{2} \cdot \frac{\|\beta\|_2^2}{N^{2}} \right] \ge 1 - \frac{1}{2N^4}.$$
	By a union bound over these two events we have that $\frac{1}{2} \cdot \frac{\|\beta\|_2^2}{N^{2}} \leq \frac{1}{s} \sum_{j=1}^s |x_{\tv_j}|^2	\le \frac{3}{2} \cdot \frac{\|\beta\|_2^2}{N^{2}}$ with probability at least $1 - \frac{1}{N^4}$.

\end{proofof}

\section{Proof of Lemma~\ref{lem:liftsize}}\label{app:bernoulli}

\begin{lemma}\label{lem:sxsize}
	For every power of two integer $n$ and positive integer $d$, if $x \in \C^{n^d}$ is a random support signal as per Definition~\ref{def:randsupport}, the following conditions hold. If $N=n^d$, $\bvec = (B_1, B_2, \dots, B_d)$  is a vector of powers of two such that $B_j \mid n$ for all $j = 1,2, \dots, d$ and $|\bvec| \geq 4k$, then, with probability at least $1 - \frac{1}{N^3}$ over $x$,
	\[
	|S_x(\bvec, \bb)| = O(\log N)
	\]
	for all $\bb \in [\bv]$ (where $S_x(\bvec, \bb)$ is the set from Definition~\ref{def:congruence-class}).
\end{lemma}
\begin{proof}
	Note that for every $\bb \in [\bvec]$, we have that for each $\ff \in [n]^d$ with $\ff \equiv \bb \pmod{[\bvec]}$, $\Pr[\ff\in S] = k/N$. Then, since there are $N/|\bvec|$ such $\ff$ for every fixed $\bb \in [\bvec]$, it follows that,
	\[
	\E\left[ \left| S_x(\bvec, \bb)\right| \right] \leq \frac{k}{N} \cdot \frac{N}{|\bvec|} \leq \frac{1}{4}.
	\]
	Hence, by the Chernoff bound, it follows that $|S_x(\bvec, \bb)| = O(\log N)$ with probability $1 - N^{-4}$ for any fixed $\bb\in [\bvec]$. Finally, by a union bound over all $\bb\in[\bvec]$, we have the desired result.
\end{proof} 

We now prove a lemma about the size of the sets $S^{(\bvec)}$:
\begin{lemma} \label{lem:sjprobhighdim}
	For any power of two integers $n$ and $k$, positive integer $d$ and $\bvec = (B_1, B_2, \dots, B_d)$ such that $B_1, B_2, \dots, B_d \mid n$ and $\ff = (f_1, \dots, f_d)\in [\bvec]$, we have that $\Pr[{\bf f}\in S^{(\bvec)}] \leq \left(\frac{k}{|\bv|}\right)^2$, where $S^{(\bvec)}$ is defined as in Definition \ref{bernoulli-set}.
\end{lemma}
\begin{proof}
	Suppose $\ff \in [\bvec]$. Then, observe that there are $\left(\frac{n}{B_1}\right)\left(\frac{n}{B_2}\right) \cdot \left(\frac{n}{B_d}\right) = \frac{N}{B_1 B_2 \cdots B_d}$ elements $\gg = (g_1, \dots, g_d) \in [n]^d$ such that $\ff \equiv \gg \pmod{\bvec}$. Note that $\ff \in S^{(\bvec)}$ if at least two of these elements lies in $S$. Thus, for every $\ff \in [\bv]$ we have,
	\begin{align*}
	\Pr[\ff \in S^{(\bvec)}] &= 1 - \left(1 - \frac{k}{N}\right)^{\frac{N}{|\bv|}} - \frac{N}{|\bv|} \cdot \frac{k}{N} \left(1 - \frac{k}{N}\right)^{\frac{N}{|\bv|}-1}\\
	&\leq 1 - \left(1-\frac{k}{N}\right)^{\frac{N}{|\bv|}-1} \left(1 - \frac{k}{N} + \frac{k}{|\bv|}\right)\\
	&\leq 1 - \left( 1 - \frac{k}{|\bv|} + \frac{k}{N} \right) \left(1 - \frac{k}{N}+\frac{k}{|\bv|}\right)\\
	&= 1 - \left(1 - \left(\frac{k}{|\bv|} - \frac{k}{N}\right)^2 \right)\\
	&= \left(\frac{k}{|\bv|} - \frac{k}{N}\right)^2\\
	&\leq \left(\frac{k}{|\bv|}\right)^2,
	\end{align*}
	since $|\bv| \leq n^d = N$.
\end{proof}

As a consequence, we have a bound on the expected size of $S^{(\bvec)}$.
\begin{lemma}\label{lem:expcollision}
	For any power of two integers $n$ and $k$, any $\bvec = (B_1, B_2, \dots, B_d)$ such that $B_1, B_2, \dots, B_d \mid n$, we have $\E\left[|S^{(\bvec)}|\right] \leq \frac{k^2}{|\bv|}$.
\end{lemma}
\begin{proof}
	Simply note that
	\begin{align*}
	\E\left[|S^{(\bvec)}|\right] &= \sum_{\ff\in [\bvec]} \Pr[\ff\in S^{(\bvec)}]\\
	&\leq |\bv| \cdot \left(\frac{k}{|\bv|}\right)^2\\
	&= \frac{k^2}{|\bv|},
	\end{align*}
	by Lemma~\ref{lem:sjprobhighdim}.
\end{proof}

We are now ready to proof Lemma~\ref{lem:liftsize}.

\begin{proofof}{Lemma~\ref{lem:liftsize}}
	Consider a fixed $\bb\in[\bvec']$. Note that there are $m = \frac{B_1 B_2 \cdots B_d}{B_1' B_2' \cdots B_d'}$ values of $\ff\in [\bvec]$ such that $\ff\equiv\bb\pmod{\bvec'}$. Moreover, by Lemma~\ref{lem:sjprobhighdim}, each such $\ff$ lies in $S^{(\bvec)}$ with identical probability
	\[
	p \leq \left(\frac{k}{|\bv|}\right)^2,
	\]
	and these events are all independent. Thus, 
	\begin{align*}
	\E\left[\left| S^{(\bvec)} \cap \{\ff\in [\bvec]: \ff \equiv \bb \pmod{\bvec'}\} \right| \right] &\leq mp \\
	&\leq \frac{k^2}{|\bv| \cdot |\bv'|}\\
	&\leq 1.
	\end{align*}
	Thus, by the Chernoff bound, we have that
	\[
	\left| S^{(\bvec)} \cap \{\ff\in [\bvec]: \ff \equiv \bb \pmod{\bvec'}\} \right| = O(\log N) 
	\]
	with probability at least $1 - \frac{1}{N^4}$, as desired. Finally, taking a union bound over all $|\bv'| \leq N$ values of $\bb \in [\bvec]$ gives the desired result.
\end{proofof}

\end{appendix}

\end{document}